\pdfoutput=1
\newif\iffull
\fulltrue
\newcommand{\full}[2]{\iffull#1\else#2\fi}
\newif\ifanon
\iffull\anonfalse\else\anontrue\fi
\iffull
\documentclass[acmsmall,screen,nonacm]{acmart}
\else
\documentclass[acmsmall,screen,review,anonymous,nonacm]{acmart}
\makeatletter\@ACM@reviewfalse\makeatother
\fi
\ifanon
\newcommand{\repo}{\url{https://anonymous.4open.science/r/GLP-02A3}}
\else
\newcommand{\repo}{\url{https://github.com/EShapiro2/GLP}}
\fi

\usepackage{textcomp}
\usepackage{xspace}
\usepackage{array}
\usepackage[most]{tcolorbox}
\usepackage{enumitem}
\usepackage{thmtools}
\usepackage{tikz}
\usetikzlibrary{arrows.meta,positioning,fit}

\newcommand{\mypara}[1]{\par\addvspace{\smallskipamount}\noindent\textbf{#1.}}
\newcommand{\temph}[1]{\emph{#1}}
\newcommand{\remove}[1]{}
\newcommand{\udi}[1]{\textcolor{blue}{[Udi: #1]}}
\newcommand{\idit}[1]{\textcolor{teal}{[Idit: #1]}}

\ifanon\renewcommand{\udi}[1]{\ifhmode\unskip\fi}\renewcommand{\idit}[1]{\ifhmode\unskip\fi}\fi

\newcommand{\GLP}{\textsc{GLP}\xspace}

\newcommand{\calV}{\mathcal{V}}
\newcommand{\calG}{\mathcal{G}}

\newcommand{\calF}{\mathcal{F}}
\newcommand{\calN}{\mathbb{N}}
\newcommand{\calA}{\mathcal{A}}

\newcommand{\calQ}{\mathcal{Q}}
\newcommand{\calP}{\mathcal{P}}

\newcommand{\ia}{\textit{i}}
\newcommand{\ib}{\textit{ii}}
\newcommand{\ic}{\textit{iii}}
\newcommand{\id}{\textit{iv}}
\newcommand{\iv}{\textit{v}}

\newcounter{pc}

\setlist{nosep, leftmargin=*}
\setlist{itemsep=1pt, topsep=3pt, leftmargin=*}

\AtEndPreamble{\theoremstyle{acmdefinition}\newtheorem{remark}[theorem]{Remark}}

\title[Volition Elicitation]{Volition Elicitation:  Operational Semantics for People and Their Machines\full{ (Full Version)}{}}

\author{Idit Keidar}
\affiliation{%
  \institution{Technion -- Israel Institute of Technology}
  \city{Haifa}
  \country{Israel}
}
\author{Ehud Shapiro}
\affiliation{%
  \institution{London School of Economics and Weizmann Institute of Science}
  \country{UK and Israel }
}
\iffull\authorsaddresses{}\fi

\begin{document}

\begin{abstract}
The most prevalent distributed systems today include people and their personal machines (smartphones), and computations in them are driven by people's volitions.  Volition-guarded multiagent atomic transactions were proposed as an abstract specification language for such systems, in which a transaction has a precondition on the machine states of its participants and is guarded by the volitions of some of them.  

Here we develop an abstraction cascade that derives smartphone apps from such specifications and prove it correct down to the base programming language, GLP.  The cascade consists of: (\ia) communicating volitional agents (CVA), which implement every local volition-guarded specification with simple guards, correctly, and completely for its strongly fair runs, by two-phase commit coordinated by the initiating agent; (\ib) volition-guarded GLP (vGLP), an extension of GLP with volitional procedures, whose goals put an interactively-typed question to the person, into which that implementation compiles; and (\ic) GLP, which implements vGLP by eliciting the person's answers through the widgets of their interactive types, correctly, and completely for complete widgets.  We demonstrate the approach on three grassroots platforms --- a social graph, a social network, and currencies: a compiler carries the social graph and the currencies down the cascade to GLP programs that run with scripted persons, and earlier programs of the three run as mini-apps of one smartphone app.
\end{abstract}



\maketitle

\section{Introduction}\label{sec:introduction}

\mypara{Volition-guarded multiagent atomic transactions} The most prevalent distributed systems today include people and their personal machines (smartphones). In such systems, computations are driven by people's volitions: a payment when a person wishes to pay someone, befriending when two people wish to become friends, etc. Volition-Guarded Multiagent Atomic Transactions~\cite{lewis2026volitional} were proposed as an abstract specification language for such systems.  In them  (\ia) each agent consists of a person and their machine; (\ib) a multiagent atomic transaction has a precondition on the machine states of its participants and is guarded by the volitions of some of them.  For example, the befriending of two people is guarded by both persons and unfriending by either; a payment is guarded by the payer, a swap by both parties.  An illustrative example is child-safe online befriending, from work in preparation:  it is guarded by the volitions of the two children and one parent of each, and its precondition is that the machines do not record the two children as friends and do record the two parents as friends (Appendix~\ref{app:childsafe}).  

\mypara{Grassroots platforms}  We focus on platforms that can have multiple, volitionally-forming instances that may coalesce by mutual consent, possibly forming a single global instance, referred to as \emph{grassroots platforms}~\cite{shapiro2023grassrootsBA,shapiro2025atomic,shapiro2025characterising,lewis2026volitional}.   Practically, a grassroots platform consists of people running a person-to-person smartphone app.  

\mypara{Example: minting and swapping grassroots coins}  Grassroots coins~\cite{shapiro2024gc} are units of debt that any person can issue (Section~\ref{sec:coins}).  The machine state of each agent $p$ is its holdings $c_p$, a multiset of coins, and \textcent$^k_p$ is $k$ coins issued by $p$.  Two of the transactions of grassroots coins are:
\begin{tcolorbox}[colback=gray!5!white,colframe=black!75!black,top=2pt,bottom=2pt]
\begin{enumerate}
    \item \textbf{Mint$_p(k)$}: $c'_p := c_p\cup\text{\textcent}^k_p$.  Guarded by $p$.
    \item \textbf{Swap$_{p,q}(\text{\textcent}^k_u,\text{\textcent}^{k'}_v)$}: $c'_p := c_p\setminus\text{\textcent}^k_u\cup\text{\textcent}^{k'}_v$ and $c'_q := c_q\setminus\text{\textcent}^{k'}_v\cup\text{\textcent}^k_u$, provided $\text{\textcent}^k_u\subseteq c_p$ and $\text{\textcent}^{k'}_v\subseteq c_q$.  Guarded by $\{p,q\}$.
\end{enumerate}
\end{tcolorbox}
\noindent Minting changes the state of one agent, and takes place only if its person wills it.  A swap changes the states of both parties at once, each giving coins and taking the other's, and takes place only if both persons will it.  The currency agent of Section~\ref{sec:vglp} implements both.  A mint is a request of its person.  A proposed swap is shown to the other party on a \emph{card}, a panel that shows the proposal and has two buttons, yes and no.

\mypara{Previous work}
Previous work~\cite{lewis2026volitional} presented volition-guarded specifications for the grassroots social graph and grassroots currencies~\cite{shapiro2024gc}, and here we specify the grassroots social network~\cite{shapiro2023gsn}.
Volition-guarded transactions thus condition machine transitions on persons' volitions.  Formally, such a transition can occur only if the will for it is part of the volitional state of the person, namely is in the person's head.  Given today's technologies, the only way for a machine to know what's in a person's head is for the person to express it; this is the fundamental reason for a UI.  
One of the contributions of this paper is to rationalise the design of a UI that can elicit volitions from the person's head, when needed.
Existing programming languages treat the person as a source of input --- a stream of clicks~\cite{elliott1997fran,carlsson1993fudgets,czaplicki2013elm,krishnaswami2011gui}, an event~\cite{plasmeijer2007itasks,steenvoorden2019tophat,harel1987statecharts,berry1992esterel}, the environment~\cite{hyland2000pcf,xia2020itrees,wegner1997interaction} --- so input drives computation, but in none does a transaction wait on the standing wills of its guards, with an else-branch taken when a guard does not will it (Section~\ref{sec:related-work}).

\mypara{The abstraction cascade} Here we develop an \emph{abstraction cascade} that derives smartphone grassroots apps from such abstract specifications, as shown in Figure~\ref{fig:cascade}, and prove it correct down to the base programming language, GLP (Section~\ref{sec:cascade}).  The passage from a legal text to a running platform goes through the cascade, from the legal text at the top~\cite{shapiro2026formalising} to an app on each person's own phone at the bottom.  This paper supplies the arrows (3)$\rightarrow$(6), each proved correct, starting from (3), volition-guarded multiagent atomic transactions~\cite{lewis2026volitional}.  Its layers (4)--(8) in Figure~\ref{fig:cascade} are:

\begin{figure}[!t]
\centering
\begin{tikzpicture}[
  layer/.style={draw=black!55, rounded corners=2pt, fill=gray!5,
                text width=\dimexpr\linewidth-1.1cm\relax, inner sep=4pt,
                align=left, font=\small},
  low/.style={layer, fill=none, font=\scriptsize},
  arrbig/.style={-{Latex[length=2.6mm,width=3.2mm]}, line width=1.6pt},
  arrsmall/.style={-{Latex[length=1.6mm,width=2mm]}, line width=1pt},
  lab/.style={font=\tiny, inner sep=1pt, midway, right=1mm},
  labs/.style={lab, right=0.8mm}]
  \node[low] (c1) {\textbf{(1) The legal text}~\cite{shapiro2026formalising}.  The clauses of a grassroots social contract, each describing an act.};
  \node[low, below=4mm of c1] (c2) {\textbf{(2) Act schemas}~\cite{shapiro2026formalising}.  A formal language for the acts those clauses provide for.};
  \node[layer, below=7.5mm of c2] (c3) {\textbf{(3) Volition-guarded multiagent atomic transactions}~\cite{lewis2026volitional} (Section~\ref{sec:foundations}).  Each agent a person and their machine; a transaction over several agents, with a precondition on their machine states and guarded by their persons' volitions.};
  \node[layer, below=5mm of c3] (c4) {\textbf{(4) Communicating volitional agents} (Section~\ref{sec:cva}).  An implementation-ready restriction of (3) in which the only non-unary transactions are agent discovery, agent introduction and message delivery.};
  \node[layer, below=5mm of c4] (c5) {\textbf{(5) Volition-guarded GLP} (Section~\ref{sec:vglp}), an extension of (6) with volitional procedures with interactive types that determine the app's user interface (Section~\ref{sec:elicitation}).};
  \node[layer, below=5mm of c5] (c6) {\textbf{(6) Grassroots Logic Programs (GLP)}~\cite{shapiro2025glp,shapiro2026implementing,shapiro2026types} (Section~\ref{sec:glp}), a polymorphically-typed multiagent concurrent logic programming language.};
  \node[low, below=6.5mm of c6] (c7) {\textbf{(7) Dart and Flutter}~\cite{dart2024}, deployable on iPhone and Android.};
  \node[low, below=4mm of c7] (c8) {\textbf{(8) The app}~\cite{shapiro2026gsg}, as it runs on each person's own phone (Section~\ref{sec:ui-primitives}).};
  \draw[arrsmall] (c1) -- node[labs] {[informal~\cite{shapiro2026formalising}]} (c2);
  \node[draw,line width=1pt,rounded corners=4pt,fit=(c3)(c4)(c5)(c6),inner sep=2.5mm] (frame) {};
  \draw[arrsmall] (c2) -- node[labs,pos=0.28] {\cite{shapiro2026formalising}} (frame.north);
  \draw[arrbig] (c3) -- node[lab] {[Section~\ref{sec:decomposition}]} (c4);
  \draw[arrbig] (c4) -- node[lab] {[Section~\ref{sec:cva-implementation}]} (c5);
  \draw[arrbig] (c5) -- node[lab] {[Section~\ref{sec:compilation}]} (c6);
  \draw[arrsmall] (c6) -- node[labs,pos=0.72] {\cite{shapiro2026implementing}} (c7);
  \draw[arrsmall] (c7) -- node[labs] {\cite{dart2024}} (c8);
\end{tikzpicture}
\caption{The abstraction cascade, from a legal text to a running platform.  The framed (3)--(6) are this paper's subject: (3) is~\cite{lewis2026volitional}, recalled in Section~\ref{sec:foundations}, (4) is Section~\ref{sec:cva}, and the arrows (3)$\rightarrow$(6) are this paper.  Each arrow carries the reference in which its mapping is defined; the first, rendering a legal text as act schemas, is the one step left to judgement.}
\label{fig:cascade}
\end{figure}

\mypara{(4) Communicating volitional agents} 
Implementing a multiagent atomic transactions-based specification must go through unary transactions, carried out by each agent.  In the case of such volition-guarded transactions, the corresponding implementation-ready model is Communicating Volitional Agents (CVA), defined in Section~\ref{sec:cva}.  Mathematically, CVA is but a restriction of the general model, in which the only non-unary transactions are agent discovery, agent introduction and message delivery.  We prove that every CVA protocol is grassroots (Proposition~\ref{prop:cva-grassroots}).
When implementing CVA on a smartphone, the basic requirement from the smartphone UI is one: to elicit, from the person operating the smartphone, the volitions the transactions require.

A volition-guarded transaction changes the machine states of its participants at once, while a communicating volitional agent changes only its own state and reaches another only by a message.  
Here we implement the first by the second, for every local volition-guarded specification with simple guards (Section~\ref{sec:decomposition}), correctly, and completely for its strongly fair runs.  We use two-phase commit coordinated by the initiating agent~\cite{gray1978notes}. 
Wound-wait~\cite{rosenkrantz1978system} orders contending attempts, and an attempt that asks a person who does not will it closes by the else-branch of a volition-guarded if-then-else transaction~\cite{lewis2026volitional} (Section~\ref{sec:ite}), taken when the person does not will its then-branch. 
The protocol is free of deadlock and livelock.  It is correct for runs in which an initiator knows its participants, no agent is kept busy without end, and a machine whose person wills a transaction again and again eventually proposes it. 
It is complete for the runs of the specification in which a transaction enabled again and again is taken again and again.  The model is asynchronous and fault-free: it has no notion of time, and no agent crashes or misbehaves. 

\mypara{(5) Volition-guarded GLP}  Volition-guarded GLP (vGLP) extends GLP with volitional procedures.  A goal of a volitional procedure puts a question to the person, and a clause that reads the answer waits for it.  A volitional procedure is declared \verb|procedure (|$T$\verb|)*|$p(T_1,\ldots,T_n)$: each goal $p(S_1,\ldots,S_n)$ of it receives an \emph{interactive variable} of the \emph{interactive type} $T$ as an added argument, and each clause \verb|(|$A$\verb|)*|$p(S_1,\ldots,S_n)$ \verb|:-| $G$ \verb"|" $B$ matches it with its interactive term $A$.  The interactive type says which positions the program writes, for the person to see, and which positions the person writes, as their answer. 

Two fragments of the currency agent (Section~\ref{sec:platforms}) preview the language, \verb|...| eliding arguments a clause carries through unchanged.  The first is the agent, whose question is the stream of the person's requests, served by an ordinary procedure, with the clause by which the person mints coins.  The second is a \emph{responder} goal, spawned per arriving swap proposal, by which the person answers it:

\begin{verbatim}
Request ::= mint(Amount)
          ; swap(Friend, GiveCoin, GiveAmount,
                 WantCoin, WantAmount)
          ; pay(Friend, Amount)
          ; redeem(Friend, WantCoin).
procedure (Stream(Request)?)*agent(...).

%% Ask the person for their requests, once, and serve them.
(Reqs)*agent(Id, ..., Outs, Holdings, Serial) :-
    serve(Reqs?, Id?, ..., Outs?, Holdings?, Serial?).

%% Mint_p(k), guarded by p: c'_p = c_p + cent^k_p.
serve([mint(Amount)|Reqs], Id, ..., Outs, Holdings,
      Serial) :-
    ground(Id?), integer(Amount?), Amount? >= 1 |
    create_coins(Id?, Amount?, Serial?, New),
    Serial1 := Serial? + Amount?,
    append_coins(Holdings?, New?, Holdings1),
    show(minted(Amount?), Outs?, Outs1),
    after_change(Reqs?, Id?, ..., Outs1?, Holdings1?,
                 Serial1?).

%% Show the person the offer, and pass their answer on.
Offer ::= offer(From, Want, Offered).
YesNo ::= yes ; no.
Card ::= card(Offer, YesNo?).
procedure (Card)*respond_swap(Offer?, SwapResponse,
    Stream(UserInMsg)).
(card(offer(From?, Want?, Offered?), Answer))*
  respond_swap(offer(From, Want, Offered), Resp?,
               Ans?) :-
    ground(From?), ground(Want?), ground(Offered?) |
    decide_swap(Answer?, From?, Want?, Offered?, Resp,
                Ans).
\end{verbatim}

The agent asks its question once, and the person answers it one request at a time.  The mint clause waits for a \verb|mint| request at the head of the stream, which supplies \verb|Amount|.  The request form stays on screen as the interface's command-line box.  The responder writes into its card the friend who proposes the swap, the coins they want and the coins they offer, and passes the person's yes or no to \verb|decide_swap| (Section~\ref{sec:vglp}).  The card has two buttons, and no is the decline.

\mypara{The operational semantics of volition-guarded GLP} 
We then define the operational semantics of vGLP.  The state of each agent has two parts: the questions its person holds, and the goals its machine is running, called its resolvent. 
A goal of a volitional procedure is first asked: the goal receives one end of its interactive variable, and the person the other.  The program's output reaches the person before the inputs inside it, and a clause that reads the person's part reduces only after the person has answered.  An answer cannot be withdrawn.  The machine cannot compel the person, so in vGLP an answer carries no liveness obligation. 
We prove that on programs with no volitional procedures the semantics coincides with GLP's (conservativity), and that a question is assigned only by its person's answer (volitional soundness).

We compile the canonical CVA platform of every effective local specification with simple guards into vGLP, one goal per agent (Section~\ref{sec:cva-implementation}), and prove the compilation correct.

\mypara{(6) Grassroots Logic Programs and (7) Dart}
For AI to derive smartphone-based grassroots platforms from such abstract specifications, Grassroots Logic Programs (GLP)~\cite{shapiro2025glp,shapiro2026implementing,shapiro2026types}, a high-level, polymorphically-typed, multiagent, concurrent logic programming language, was designed and implemented in Dart/Flutter, which is deployable on most smartphones (iPhone/Android).  So far, GLP implementations of these platforms were derived in the same way from (volition-free) multiagent atomic transactions~\cite{shapiro2025atomic}.

However, specifying the role of people in such platforms is out-of-reach for this formalism, and hence the step of designing and implementing the app's UI was an ad hoc,  ``manual'' process.  

To rectify that, we need formal foundations for the implementation of volition-guarded transactions, not only for attaining mathematical elegance and completeness, but, more pragmatically, for AI to methodically develop the volition-eliciting UI for abstractly-specified grassroots platforms.  Moreover, we wish for a single smartphone app, running the GLP engine, to support multiple, dynamically-loaded grassroots platforms written in GLP, and for that we cannot have an ad-hoc UI for each new platform.

As the semantics of vGLP requires the person's answer before a reduction that reads it, a correct implementation of vGLP in GLP must elicit it: finding out ``what's in the person's head'' is the rationale for the UI.  We derive the UI from the semantics: an open question is presented as the widget of its interactive type, which shows what the program writes for the person and forms the person's input into a term of the type.  Widgets are associated with types by declarations, scoped as type declarations are, and every interactive type has a default widget built from its constructors.  We prove that a construct offers only answers of its question's type, and every such answer where its widget is complete (elicitation completeness).  The implementation compiles vGLP onto GLP, and we prove the implementation of vGLP by GLP correct, and complete for complete widgets.

\mypara{(8) The app}  We demonstrate the approach on the grassroots social graph, grassroots social network, and grassroots currencies, starting from their volition-guarded multiagent atomic transactions specification and going down the cascade: the CVA implementation of each platform is its canonical platform, and its vGLP program is the compilation of that platform.  A compiler from act schemas to the canonical platform compiles the social graph and the currencies, with coins of one unit, and the compiled programs run on GLP with scripted persons.  The three platforms run as mini-apps of one app on a physical smartphone, the grassroots super-app~\cite{shapiro2026gsg}, from programs written earlier, in an earlier syntax of vGLP and compiled by an earlier compilation into an agent and a mediator (Section~\ref{sec:ui-primitives}).

\mypara{Paper outline}  Section~\ref{sec:foundations} recalls volition-guarded multiagent atomic transactions, volition-guarded if-then-else transactions, transition systems and implementations; grassroots protocols are recalled in Appendix~\ref{app:grassroots}.  Section~\ref{sec:cva} defines CVA, proves every CVA protocol volitionally grassroots, and implements every local volition-guarded specification with simple guards by CVA, by two-phase commit coordinated by the initiating agent.  Section~\ref{sec:glp} recalls GLP: its syntax, its reduction and its multiagent semantics.  Section~\ref{sec:vglp} defines vGLP as its extension and the operational semantics vmaGLP, with volitional soundness, then compiles the canonical CVA platform of every effective local specification with simple guards into vGLP, and proves the compilation correct.  Section~\ref{sec:elicitation} derives the elicitation constructs from interactive types, presents widgets and their declarations, proves elicitation completeness, defines the compilation of vGLP onto GLP, and proves the implementation of vGLP by GLP correct, and complete for complete widgets; it ends by composing the implementations of the cascade, from an effective local specification to GLP.  Section~\ref{sec:platforms} presents the three platforms as volition-guarded transactions, whose CVA implementations are their canonical platforms and whose vGLP programs are the compilations of those platforms.  Section~\ref{sec:ui-primitives} shows the three platforms running as mini-apps of one smartphone app. 
Section~\ref{sec:related-work} reviews related work, and Section~\ref{sec:conclusion} concludes.

\section{Volition-Guarded Multiagent Atomic Transactions}\label{sec:foundations}

This section recalls the mathematical foundations on which the paper is based~\cite{shapiro2021multiagent,lewis2026volitional}.
Each agent consists of a \emph{person} and a \emph{machine} (e.g., smartphone) operated by the person.  A \emph{machine transaction} specifies what machines do; a \emph{volition guard} specifies which persons must be willing for it to occur.

We assume a potentially infinite set of \emph{agents} $\Pi$, but consider only finite subsets of it, so we assume any particular set of agents $P \subset \Pi$ to be nonempty and finite.  We use $\subset$ to denote the strict subset relation and $\subseteq$ when equality is also possible.  As standard, we use $S^P$ to denote the set of all total functions from $P$ to $S$, and if $c\in S^P$ we use $c_p$ to denote the value of $c$ at $p\in P$.

\iffull
\subsection{Volition-Guarded Transactions}

\begin{definition}[Machine State, Configuration, Transaction, Volition-Guarded Multiagent Atomic Transaction]\label{def:mt}
Given an arbitrary set $S$ of \temph{machine states}, with a designated \temph{initial state} $s_0 \in S$, and agents $Q \subset \Pi$, a \temph{machine configuration} over $Q$ is a member of $S^Q$, and a \temph{machine transaction} over \temph{participants} $Q$ is a pair $c\rightarrow c' \in (S^Q)^2$ such that $c\ne c'$. Given such a machine transaction $t$, a \temph{volition-guarded multiagent atomic transaction} over $t$ --- henceforth, \temph{volition-guarded transaction} --- is a pair $(t,Q')$ where $Q'\subseteq Q$ are its \temph{guards}.
\end{definition}

Here a guard is a person: the guards of a transaction are the participants who must will it, and a transaction with no guards is \emph{unguarded}.  The condition on the machine states, that the participants be in the states $c$ of $c\rightarrow c'$, its \emph{source}, is the transaction's \emph{precondition}.  Machine transactions are atomic and asynchronous~\cite{shapiro2021multiagent}---they can be carried out by their participants at any time, regardless of the states of non-participants.  Participants include both active agents (whose state changes) and stationary agents (whose state is a precondition but does not change).  Volition-guarded transactions are machine transactions that can be carried out only if their guards $Q'\subseteq Q$ are willing.  When we say a transaction is ``guarded by $\{p,q\}$,'' both must be willing; when we say it is ``guarded by either $p$ or $q$,'' we mean there are two volition-guarded transactions over the same machine transaction, $(t,\{p\})$ and $(t,\{q\})$, so that either person's volition suffices. 
Distinct machine transactions can represent ``the same action'' in different configurations, and we identify them by an equivalence relation on machine transactions.
\begin{definition}[Transaction Equivalence]\label{def:equivalence}
Given a set of machine transactions $R$, a \temph{transaction equivalence} is an equivalence relation $\sim$ on $R$ such that $t \sim t'$ implies $t$ and $t'$ have the same participants, and the set $R/{\sim}$ of its classes is countable.  We write $[t]$ for the equivalence class of $t$ under $\sim$.
\end{definition}

A person's volitional state is a set of equivalence classes of machine transactions they are willing their machine to participate in.  The will names the equivalence class, not a label on a transition: one action is a set of transitions, one per configuration it may occur in, and the will is fulfilled by whichever of them is taken.

\begin{definition}[Agent State and Configuration]\label{def:agent-state}
Given agents $P$, states $S$ with initial state $s_0$, a set of machine transactions $T$ each over its own participants $Q\subseteq P$ and $S$, and equivalence $\sim$ on $T$, an \temph{agent state} is a pair $(V,m)\in \calA = (2^{T/\sim}  \times S)$ where  $V$ is its \temph{volitional state} and $m \in S$ its  \temph{machine state}.  The \temph{initial agent state} is $(\emptyset,s_0)$.  An \temph{agent configuration} $c$ over $P$, $S$, $T$, and $\sim$ is a member $c\in \calA^P$ in which $c^v_p \subseteq (T/{\sim})_p$ for every $p\in P$, where $(T/{\sim})_p$ denotes the classes in $T/{\sim}$ in which $p$ is a participant; we write $c^v_p$ for the volitional state and $c^m_p$ for the machine state of agent $p$ in $c$.
\end{definition}

\begin{definition}[Volitional Transaction~\cite{lewis2026volitional}]\label{def:vmat}
Given agents $P$, states $S$, machine transactions $T$ over $P$ and $S$, and equivalence $\sim$ on $T$:
\begin{enumerate}
    \item A \temph{volition change of agent $p\in P$} is a pair $c\rightarrow c'$ of agent configurations over $\{p\}$, $S$, $T$, and $\sim$ such that $c^v_p, c'{^v_p} \subseteq (T/{\sim})_p$ and $c^v_p \ne c'{^v_p}$, and $c^m_p = c'{^m_p}$.
    \item A \temph{volitional machine transaction} induced by a volition-guarded transaction $(t,Q')$, for some $t= (d\rightarrow d')\in T$ over $Q'\subseteq Q\subseteq P$, is a pair $c\rightarrow c'$ where $c\ne c'$ are agent configurations over $P$, $S$, $T$, and $\sim$ such that $[t]\in c^v_q$ for every $q\in Q'$; $c^m_p = d_p$ and $c'{^m_p} = d'_p$ for every $p\in Q$; $c^m_p = c'{^m_p}$ for every $p\in P\setminus Q$; and $c'{^v_p} = c^v_p \setminus \{[t]\}$ for every $p\in P$.
    \item A \temph{volitional transaction} is a volition change or a volitional machine transaction.
\end{enumerate}
\end{definition}
When a volitional machine transaction induced by $(t,Q')$ is taken, the class $[t]$ is removed from every agent's volitional state: the will is fulfilled by any equivalent transaction.  A person may independently change their volitional state via volition changes, which may add or remove classes; beyond these, the framework removes a class from $c^v_p$ only upon fulfilment.

\fi

\subsection{Volition-Guarded If-Then-Else Transactions}\label{sec:ite}

In the transactions above, a person's not willing a transaction has no consequence: no transaction is enabled by the absence of a will.  Communicating volitional agents (Section~\ref{sec:cva}) need such a consequence. 
Without it, the implementation of the transactions by communicating volitional agents (Section~\ref{sec:decomposition}) could not close an attempt whose guard does not will it, nor release the parts its participants hold prepared for it.  This is the \emph{volition-guarded if-then-else transaction} of~\cite{lewis2026volitional}: in machine state $s$, if $p$ wills it, $s'$ is taken, else $s''$.  Its else-branch is taken when the person does not will the then-branch.  In this paper the else-branch is used only within CVA.

\begin{definition}[Volition-Guarded If-Then-Else Transaction, Then-Branch, Else-Branch~\cite{lewis2026volitional}]\label{def:refusal}
Given machine states $S$, a \temph{volition-guarded if-then-else transaction} over a single participant $p\in\Pi$ is a triple $s\rightarrow s'\,;\,s''$ \temph{guarded by} $p$, with $s,s',s''\in S$, $s\ne s'$ and $s\ne s''$; the unary machine transaction $s\rightarrow s'$ is its \temph{then-branch} and $s\rightarrow s''$ its \temph{else-branch}.
\end{definition}

\begin{definition}[Else Transaction~\cite{lewis2026volitional}]\label{def:else}
Given agents $P$, states $S$, machine transactions $T$ over $P$ and $S$, and equivalence $\sim$ on $T$, an \temph{else transaction} induced by a volition-guarded if-then-else transaction $s\rightarrow s'\,;\,s''$ guarded by $p\in P$, with $s\rightarrow s'\in T$, is a pair $c\rightarrow c'$ of agent configurations over $P$, $S$, $T$, and $\sim$ such that $[s\rightarrow s']\notin c^v_p$; $c^m_p = s$ and $c'{^m_p} = s''$; $c^m_q = c'{^m_q}$ for every $q\in P\setminus\{p\}$; and $c'{^v_q} = c^v_q$ for every $q\in P$.
\end{definition}

The then-branch, taken when $p$ wills it, is the unary volition-guarded transaction $(s\rightarrow s',\{p\})$ and induces volitional machine transactions per Definition~\ref{def:vmat}; a unary volition-guarded transaction is the special case with no else-branch, and an unguarded one is $(s\rightarrow s',\emptyset)$.  The else transaction is taken when $p$ does not will the then-branch; it fulfils no will and removes no class.  The two are placed in one class below, so that liveness obliges the machine to take one branch or the other, with no time bound.  Only a unary guarded transaction has an else-branch: an unguarded one requires no will, and a binary one would let an agent of another group, who never wills it, trigger it (Definition~\ref{def:grassroots}, Proposition~\ref{prop:oblivious-criterion}).  Where a set of transactions may include volition-guarded if-then-else transactions, as a CVA platform's may (Section~\ref{sec:cva}), \emph{volition-guarded transaction} covers them.

\iffull
\subsection{Transition Systems}

\begin{definition}[Transition System, Computation, Run, Safe, Live, Correct~\cite{shapiro2021multiagent,lewis2026volitional}]
\label{def:ts}
A \temph{transition system} is a tuple $TS = (C, c_0, T, {\sim})$ where $C$ is an arbitrary set of \temph{configurations}, $c_0 \in C$ a designated \temph{initial configuration}, $T \subseteq C \times C$ a set of \temph{transitions}, each a pair $c \rightarrow c'$ of non-identical configurations $c \ne c' \in C$, and ${\sim}$ a partial equivalence relation on $T$ --- an equivalence on a subset $\mathit{dom}({\sim}) \subseteq T$; we write $[t]$ for the class of $t \in \mathit{dom}({\sim})$ under ${\sim}$.

A \temph{computation} is a (nonempty, finite or infinite) sequence of configurations $c_1, c_2, \ldots$; it is a \temph{run} if $c_1 = c_0$, and \temph{safe} if $c_i \rightarrow c_{i+1} \in T$ for every two consecutive configurations. We write $c \xrightarrow{*} c'$ for the existence of a safe computation from $c$ to $c'$ (empty if $c = c'$). A class $[t]$ is \temph{enabled} in $c$ if $c \rightarrow c' \in [t]$ for some $c'$. A run is \temph{live} if no class $[t]$, $t \in \mathit{dom}({\sim})$, is enabled in every configuration of some suffix in which no member of $[t]$ occurs, and \temph{correct} if it is safe and live.
\end{definition}

Liveness is weak fairness, per class.  A partial equivalence, rather than a total one, is used so that some transitions may carry no liveness requirement: only transitions in the domain of ${\sim}$ form classes and thereby incur a liveness obligation, while transitions outside the domain --- belonging to no class --- may occur in a correct run but are never required to.

\begin{definition}[Multiagent Transition System~\cite{shapiro2021multiagent}]\label{def:mts}
Given agents $P \subset \Pi$ and an arbitrary set $S$ of \temph{local states} with a designated \temph{initial local state} $s_0\in S$, a \temph{multiagent transition system} over $P$ and $S$ is a transition system $TS= (C,c_0,T,{\sim})$ with \temph{configurations} $C:= S^P$, \temph{initial configuration} $c_0:= \{s_0\}^P$, \temph{transitions} $T\subseteq C^2$ a set of transactions over $P$ and $S$, and ${\sim}$ a partial equivalence on $T$.
\end{definition}

Rather than specifying a multiagent transition system directly, we specify it via transactions, each over its own participants; the closure operator induces from them transitions over a fixed $P$, in which non-participants are stationary.

\begin{definition}[Transaction Closure~\cite{lewis2026volitional}]\label{def:closure}
Let $P\subset \Pi$, $S$ a set of local states, and $C:=S^P$.  For any transition or transaction $t = c\rightarrow c'$, we write $t_q := c_q\rightarrow c'_q$ and say $p$ is \temph{stationary} in $t$ if $c_p = c'_p$.  For a transaction $t=(c\rightarrow c')$ over $S$ with participants $Q$, the \temph{$P$-closure of $t$}, $t{\uparrow}P$, is the set of transitions over $P$ and $S$ defined by:
$$
t{\uparrow}P := \begin{cases} \{ t' \in C^2  :
\forall q\in Q.(t_q = t'_q) \wedge \forall p\in P\setminus Q.(p\text{ is stationary in }t')\} & \text{if } Q\subseteq P \\
\emptyset & \text{otherwise}
\end{cases}
$$
If $R$ is a set of transactions, each $t\in R$ over some $Q$ and $S$, then the \temph{$P$-closure of $R$} is $R{\uparrow}P := \bigcup_{t\in R} t{\uparrow}P$.  Given a relation ${\sim}$ on $R$, its \temph{$P$-closure} ${\sim}{\uparrow}P$ is the relation on $R{\uparrow}P$ with $\hat t \mathrel{({\sim}{\uparrow}P)} \hat t'$ iff $\hat t \in t{\uparrow}P$ and $\hat t' \in t'{\uparrow}P$ for some $t \sim t'$.
\end{definition}

A transaction names only its participants.  Its closure makes it transitions of the whole system over $P$: the participants move as the transaction says, and every other agent stays where it is.  The closure and the equivalence do different jobs: the closure turns transactions into transitions, and the equivalence groups the transitions that are one act.  If distinct transactions in $R$ have disjoint $P$-closures, then each transition in $R{\uparrow}P$ is induced by one transaction, and ${\sim}{\uparrow}P$ relates two transitions exactly when their inducing transactions are related, so ${\sim}{\uparrow}P$ is a partial equivalence whenever ${\sim}$ is.  Disjointness holds when every participant of every transaction changes state, but a participant need not: a participant whose state is part of the precondition and does not change, such as a voter, is allowed, provided no other transaction in $R$ moves the same agents in the same way.

\begin{definition}[Transactions-Based Multiagent Transition System, adapted from~\cite{shapiro2021multiagent,lewis2026volitional}]\label{def:tbmts}
Given agents $P \subset \Pi$, local states $S$ with initial local state $s_0\in S$, a set of transactions $R$, each $t\in R$ over some $Q\subseteq P$ and $S$, in which distinct transactions have disjoint $P$-closures, and a partial equivalence ${\sim}$ on $R$, the \temph{transactions-based multiagent transition system} over $P$, $S$, $R$, and ${\sim}$ is the multiagent transition system $TS= (S^P,\{s_0\}^P,R{\uparrow}P,{\sim}{\uparrow}P)$.
\end{definition}

A set of volition-guarded transactions with an equivalence on their underlying machine transactions likewise induces a multiagent transition system, over agent states:

\begin{definition}[Volition-Guarded Transaction Set~\cite{lewis2026volitional}]\label{def:vg-set}
Let $P\subset \Pi$ be agents and $S$ machine states with initial state $s_0$.  A \temph{volition-guarded transaction set} over $P$ and $S$ is a set $R$ of volition-guarded transactions and volition-guarded if-then-else transactions, where a volition-guarded if-then-else transaction $s\rightarrow s'\,;\,s''$ guarded by $p$ stands in $R$ for its then-branch $(s\rightarrow s',\{p\})$ together with its else-branch, such that every $(t,Q')\in R$ has the participants of $t$ in $P$, every guard $p$ of a volition-guarded if-then-else transaction is in $P$, and distinct underlying machine transactions, else-branches included, have disjoint $P$-closures.
An equivalence $\sim$ on the set $T_R := \{t : (t,Q')\in R\text{ for some }Q'\}$ of underlying machine transactions (a total equivalence, hence a partial equivalence in the sense of Definition~\ref{def:ts}) is \temph{compatible} with $R$ if no else-branch is in $T_R$, and two volition-guarded if-then-else transactions with the same else-branch have $\sim$-equivalent then-branches.
\end{definition}

\begin{definition}[Volitional Multiagent Transition System~\cite{lewis2026volitional}]\label{def:vmts}
Given a volition-guarded transaction set $R$ over $P$ and $S$ and an equivalence $\sim$ compatible with it, the \temph{volitional multiagent transition system induced by $(S,R,{\sim})$ over $P$} is the multiagent transition system $(\calA^P,c_0,T_V,{\sim_V})$ where:
\begin{enumerate}
    \item $\calA := 2^{T_R/\sim} \times S$ is the \temph{agent state space};
    \item $c_0 \in \calA^P$ is the \temph{initial agent configuration}, with $c_0{^v_p}=\emptyset$ and $c_0{^m_p}=s_0$ for every $p\in P$;
    \item $T_V$ consists of all transitions $e\to e'\in (\calA^P)^2$ of one of three forms: (\ia)~a \temph{volition change} of some $p\in P$---$e^v_p, e'{^v_p}\subseteq (T_R/{\sim})_p$ and $e^v_p \ne e'{^v_p}$, $e^m_p = e'{^m_p}$, and $e_r = e'_r$ for every $r\in P\setminus\{p\}$; (\ib)~a \temph{volitional machine transaction} induced by some volition-guarded transaction $(t,Q')\in R$ per Definition~\ref{def:vmat}(2); or (\ic)~an \temph{else transaction} induced by some volition-guarded if-then-else transaction in $R$ per Definition~\ref{def:else};
    \item $\sim_V$ is the restriction of ${\sim}{\uparrow}P$ (Definition~\ref{def:closure}) to the volitional machine transactions and else transactions, relating two of them whenever their inducing machine transactions are $\sim$-equivalent, where the inducing machine transaction of an else transaction is the then-branch $s\rightarrow s'$ of its volition-guarded if-then-else transaction, and leaves every volition-change transition outside its domain.
\end{enumerate}
\end{definition}
The two conditions on else-branches give each else transaction one class: an else transaction is related through the then-branch of its if-then-else transaction, and two if-then-else transactions with the same else-branch have equivalent then-branches, so $\sim_V$ is a partial equivalence.  The set $R$ of volition-guarded transactions that specifies a platform thus determines the volitional machine transitions of the induced volitional multiagent transition system, while volition-change transitions are freely available to every agent independent of $R$.

We define when a volition-guarded transaction is enabled, combining the machine precondition with the volitions of its guards.

\begin{definition}[Enabled~\cite{lewis2026volitional}]\label{def:enabled}
Given a set of volition-guarded transactions, each $(t,Q')$ with $t = d \rightarrow d'$ a machine transaction over some  $Q' \subseteq Q\subseteq P$ and $S$, and an equivalence $\sim$ on machine transactions: the volition-guarded transaction $(t,Q')$ is \temph{machine-enabled} in agent configuration $c$ over $P$ if $c^m_p = d_p$ for every $p \in Q$, and \temph{enabled} if in addition $[t] \in c^v_q$ for every $q \in Q'$.  The else-branch of a volition-guarded if-then-else transaction $s\rightarrow s'\,;\,s''$ guarded by $p$ is \temph{enabled} in $c$ if $c^m_p = s$ and $[s\rightarrow s']\notin c^v_p$.  An equivalence class $[t]$ is \temph{enabled} in $c$ if some volition-guarded transaction $(t',Q')$ with $t' \in [t]$ is enabled in $c$, or the else-branch of some volition-guarded if-then-else transaction with then-branch $t'\in[t]$ is enabled in $c$.
\end{definition}
A volition-guarded transaction with an empty guard ($Q' = \emptyset$) requires no volitions and is enabled whenever its machine precondition is met.  The class of the then-branch of a volition-guarded if-then-else transaction guarded by $p$ is thus enabled whenever $p$'s machine is in the state $s$: 
by its then-branch if $p$ wills it, and by its else-branch if not.
  For the volitional multiagent transition system above, a class of volitional machine transactions and else transactions is enabled at a configuration in the sense of Definition~\ref{def:ts} exactly when some volition-guarded transaction or else-branch inducing it is enabled in the above sense; this determines which runs are live and correct per Definition~\ref{def:ts}.  Volition-change transitions belong to no class and so impose no liveness obligation; personal choices remain free.

A \emph{protocol} $\calF$ is a family of multiagent transition systems, $\calF(P)$ for each finite $P \subset \Pi$, over a common local-states function assigning to each $P$ a set of local states containing the initial state~\cite{shapiro2023grassrootsBA,lewis2026volitional}.  A set of volition-guarded transactions with a transaction equivalence over a local-states function thus defines a protocol: it assigns to each $P$ the volitional multiagent transition system it induces over $P$.

\subsection{Implementations}

An implementation relates a specification transition system to one that implements it, via a mapping of configurations.

\begin{definition}[Implementation (Refinement)~\cite{shapiro2021multiagent,shapiro2026implementing,lewis2026volitional}]
\label{def:implementation}
Given two transition systems $TS = (C, c_0, T, {\sim})$, the \temph{specification}, and $TS' = (C', c'_0, T', {\sim'})$, an \temph{implementation of $TS$ by $TS'$} is a function $\sigma : C' \rightarrow C$ with $\sigma(c'_0) = c_0$, in which case the pair $(TS', \sigma)$ is \temph{an implementation of $TS$}.  Given a computation $r' = c'_1, c'_2, \ldots$ of $TS'$, $\sigma(r')$ is the computation $\sigma(c'_1), \sigma(c'_2), \ldots$ with consecutive repetitions removed; a transition $c' \rightarrow d'$ of $TS'$ with $\sigma(c') = \sigma(d')$ is a \temph{stutter}, contributing no step to $\sigma(r')$.
\end{definition}

\begin{definition}[Correct and Complete Implementation~\cite{shapiro2021multiagent,shapiro2026implementing,lewis2026volitional}]
\label{def:implementation-properties}
An implementation $(TS', \sigma)$ of $TS$ is \temph{correct} if $\sigma$ maps every correct run of $TS'$ to a correct run of $TS$, and \temph{complete} if every correct run $r$ of $TS$ is $\sigma(r')$ for some correct run $r'$ of $TS'$.
\end{definition}

A correct implementation is a refinement: a map from the implementing system's configurations to the specification's under which every correct run of the one is a correct run of the other.  The implementations of this paper each hold for the runs of the implementing system that meet a condition --- that an initiator knows its participants, that variables are named apart, that a question put to a person is answered --- and we attach the condition to the system.

\begin{definition}[Fairness Condition]\label{def:fairness}
A \temph{fairness condition} on a transition system $TS$ is a set $F$ of runs of $TS$.  A run of $TS$ under $F$ is \temph{correct} if it is correct and in $F$.
\end{definition}

Definition~\ref{def:implementation-properties} applies to a transition system with a fairness condition, whether it is implemented or implements, with its correct runs being those that are correct under the condition: 
where no condition is stated, every run is in it.
\fi

\section{Communicating Volitional Agents}\label{sec:cva}

The framework of Section~\ref{sec:foundations} abstracts over \emph{how} a volition-guarded multiagent atomic transaction is carried out: a $k$-ary transaction simultaneously updates the machine states of all its participants, with no notion of how participating machines come to act in concert.  An implementation on networked smartphones has no such simultaneity: every exchange between two agents is a separate asynchronous message-passing event, and a $k$-ary transaction must be implemented by a protocol composed of such events.

\subsection{Definition}

To achieve that, we define \emph{communicating volitional agents} (CVA), an implementation-ready restriction of volition-guarded transactions; the implementation of volition-guarded transactions by CVA is Section~\ref{sec:decomposition}.  CVA volition-guarded machine transactions are drawn from four syntactic forms: binary \emph{discover} transactions by which two agents come to know each other, by the will of both persons; binary \emph{introduce} transactions by which two agents that a common peer introduced come to know each other, with no person asked; binary \emph{communicate} transactions that move a message from one agent's outbox to another's inbox; and unary \emph{app} transactions.  Thus, a CVA platform need only specify its app transactions.  Moreover every CVA protocol is volitionally grassroots (Proposition~\ref{prop:cva-grassroots}).

\mypara{Machine states} A CVA agent's machine state has four components: a set of \emph{known peers} with whom the agent can communicate, an \emph{outbox} of messages awaiting delivery, an \emph{inbox} of messages received, and an \emph{app state} in which each platform specifies whatever data it requires (friend sets, groups, coins, feeds, and so on).  A \emph{message} is a triple of sender, recipient, and cargo; outboxes and inboxes are sets of messages.  The cargo space $C$ and the app state space $A$ are parameters of the CVA platform.

\begin{definition}[Messages]\label{def:messages}
Given a set $C$ of \temph{cargoes}, the set of \temph{messages} over $P\subset \Pi$ is
\[
M(P) \;:=\; \{\mathsf{message}(s,r,c) \mid s\ne r\in P,\; c\in C\},
\]
with $\mathsf{message}(s,r,c)$ a message of \temph{sender} $s$, \temph{recipient} $r$, and \temph{cargo} $c$.
\end{definition}

A message goes to another agent: an agent changes its own state by an app transaction, with no message.

\begin{definition}[CVA Machine States and Configurations]\label{def:cva-state}
Given messages over $P$ and an \temph{app state space} $A$ with initial state $a_0\in A$, a \temph{CVA machine state} over $P$ is a tuple
\[
(\mathit{known},\, o,\, i,\, a) \;\in\; S(P) \;:=\; 2^{P} \times 2^{M(P)} \times 2^{M(P)} \times A,
\]
where $\mathit{known}\subseteq P$ is a set of \temph{known peers}, $o\subseteq M(P)$ is an \temph{outbox}, $i\subseteq M(P)$ is an \temph{inbox}, and $a\in A$ is an \temph{app state}; the \temph{initial machine state} is $(\emptyset,\emptyset,\emptyset,a_0)$.  A \temph{CVA machine configuration} over $P$ is a member $c\in S(P)^P$; for $p\in P$ we write $c_p = (\mathit{known}_p,\, o_p,\, i_p,\, a_p)$ for the machine state of $p$ in $c$.  The \temph{initial machine configuration} over $P$ assigns the initial machine state to every $p\in P$.
\end{definition}

$P\subseteq P' \implies S(P)\subseteq S(P')$, so that the agent states of the transition system below (Definition~\ref{def:agent-state}) grow with $P$, as a protocol requires.

\begin{definition}[App Transaction, Effect]\label{def:app-transaction}
Given cargoes $C$ and an app state space $A$, an \temph{app transaction} of $p$ is a unary volition-guarded transaction, or a volition-guarded if-then-else transaction (Definition~\ref{def:refusal}), over $p$ and CVA machine states, each of whose targets $c'_p$, from its source $c_p$, has $\mathit{known}'_p=\mathit{known}_p$, $i'_p\subseteq i_p$, $o_p\subseteq o'_p$ and $o'_p\setminus o_p\subseteq\{\mathsf{message}(p,r,x): r\in\mathit{known}_p,\ x\in C\}$, a cargo $\mathsf{intro}(q)$ having $q\in\mathit{known}_p$.  Its \temph{effect} $\mathit{eff}(\tau)$ is $(a'_p,\ o'_p\setminus o_p,\ i_p\setminus i'_p)$, and its \temph{else-effect} $\mathit{eff}'(\tau)$ is the same triple for the else-branch's target.
\end{definition}

\begin{definition}[CVA Platform]\label{def:cva-platform}
A \temph{CVA platform} is a tuple $(C, A, a_0, T, {\sim})$ of cargoes $C$, which include the \temph{introductions} $\mathsf{intro}(q)$, $q\in\Pi$, an app state space $A$ with initial state $a_0\in A$, a set $T$ of app transactions, and a transaction equivalence ${\sim}$ on $T$, such that for $\tau\sim\tau'$ in $T$ whose sources differ only in $\mathit{known}$, $\mathit{eff}(\tau)=\mathit{eff}(\tau')$ and $\mathit{eff}'(\tau)=\mathit{eff}'(\tau')$.
The \temph{transactions} of the platform over $P\subset\Pi$ are the members of $T$ over $P$ and, for all $p\ne q\in P$ and $m\in M(P)$, the following machine transactions, in which nothing else changes:
\begin{enumerate}
    \item \textbf{Discover}$(p,q)$, over $\{p,q\}$ and guarded by $\{p,q\}$: $\mathit{known}'_p = \mathit{known}_p\cup\{q\}$ and $\mathit{known}'_q = \mathit{known}_q\cup\{p\}$;
    \item \textbf{Introduce}$(p,q)$, over $\{p,q\}$, unguarded, provided $q\notin\mathit{known}_p$ and, for some $r$ and $r'$, $m=\mathsf{message}(r,p,\mathsf{intro}(q))$ is in $i_p$ and $m'=\mathsf{message}(r',q,\mathsf{intro}(p))$ in $i_q$: $\mathit{known}'_p = \mathit{known}_p\cup\{q\}$, $\mathit{known}'_q = \mathit{known}_q\cup\{p\}$, $i'_p = i_p\setminus\{m\}$ and $i'_q = i_q\setminus\{m'\}$;
    \item \textbf{Communicate}$(p,q,m)$, over $\{p,q\}$, unguarded, provided $m=\mathsf{message}(p,q,x)\in o_p$ and $m\notin i_q$: $o'_p = o_p\setminus\{m\}$ and $i'_q = i_q\cup\{m\}$;
\end{enumerate}
with ${\sim}$ extended to relate every Discover and Introduce of the same $p$ and $q$, and two Communicates of the same $p$, $q$ and $m$.
\end{definition}

\begin{definition}[CVA Agent State, Transition System, Protocol]\label{def:cva-agent-state}
A \temph{CVA agent state} of $p$ is a pair $(V,s)$ of a \temph{volitional state} $V$, a set of classes of transactions guarded by $p$, and a CVA machine state $s\in S(P)$; a \temph{CVA agent configuration} over $P$ assigns each $p\in P$ a CVA agent state of $p$.  The \temph{transition system} of a CVA platform over $P$ is the volitional multiagent transition system its transactions over $P$ induce (Definition~\ref{def:vmts}), on CVA agent configurations; the \temph{protocol} of the platform assigns it to each finite $P\subset\Pi$.
\end{definition}

Discover, Introduce and Communicate are binary.  Knowing is mutual, so Discover and Introduce each add each of the two agents to the other's known peers at once, and Discover is guarded by both persons, which a unary transaction, guarded at most by its one participant, cannot be.  Communicate moves a message from one agent to the other.

A CVA platform specifies its app transactions; Discover, Introduce and Communicate are built in.  An app transaction reads its agent's machine state and changes only the app state, the outbox, adding messages addressed to known peers, and the inbox, deleting messages it has read.  It may require that a peer be known, and its effect does not depend on which other peers are known.  The reason is the implementation of Section~\ref{sec:cva-implementation}: there an agent comes to know a peer by exchanging hellos with it, one hello at a time, so while a Discover is under way the two agents' sets of known peers differ from the specification's.  A test that one peer is known stays correct, and a computation over the whole set would not. 
The transition system of a platform is a volitional multiagent transition system (Definition~\ref{def:vmts}), which requires distinct transactions to have disjoint $P$-closures, and CVA transactions do.  Each changes the state of every one of its participants: a binary one changes the states of both, and a unary one the state of its one participant.  Two transactions of different forms change different components: Discover changes only the known peers; Introduce changes the known peers and the two inboxes; Communicate changes an outbox and an inbox; and an app transaction changes the app state, the outbox and the inbox, and not the known peers.

An app transaction may be unguarded, no person's will being needed for it, 
and its precondition may require an arrived message.  A transaction that requires a message may delete it.
A guarded app transaction that requires an arrived message may have an else-branch.  Then it is taken on the message either way, by its then-branch where its person wills it and by its else-branch where they do not, and either branch may delete the message.
  A Discover carries no message: each takes the other's identifier from what the other shows; an implementation realises it by a cold call that the called agent accepts, the two hellos of Section~\ref{sec:cva-implementation}.  An Introduce is taken by the machines, with no person asked: a peer that knows two agents sends each an introduction to the other, and the two come to know each other; it is equivalent to their Discover, and so fulfils a will to discover.  Machines come to know each other by a Discover their persons will, or by an introduction; the platform records who is a friend, a parent or a child.

\mypara{CVA protocols are volitionally grassroots} A protocol is \emph{grassroots} if its groups run undisturbed by one another and can come to interact, and \emph{volitionally grassroots} if their first interaction is willed by a member of each~\cite{lewis2026volitional}; the definitions are recalled in Appendix~\ref{app:grassroots}.  Both rest on the same fact: two agents learn of each other only by Discover, which both persons will, or by Introduce, through a peer that already knows both, and an agent addresses only agents it knows.

\begin{proposition}[CVA protocols are volitionally grassroots]\label{prop:cva-grassroots}
The protocol of every CVA platform is volitionally grassroots, hence grassroots.
\end{proposition}

The proof is in Appendix~\ref{app:proofs}.

\subsection{Implementing the Transactions by Communicating Volitional Agents}\label{sec:decomposition}

A volition-guarded transaction changes the machine states of its $k$ participants at once (Section~\ref{sec:foundations}).  A communicating volitional agent changes only its own state, and reaches another agent only by a message (Section~\ref{sec:cva}).  Here we implement the first by the second, for every local volition-guarded specification with simple guards, by two-phase commit coordinated by the agent that initiates the transaction~\cite{gray1978notes}.

\mypara{Local specifications}  In the framework of Section~\ref{sec:foundations} each agent has a machine state (Definition~\ref{def:mt}), and a transaction's precondition may relate the states of its participants.  In CVA a participant reads and changes only its own state, and learns of the transaction, with its arguments, from the initiator's message.  The implementation therefore requires two things of a transaction.  First, its precondition is a conjunction of preconditions, one on the state of each participant, which may be the trivial one, true. 
Second, its effect on each participant depends only on that participant's state and on the transaction's arguments.  The implementation below assumes this.  For example, the precondition that two participants' dates be equal is made local by naming the date in the transaction, which the initiator sends to every participant, each checking its own date against it.  The platforms of Section~\ref{sec:platforms} are local: befriending adds each party to the other's friend set, provided each is absent from it, and a swap moves each party's coins to the other, provided each holds what it offers.

The framework of Section~\ref{sec:foundations} specifies a platform by a set of volition-guarded transactions and an equivalence on them.  The definition below calls these, with the machine states, a \emph{specification}, and the paragraph after it names, for each class of transactions, the sets of persons whose volitions can enable it, which is what the implementation needs.
\begin{definition}[Volition-Guarded Specification]\label{def:vg-spec}
A \temph{volition-guarded specification} is a tuple $(S, s_0, R, {\sim})$ of machine states $S$ with initial state $s_0\in S$, a set $R$ of volition-guarded transactions over $S$, none of them a volition-guarded if-then-else transaction, and a transaction equivalence ${\sim}$ on their underlying machine transactions $T_R$, as in Definition~\ref{def:vg-set}, under which equivalent transactions have the same minimal guards: the inclusion-minimal members of $\{Q' : (t,Q')\in R\}$ and of $\{Q' : (t',Q')\in R\}$ are the same for every $t\sim t'$.
\end{definition}

The \emph{protocol} of the specification assigns to each $P\subset\Pi$ the volitional multiagent transition system that $(S,R,{\sim})$ induces over $P$ (Definition~\ref{def:vmts}), restricted to the agent configurations in which every agent wills only classes it guards.  For a class $[t]\in T_R/{\sim}$, $Q([t])$ is the set of participants of its members, and the \emph{guard sets} of $[t]$, $\calQ([t])$, are the inclusion-minimal sets among $\{Q' : (t',Q')\in R,\ t'\sim t\}$: a larger set would enable no transaction that a smaller one does not.  A \emph{guard} of $[t]$ is a member of one of its guard sets.  The class is \emph{unguarded} if $\calQ([t])=\{\emptyset\}$, and \emph{guarded} otherwise.  The specification has \emph{simple guards} if every class has one guard set, all of whose members must will it, or guard sets of one person each, any one of whom may will it.  Simple guards are needed for liveness: with two guard sets that share a person, an attempt may ask the wrong set again and again (the paragraph after Theorem~\ref{thm:decomposition}).

\begin{definition}[Local Specification, Opener]\label{def:local}
A volition-guarded specification is \temph{local} if for every class $[t]$ and every $p\in Q([t])$ there is a partial function $f^{[t]}_p : S \rightharpoonup S$ such that the members of $[t]$ are exactly the machine transactions $d\rightarrow d'$ over $Q([t])$ with $d_p\in \mathit{dom}(f^{[t]}_p)$ and $d'_p = f^{[t]}_p(d_p)$ for every $p\in Q([t])$, every such pair having $d'\ne d$.  
A local specification designates for each unguarded class $[t]$ a participant $o([t])\in Q([t])$, its \temph{opener}. 
\end{definition}

Under locality each participant $p$ of $[t]$ has its own precondition, the domain $\mathit{dom}(f^{[t]}_p)$, and its own effect, $f^{[t]}_p$, and a member of $[t]$ is enabled exactly when the precondition of every participant holds.  The class $[t]$ carries the transaction's arguments, so $f^{[t]}_p$ may depend on them.  For example, the effect of a swap on $p$ removes the coins $p$ offers and adds the coins $p$ wants, both named in the swap, and does not read the other party's state.

\mypara{The protocol}  An \emph{attempt} at a transaction is proposed by one of its participants, its \emph{initiator}.  For a guarded class the initiator is a guard whose person wills the transaction, and for an unguarded class it is the class's opener.  The other members of the initiator's guard set are \emph{asked}: their persons must approve.  The machine of every other participant votes on its own precondition, without asking its person.  Each participant, and no other agent, votes once, yes or no, to the initiator.  The initiator commits when it holds a yes from every participant, and aborts on the first no, or when asked to yield before it holds them all.  
A participant votes yes only if it can carry its part out and holds no other attempt \emph{prepared}.  From its yes vote until it learns the outcome, it holds the attempt prepared.  An answer yes approves its attempt until the attempt closes, and no transaction of the protocol withdraws it.  An asked participant first checks that it can carry its part out, and only then asks its person, once.  It votes no at once if its person answered no; on yes, it votes once it holds no other attempt prepared, yes if it still can and no otherwise.  

What should a participant do with a start that arrives while it holds another attempt prepared?  Voting no would lead to livelock: two attempts that reach two agents in different orders could abort each other again and again.  Waiting would lead to deadlock: two attempts could each wait for the other.  We use wound-wait~\cite{rosenkrantz1978system}.  A participant that holds another attempt prepared, or has a question of the same class pending, lets the start wait, and takes waiting starts in order of \emph{priority}.  The priority of an attempt is its \emph{rank}, a number, the smaller rank having priority, with ties broken by the \emph{order of identifiers}, a fixed total order of the agents.  Where a waiting attempt has priority over the attempt that blocks it, the participant asks the initiator of the blocking attempt to \emph{yield}, and that initiator aborts it unless it has already decided it.  An initiator proposes nothing while a start waits in its inbox, or while a question it answered yes waits.  So an attempt waits only on one with priority over it, or on one that is yielding or has been decided and so ends without waiting, and no cycle of waiting arises.  No participant votes no because of contention: contention aborts an attempt only by a yield to an attempt with priority over it.  
An attempt aborted in favour of one with priority over it keeps its rank when its initiator proposes the class again, so only finitely many attempts ever have priority over it, and none is starved.

\begin{definition}[Attempt]\label{def:attempt}
Fix a local volition-guarded specification with simple guards.  An \temph{attempt} is a tuple $\alpha=(i,k,\pi,[t])$ of an \temph{initiator} $i$, a \temph{number} $k\in\calN$, a \temph{rank} $\pi\le k$ and a class $[t]$. 
The initiator is a guard of $[t]$ if $[t]$ is guarded, and its opener $o([t])$ if not.  The \temph{participants} of $\alpha$ are those of $[t]$.  Its \temph{guard set} $G_\alpha$ is the guard set of $[t]$ where $[t]$ has one, $\{i\}$ where the guard sets of $[t]$ are singletons, and $\emptyset$ if $[t]$ is unguarded; the members of $G_\alpha\setminus\{i\}$ are \temph{asked} by it.
\end{definition}

The number identifies the attempt: each attempt an initiator proposes has a new number, larger than every number it has seen.  The rank orders the attempt: it is the number itself, or, where the attempt retries one of the same class that was aborted in favour of an attempt with priority over it, the rank of that one, so that the retried attempt keeps its priority.

\begin{definition}[Cargo, App State, Question, Free, Ready]\label{def:app-state}
A \temph{cargo} of the protocol is $\mathsf{start}(\alpha)$, $\mathsf{vote}(\alpha,y)$ with $y\in\{\mathsf{yes},\mathsf{no}\}$, $\mathsf{outcome}(\alpha,z)$ with $z\in\{\mathsf{commit},\mathsf{abort}\}$, or $\mathsf{yield}(\alpha)$. 
An \temph{app state} of $p$ is a tuple $(s,n,L,R)$ of a machine state $s\in S$; $p$'s \temph{number} $n\in\calN$, the largest number of an attempt it has proposed or received a start of; 
a \temph{lock} $L$, either $\bot$ or the attempt $p$ holds \temph{prepared}; and a finite \temph{record} $R$ of terms $\mathsf{asked}(\alpha)$, $\mathsf{answered}(\alpha,y)$, $\mathsf{voted}(\alpha,y)$, $\mathsf{got}(\alpha,q,y)$, $\mathsf{closed}(\alpha)$, $\mathsf{wounded}(\alpha)$ and $\mathsf{yielded}(\alpha)$.  An attempt $\alpha$ of $[t]$ is a \temph{question} of $[t]$ \temph{pending} at $p$ if $p$ has recorded $\mathsf{asked}(\alpha)$ and neither a vote on $\alpha$ nor $\mathsf{closed}(\alpha)$.  The agent $p$ is \temph{free for} $[t]$ if $L=\bot$ and no question of $[t]$ is pending at $p$, and \temph{ready for} $[t]$ if moreover its precondition holds, $s\in\mathit{dom}(f^{[t]}_p)$. 
\end{definition}

Being free and being ready are conditions on the machine state, not on the person's will.

\begin{definition}[Priority, Waiting, Blocked, Clear]\label{def:waiting}
Fix a total order on $\Pi$, the \temph{order of identifiers}.  An attempt $\alpha$ of rank $\pi$ has \temph{priority} over an attempt $\beta$ of rank $\pi'$ if $\pi<\pi'$, or $\pi=\pi'$ and the initiator of $\alpha$ precedes that of $\beta$ in the order of identifiers.  A message $\mathsf{start}(\alpha)$, $\alpha$ of $[t]$, in $p$'s inbox \temph{waits} at $p$ in three cases: while $p$'s inbox holds the start of an attempt with priority over $\alpha$; while a question with priority over $\alpha$, answered yes, is pending at $p$; and while $p$ is not free for $[t]$. 
In the last case the start is \temph{blocked} by $L$ if $L\ne\bot$, and otherwise by the question of $[t]$ pending at $p$.  A question $\alpha$ pending at $p$ and answered yes \temph{waits} while a question with priority over it and answered yes is pending at $p$, and while $L\ne\bot$, when it is \temph{blocked} by $L$.  The agent $p$ is \temph{clear} if no start is in its inbox and no question answered yes is pending at it.
\end{definition}

The app transactions of the canonical platform, defined below, use two operations on an attempt $\alpha=(i,k,\pi,[t])$ at an agent $p$: \temph{tell} $x$ adds $\mathsf{message}(p,q,x)$ to $o_p$ for every participant $q\ne p$ of $\alpha$, and \temph{vote} $y$ records $\mathsf{voted}(\alpha,y)$, adds $\mathsf{message}(p,i,\mathsf{vote}(\alpha,y))$ to $o_p$, and, if $y=\mathsf{yes}$, sets $L:=\alpha$.  A transaction's precondition is a condition on $p$'s machine state, and its guard, where it has one, is $p$ (Definition~\ref{def:mt}).  Every app transaction addresses only known peers (Definition~\ref{def:app-transaction}).

\begin{definition}[Proposing]\label{def:canonical-voting}
The \temph{proposing and voting transactions} of an agent $p$ are the app transactions of this definition and the next two.  Each deletes the message it reads, except \textbf{Wound}, and each is taken only where $p$ knows every agent it sends a message to. 
\begin{enumerate}
\item \textbf{Propose} $\alpha$, by its initiator $p=i$.  Precondition: $p$ is ready for $[t]$, so $L=\bot$, and clear.  Effect: $p$'s number $n$ is increased by one, and the new value is the number $k$ of $\alpha$. 
The rank $\pi$ of $\alpha$ is that of $p$'s last attempt of $[t]$ if $p$ has recorded that attempt wounded, and $k$ otherwise. 
$p$ records $\mathsf{voted}(\alpha,\mathsf{yes})$.  Then, if $Q([t])=\{p\}$, $s:=f^{[t]}_p(s)$ and $p$ records $\mathsf{closed}(\alpha)$; otherwise $L:=\alpha$ and $p$ tells $\mathsf{start}(\alpha)$.  Guarded by $p$ if $[t]$ is guarded, and unguarded if not.

\end{enumerate}
\end{definition}

\begin{definition}[Starting and Wounding]\label{def:canonical-start}
The proposing and voting transactions of $p$ include the following.
\begin{enumerate}
\setcounter{enumi}{1}
\item \textbf{Start}, on $\mathsf{start}(\alpha)$.  Effect: $p$'s number $n$ becomes the larger of $n$ and $k$. 
If $\mathsf{closed}(\alpha)$ is recorded, nothing more.  Otherwise, provided the start does not wait at $p$: if $s\notin\mathit{dom}(f^{[t]}_p)$, $p$ votes $\mathsf{no}$; if $s\in\mathit{dom}(f^{[t]}_p)$ and $p$ is asked, $p$ records $\mathsf{asked}(\alpha)$; and if $s\in\mathit{dom}(f^{[t]}_p)$ and $p$ is not asked, $p$ votes $\mathsf{yes}$. 
Unguarded.

\item \textbf{Wound}, on a start of $\alpha$ waiting at $p$, or on a question $\alpha$ pending at $p$ and answered yes.  Precondition: $\alpha$ is blocked by an attempt $\beta$ over which $\alpha$ has priority; $p$ has not recorded $\mathsf{yielded}(\beta)$; and, if $p$ initiated $\beta$, $p$ has not recorded a yes of every other participant of $\beta$.  Effect: if $p$ initiated $\beta$, $L:=\bot$, $p$ records $\mathsf{closed}(\beta)$ and $\mathsf{wounded}(\beta)$, and $p$ tells $\mathsf{outcome}(\beta,\mathsf{abort})$.  Otherwise $p$ adds $\mathsf{message}(p,j,\mathsf{yield}(\beta))$ to $o_p$, $j$ being the initiator of $\beta$, and records $\mathsf{yielded}(\beta)$.  A start it reads stays in $p$'s inbox.  
Unguarded.

\end{enumerate}
\end{definition}

\begin{definition}[Answering and Voting]\label{def:canonical-answer}
The proposing and voting transactions of $p$ include the following.
\begin{enumerate}
\setcounter{enumi}{3}
\item \textbf{Answer} $\alpha$.  Precondition: $\alpha$ is a question pending at $p$, and $p$ has recorded no answer to it.  Effect: $p$ records $\mathsf{answered}(\alpha,\mathsf{yes})$.  Guarded by $p$.  Its \temph{else-branch} (Definition~\ref{def:refusal}) records $\mathsf{answered}(\alpha,\mathsf{no})$.

\item \textbf{Vote} $\alpha$.  Precondition: $\alpha$ is a question pending at $p$, $p$ has recorded $\mathsf{answered}(\alpha,y)$, and $y=\mathsf{no}$ or $\alpha$ does not wait at $p$ (Definition~\ref{def:waiting}).
Effect: $p$ votes $\mathsf{yes}$ if $y=\mathsf{yes}$ and $s\in\mathit{dom}(f^{[t]}_p)$, and $\mathsf{no}$ otherwise.  Unguarded.
\end{enumerate}
\end{definition}

\begin{definition}[Deciding]\label{def:canonical-deciding}
The \temph{deciding transactions} of an agent $p$ are the app transactions below.  Each deletes the message it reads, and each is taken only where $p$ knows every agent it sends a message to.
\begin{enumerate}
\setcounter{enumi}{5}
\item \textbf{Tally}, by the initiator $p=i$, on $\mathsf{vote}(\alpha,y)$ from $q$.  Effect: if $\mathsf{closed}(\alpha)$ is recorded, nothing more.  Otherwise $p$ records $\mathsf{got}(\alpha,q,y)$, and then, if $y=\mathsf{no}$, $L:=\bot$, $p$ records $\mathsf{closed}(\alpha)$ and tells $\mathsf{outcome}(\alpha,\mathsf{abort})$.  Unguarded.

\item \textbf{Commit} $\alpha$, by its initiator $p=i$.  Precondition: a yes of every participant other than $p$ is recorded, and $\mathsf{closed}(\alpha)$ is not.  Effect: $s:=f^{[t]}_p(s)$, $L:=\bot$, and $p$ records $\mathsf{closed}(\alpha)$ and tells $\mathsf{outcome}(\alpha,\mathsf{commit})$.  Unguarded.

\item \textbf{Yield}, by the initiator $p=i$, on $\mathsf{yield}(\alpha)$ from $q$.  Effect: if $\mathsf{closed}(\alpha)$, or a yes of every participant other than $p$, is recorded, nothing more.  Otherwise $L:=\bot$, and $p$ records $\mathsf{closed}(\alpha)$ and $\mathsf{wounded}(\alpha)$ and tells $\mathsf{outcome}(\alpha,\mathsf{abort})$.  Unguarded.

\item \textbf{Outcome}, by a participant $p\ne i$, on $\mathsf{outcome}(\alpha,z)$ from $i$.  Effect: $s:=f^{[t]}_p(s)$ if $z=\mathsf{commit}$, $L:=\bot$ if $L=\alpha$, and $p$ records $\mathsf{closed}(\alpha)$.  Unguarded.
\end{enumerate}
\end{definition}

\begin{definition}[Equivalence of the Canonical Platform]\label{def:canonical-equivalence}
The \temph{transaction equivalence} of the canonical platform has the following classes.
\begin{itemize}
\item For a guarded class $[t]$ and a guard $p$ of it, $p$'s \temph{approval class} on $[t]$: $p$'s \textbf{Propose} transactions for attempts of $[t]$, and $p$'s \textbf{Answer} transactions, by either branch, on attempts of $[t]$.  It is guarded by $p$. 
\item For an unguarded class $[t]$, the \temph{opening class} on $[t]$ of its opener: its \textbf{Propose} transactions for attempts of $[t]$.
\item For an agent and a message: its \textbf{Start}, \textbf{Tally}, \textbf{Yield} or \textbf{Outcome} transactions on that message.
\item For an agent and an attempt: its \textbf{Vote} transactions on that attempt, and its \textbf{Commit} transactions on it.
\item For an agent and two attempts: its \textbf{Wound} transactions for them.
\end{itemize}
\end{definition}

The approval class is what a person wills in CVA when they will $[t]$ in the specification: their machine fulfils it either by proposing an attempt of $[t]$ or by answering yes to one.  The else-branch of \textbf{Answer} records the answer $\mathsf{no}$, which no other transaction records, so no else-branch is a machine transaction of the platform.  It records its attempt, so two \textbf{Answer} transactions with the same else-branch are of one class (Definition~\ref{def:vmts}).

\begin{definition}[Canonical CVA Platform]\label{def:canonical-cva}
The \temph{canonical CVA platform} of a local volition-guarded specification with simple guards is the CVA platform (Definition~\ref{def:cva-platform}) with the cargoes of Definition~\ref{def:app-state} and the introductions, the app states of Definition~\ref{def:app-state}, initial app state $(s_0,0,\bot,\emptyset)$, the app transactions of Definitions~\ref{def:canonical-voting} to~\ref{def:canonical-deciding}, and the transaction equivalence of Definition~\ref{def:canonical-equivalence}.
\end{definition}

\mypara{Correctness}  The transaction is carried out, at every participant at once, by the \textbf{Tally} that records the last yes of its attempt at its initiator: from then on every participant's effect is settled, and the \textbf{Commit} and \textbf{Outcome} that apply it later change nothing the specification records.

\begin{definition}[Decided, Committed, Aborted, Settled State]\label{def:decided}
Let $c$ be a configuration of the canonical CVA platform over $P$.  An attempt $\alpha$ is \temph{decided} in $c$ if its initiator has recorded a yes of every other participant, or, $\alpha$ having one participant, has proposed it.  It is \temph{closed} at a participant $p$ once $p$ has recorded $\mathsf{closed}(\alpha)$, \temph{committed} there if the transaction that recorded it applied $p$'s effect, and \temph{aborted} there otherwise.  The \temph{settled state} of $p$ in $c$ is $s_p$, with $f^{[t]}_p$ applied if $L_p$ is an attempt of $[t]$ decided in $c$ and $s_p\in\mathit{dom}(f^{[t]}_p)$.
\end{definition}

For the rest of this section fix a local volition-guarded specification with simple guards and agents $P\subset\Pi$, and write $\calP$ for the transition system of the canonical CVA platform over $P$ and $\calF(P)$ for the volitional multiagent transition system the specification induces over $P$.

\begin{restatable}[Atomic Commitment]{proposition}{restateCommitment}\label{prop:commitment}
In every safe run of $\calP$: (\emph{agreement}) no attempt is committed at one participant and aborted at another; and (\emph{validity}) an attempt is committed at a participant only if it is decided, each participant having voted yes on it holding no other attempt prepared, at a machine state in the domain of its effect, and each member of its guard set having willed its approval class on the attempt's class as it proposed the attempt or answered yes to it.  In every safe, live run of $\calP$: (\emph{termination}) every attempt is eventually closed at each of its participants.
\end{restatable}

The proof is in Appendix~\ref{app:proofs}.

\begin{definition}[Connected, Unobstructed]\label{def:connected}
A run of $\calP$ is \temph{connected} if an agent $p$ that, in infinitely many of its configurations, has its settled state (Definition~\ref{def:decided}) in $\mathit{dom}(f^{[t]}_p)$ and holds its approval class on $[t]$ or is the opener of the unguarded class $[t]$ knows every other participant of $[t]$ from some configuration on.  A run is \temph{unobstructed} if, for every agent $p$ and class $[t]$, if infinitely many of its configurations have the settled state of $p$ (Definition~\ref{def:decided}) 
in $\mathit{dom}(f^{[t]}_p)$ and $p$ holding its approval class (Definition~\ref{def:canonical-equivalence})
on $[t]$ or being the opener of the unguarded $[t]$, then $p$ is free for $[t]$ and clear in infinitely many of them.
\end{definition}

\begin{definition}[Strongly Fair, Fairness of $\calP$]\label{def:fairness-cva}
A run of $\calP$ is \temph{strongly fair} if, for every agent $p$ and class $[t]$ such that a \textbf{Propose} of $p$ for an attempt of $[t]$ is enabled in infinitely many of its configurations, a \textbf{Propose} of $p$ for an attempt of $[t]$ is taken, or an attempt of $[t]$ is decided, in infinitely many.  The \temph{fairness condition} of $\calP$ is the set of its connected, unobstructed and strongly fair runs.
\end{definition}

The protocol cannot secure these three conditions.  Knowing is never undone, so whether a run is connected follows from its Discover and Introduce transactions.  Unobstructedness excludes an agent kept busy without end.  The protocol resolves contention, and aborts an attempt only by a yield to an attempt with priority over it.  An attempt that aborts changes no machine state of the specification, and spends the approvals its guards gave it, a volition change of each (Definition~\ref{def:decomposition-sigma}): a guard who still wills the transaction wills it again.  The specification's liveness is weak fairness.  For our implementation to achieve it, strong fairness is needed.  The reason is that \textbf{Propose} may be enabled infinitely often but never continuously.  The compiled agent of Section~\ref{sec:cva-implementation} secures strong fairness (Section~\ref{sec:cascade}).

\begin{definition}[Implementation of Volition-Guarded Specifications by CVA, Correct, Complete]\label{def:vg-by-cva}
The \temph{implementation of volition-guarded specifications by CVA} is the function taking a local volition-guarded specification with simple guards to its canonical CVA platform (Definition~\ref{def:canonical-cva}).  It is \temph{correct} if for every such specification and $P\subset\Pi$ some $\sigma$ is a correct implementation of $\calF(P)$ by $\calP$, and \temph{complete for strongly fair runs} if for every such specification and $P\subset\Pi$ some such $\sigma$ also takes some correct run of $\calP$ to each correct run of $\calF(P)$ that is \temph{strongly fair}, every class enabled in infinitely many of its configurations being taken in infinitely many.
\end{definition}

\begin{restatable}[CVA Implementation]{theorem}{restateCVAImplementation}\label{thm:decomposition}
The implementation of volition-guarded specifications by CVA is correct, and complete for strongly fair runs.
\end{restatable}

The proof, with the mapping $\sigma$, is in Appendix~\ref{app:proofs}.  Correctness is per transition and then per class: the yes that decides an attempt is the specification's transaction, enabled by validity, and every other step a stutter or a volition change; and a class enabled forever is decided, since it is proposed again and again, contention aborts its attempts only in favour of the finitely many attempts with priority over them, and with simple guards no person whose will keeps the class enabled refuses it.  Simple guards are needed: with two guard sets sharing a person, as Alice with Bob and Alice with Carol, Alice's attempt that asks Bob while only Carol wills fails, and so does her next, which asks Carol while only Bob wills, and the class, enabled throughout, may never be taken.  Completeness is by construction: for a correct, strongly fair run of the specification, a run of $\calP$ in which every agent knows every other carries its transactions out one at a time, and idles, between them, on the attempts its persons' unmet wills open.

\section{Grassroots Logic Programs}\label{sec:glp}

This section recalls GLP~\cite{shapiro2025glp}, the polymorphically-typed multiagent concurrent logic programming language that vGLP extends (Section~\ref{sec:vglp}): its syntax, its reduction, and its multiagent operational semantics maGLP, which vmaGLP extends with the person.  GLP is typed by the parameterised moded type system of~\cite{shapiro2026types}, which proved indispensable for AI-based programming and is not restated here: a program carries type definitions and procedure declarations.  Its implementation is~\cite{shapiro2026implementing}.

\subsection{Syntax}

The language extends logic programs by adding a paired \emph{reader} $X?$ to every ``ordinary'' logic variable $X$, now called a \emph{writer}~\cite{shapiro2025glp}.

\begin{definition}[Variables, Terms, Goals, Clauses~\cite{shapiro2025glp}]
\label{def:glp-variables}
\label{def:so-invariant}
Let $\calV$ denote the set of LP variables (identifiers beginning with uppercase), henceforth called \temph{writers}, and $\calV? = \{X? \mid X \in \calV\}$ the \temph{readers}; $\hat\calV = \calV \cup \calV?$, and a writer $X$ and its reader $X?$ form a \temph{variable pair}.  Terms, unit goals, goals, and clauses are as in LP, over $\hat\calV$.  A term, goal, or clause satisfies the \temph{single-occurrence (SO) invariant} if every variable occurs in it at most once, and a clause satisfies the \temph{single-reader/single-writer (SRSW) restriction} if it satisfies SO and a variable occurs in it iff its paired variable does.
\end{definition}

\begin{definition}[Guarded Clause, Procedure, GLP Program~\cite{shapiro2025glp}]
\label{def:guarded-clause}
\label{def:glp-program}
A \temph{guarded clause} has the form $H$ \verb|:-| $G$ \verb"|" $B$, where $H$ is the head, $G$ a conjunction of guard predicates, and $B$ the body; guard arguments are readers paired to head writers.  A \temph{GLP program} is a finite sequence of guarded clauses, each satisfying SRSW; the clauses for the same predicate form a \temph{procedure}.  The set of goals $\hat\calG(M)$ of a program $M$ includes all goals over $\hat\calV$ and the vocabulary of $M$ that satisfy SO.
\end{definition}

\begin{example}[Fair Merge~\cite{shapiro2025glp}]
\label{ex:merge}
The quintessential concurrent logic program for fairly merging two streams, in ordinary clauses:\footnote{Parameterised moded-type definitions (\texttt{T ::= ...}) and declarations (\texttt{procedure ...}) used informally in examples are formally introduced in~\cite{shapiro2026types}.}
\begin{verbatim}
Stream(X) ::= [] ; [X|Stream(X)].

procedure merge(Stream(X)?, Stream(X)?, Stream(X)).
merge([X|Xs], Ys, [X?|Zs?]) :- merge(Ys?, Xs?, Zs).
merge(Xs, [Y|Ys], [Y?|Zs?]) :- merge(Xs?, Ys?, Zs).
merge([], Ys, Ys?).
merge(Xs, [], Xs?).
\end{verbatim}
The goal \verb=merge([1,2,3|Xs?],[a,b|Ys?],Zs)= satisfies SO; each clause satisfies SRSW. The first clause swaps inputs in the recursive call, ensuring fairness when both streams are available.
\end{example}

\iffull
\subsection{Reduction}

\begin{definition}[Substitutions and Assignments~\cite{shapiro2025glp}]
\label{def:writers-assignment}
A \temph{writer assignment} is a term of the form $X := T$, $X\in\calV$, $T\notin\calV$, satisfying SO. Similarly, a \temph{reader assignment} is a term of the form $X? := T$, $X?\in\calV?$, $T\notin\calV$, satisfying SO. A \temph{writers (readers) substitution} $\sigma$ is the substitution implied by a set of writer (reader) assignments that jointly satisfy SO. Given a writers assignment $X := T$, its \temph{readers counterpart} is $X? := T$, and given a writers substitution $\sigma$, its \temph{readers counterpart} $\sigma?$ is the readers substitution defined by $X?\sigma? = X\sigma$. Given a reader assignment $X? := T$, its \temph{writers counterpart} is $X := T$, and given a readers substitution $\tau$, its \temph{writers counterpart} $\tau!$ is the writers substitution defined by $X\tau! = X?\tau$. The \temph{pair completion} of a readers substitution $\sigma$ is $\sigma^\star = \sigma \cup \sigma!$, applied to a fixed point.
\end{definition}

\begin{definition}[Renaming, Renaming Apart, Writer MGU~\cite{shapiro2025glp}]
\label{def:glp-renaming}
\label{def:writer-mgu}
A \temph{renaming} is an injective substitution $\rho: \hat\calV \to \hat\calV$ such that for each $X \in \calV$: $X\rho \in \calV$ and $X?\rho = (X\rho)?$. Two terms \temph{have a variable in common} if for some writer $X \in \calV$, either $X$ or $X?$ occurs in both. A renaming $\sigma$ \temph{renames $T'$ apart from} $T$ if $T'\sigma$ and $T$ have no variable in common.
Given two unit goals $A$ and $H$, a \temph{writer mgu} is a writers substitution $\sigma$ such that $A\sigma = H\sigma$ and $\sigma$ is most general among such substitutions.
\end{definition}

\begin{definition}[Reduction~\cite{shapiro2025glp}]
\label{def:glp-reduction}
Given a unit goal $A$ and an ordinary clause $C$, with $H$ \verb|:-| $G$ \verb"|" $B$ the result of renaming $C$ apart from $A$, the \temph{reduction} of $A$ with $C$ \temph{succeeds with result} $(B,\sigma)$ if $A$ and $H$ have a writer mgu $\sigma$ and $G\sigma?$ succeeds.
\end{definition}

Guards have three-valued semantics~\cite{shapiro2025glp}. Each guard predicate explicitly defines its \emph{success} condition. A guard \emph{suspends} if it does not succeed but some instance of it under a readers substitution would succeed. A guard \emph{fails} if no such instance exists. A guard conjunction succeeds if all members succeed; it suspends if any member suspends and none fail; it fails if any member fails.  Guard occurrences count toward SRSW satisfaction, and a guard whose success implies that $X?$ is ground licenses multiple occurrences of $X$ and $X?$ in the clause~\cite{shapiro2025glp}.
\fi

\subsection{Multiagent Semantics}

The operational semantics of GLP is maGLP.  Each agent's initial resolvent is the atomic goal \verb|agent(ch(_?,_),ch(_?,_))|, providing two bidirectional channels: the first to the person operating the machine, the second to the network~\cite{shapiro2025glp}.  Anonymous variables make the initial local states syntactically identical yet semantically unshared; variable pairs whose writer and reader are held by different agents are the communication channels, and Cold-call bootstraps shared variables between previously-disconnected agents through the network streams~\cite{shapiro2025glp}.

\begin{definition}[Goal Identity~\cite{shapiro2025glp}]
\label{def:goal-identity}
Each unit goal in a computation carries a unique \temph{identifier} assigned at spawn time: goals in an initial goal receive distinct initial identifiers, and goals spawned by a Reduce transition receive fresh identifiers. Communicate transitions preserve identifiers: the goal containing $X?$ retains its identity after instantiation.
\end{definition}

\begin{definition}[Asynchronous Resolvent~\cite{shapiro2025glp}]\label{def:resolvent}
Given a GLP program $M$, an \temph{asynchronous resolvent} over $M$ is a pair $(G, \sigma)$ where $G \in \hat\calG(M)$ and $\sigma$ is a readers substitution.
\end{definition}

Each rule below is read as in an operational semantics: above the line its condition, below it the transaction, from its participants' local states to their states after it, every other agent stationary; the prose before each rule says what the premise names.

\begin{definition}[maGLP Transactions~\cite{shapiro2025glp}]\label{def:maglp-rules}
Given a GLP program $M$ and agents $P\subset \Pi$, the \temph{maGLP transactions} over asynchronous resolvents over $M$ are the following.
\begin{description}
\item[Reduce $p$.]  $A \in G_p$ is a unit goal, and $C \in M$ is the first clause for which the reduction of $A$ with $C$ (Definition~\ref{def:glp-reduction}) succeeds, with result $(B, \hat\sigma)$.
\[
\frac{A \in G_p \qquad C \text{ the first clause reducing } A, \text{ with result } (B, \hat\sigma)}
     {(G_p,\, \sigma_p) \;\rightarrow\; ((G_p \setminus \{A\} \cup B)\hat\sigma,\; \sigma_p \circ \hat\sigma?)}
\]
\item[Communicate $p$ to $q$.]  Participants $p, q \in P$, possibly $p = q$: $\{X? := T\} \in \sigma_p$ and $X?$ occurs in $G_q$.
\[
\frac{\{X? := T\} \in \sigma_p \qquad X? \text{ occurs in } G_q}
     {(G_p,\, \sigma_p),\ (G_q,\, \sigma_q) \;\rightarrow\; (G_p,\, \sigma_p \setminus \{X? := T\}),\ (G_q\{X? := T\},\, \sigma_q)}
\]
Where $p = q$ the transaction is unary: $(G_p,\, \sigma_p) \rightarrow (G_p\{X? := T\},\, \sigma_p \setminus \{X? := T\})$.
\item[Cold-call $p$ to $q$.]  Participants $p \ne q \in P$: the network output stream in $(G_p, \sigma_p)$ has a new message \verb|msg|$(q,X)$; $(G'_p, \sigma'_p)$ is the result of advancing the network output stream and $(G'_q, \sigma'_q)$ the result of adding \verb|msg|$(q,X)$ to the network input stream.
\[
\frac{\text{the network output stream of } p \text{ has a new } \mathtt{msg}(q,X)}
     {(G_p,\, \sigma_p),\ (G_q,\, \sigma_q) \;\rightarrow\; (G'_p,\, \sigma'_p),\ (G'_q,\, \sigma'_q)}
\]
\end{description}
\end{definition}

\begin{definition}[maGLP Transition System~\cite{shapiro2025glp}]\label{def:maglp-ts}
Given a GLP program $M$ and agents $P\subset \Pi$, the \temph{maGLP transition system} over $P$ and $M$ is the transactions-based multiagent transition system (Definition~\ref{def:tbmts}) over $P$, local states being asynchronous resolvents, initial local state $(\{a_0\}, \emptyset)$ with $a_0$ the goal \verb|agent(ch(_?,_),ch(_?,_))|, the transactions of Definition~\ref{def:maglp-rules}, and the equivalence ${\sim}$ relating $t_1 \sim t_2$ iff both are Reduce transactions at the same agent $p$ reducing the same goal (by identity, Definition~\ref{def:goal-identity}) with the same clause, both are Communicate transactions from the same $p$ to the same $q$ applying the counterpart of the same writer assignment to the same goal, or both are Cold-call transactions from the same $p$ to the same $q$ delivering the same message.
\end{definition}

Reduce differs from LP in the use of a writer mgu instead of a regular mgu, and in the choice of the first applicable clause instead of any clause; Communicate communicates an assignment from its writer to its paired reader --- within one agent, or across the two agents holding the pair's two sides; Cold-call transfers a term with its variables through the network streams, creating paired variables between previously-disconnected agents~\cite{shapiro2025glp}.  Distinct transactions have disjoint $P$-closures --- every participant of each changes state --- so ${\sim}{\uparrow}P$ is an equivalence (Definition~\ref{def:closure}).  Implementation-wise, if a goal $A$ cannot be reduced now, but there is a readers substitution under which it can, the goal \emph{suspends} on these readers and is rescheduled once any of them is assigned.

\begin{definition}[Proper Run~\cite{shapiro2025glp}]\label{def:glp-proper-run}
A maGLP run is \temph{proper} if, at each agent $p$, a variable that occurs in $G_p$ after a transition but not before it also does not occur in $G_p$ at any earlier configuration.
\end{definition}

\begin{definition}[cGLP~\cite{shapiro2025glp}]\label{def:cglp}
Given a GLP program $M$, \temph{cGLP$(M)$} is the maGLP transition system over a singleton set of agents and $M$.
\end{definition}

Cold-call requires two agents, so cGLP has Reduce and Communicate only; it is the single-agent semantics.

\begin{proposition}[Monotonicity~\cite{shapiro2025glp}]
\label{prop:glp-monotonicity}
In any proper run of cGLP$(M)$, if unit goal $A$ can reduce with clause $C$ at step $i$, then either an instance of $A$ has been reduced by step $j > i$, or an instance of $A$ can still reduce with $C$ at step $j$.
\end{proposition}

\section{Volition-Guarded GLP}\label{sec:vglp}

Here we define volition-guarded GLP (vGLP), the extension of GLP (Section~\ref{sec:glp}) with volitional procedures: a procedure may put a question to the person.  Its goal receives an \emph{interactive variable} as an added argument, which the program and the person write together, each the positions its type gives it, and a clause that reads the person's part reduces only after the person has answered.  vGLP is typed as GLP is: the declaration of a volitional procedure carries its interactive type, and widget declarations say how interactive types are displayed (Section~\ref{sec:elicitation}).

\subsection{Volitional Procedures}

\begin{definition}[Volitional Procedure, Interactive Type]\label{def:volitional-procedure}
A \temph{volitional procedure} $p$ of arity $n$ is declared \verb|procedure (|$T$\verb|)*|$p(T_1,\ldots,T_n)$\verb|.|, where $T_1,\ldots,T_n$ are the types of its arguments and $T$, in writer or reader mode as an argument type is, is its \temph{interactive type}.  An interactive type is \temph{flat}: it is built from constants, \verb|String|, \verb|Integer|, \verb|Real|, \verb|Key|, \verb|Constant|, \verb|Module|, named types so built, tuples and unions, and \verb|Stream| or \verb|List| as the whole type or as an argument of a tuple, never under another \verb|Stream| or \verb|List|; it has no other recursive type, and no type parameter at a position the person writes.
\end{definition}

\begin{definition}[Clause of a Volitional Procedure, Interactive Term]\label{def:volitional-clause}
A clause of a volitional procedure $p$ with interactive type $T$ has the form \verb|(|$A$\verb|)*|$p(S_1,\ldots,S_n)$ \verb|:-| $G$ \verb"|" $B$, where the \temph{interactive term} $A$ is a term of type $T$, possibly a variable but not the anonymous variable, and it is the guarded clause $p(S_1,\ldots,S_n,A)$ \verb|:-| $G$ \verb"|" $B$ (Definition~\ref{def:glp-program}), of arity $n{+}1$.  A clause of a volitional procedure reads an argument at an input position, by a non-variable term at that position in its head or by a guard on it, only if every later clause of the procedure reads it too, and every call of a volitional procedure, in a clause body or in the initial goal, holds a writer at each output position.
\end{definition}

\begin{definition}[Ordinary Clause, vGLP Program]\label{def:vglp-program}
A clause is \temph{ordinary} if it is not a clause of a volitional procedure.  A \temph{vGLP program} is a finite sequence of ordinary clauses and clauses of volitional procedures, each satisfying SRSW, a clause of a volitional procedure satisfying it as a guarded clause of arity $n{+}1$; the clauses for the same predicate form a \temph{procedure}.  The set of goals $\hat\calG(M)$ of a vGLP program $M$ includes all goals over $\hat\calV$ and the vocabulary of $M$ that satisfy SO.
\end{definition}

The last argument of the guarded clause is matched against the goal's end of its interactive variable, and a goal of a volitional procedure of arity $n$ is $n$-ary until it is asked and of arity $n{+}1$ after (Definition~\ref{def:vglp-ts}).

Here are vGLP fragments of the currency agent of the repository (Section~\ref{sec:platforms}), written by hand in this syntax, in which \verb|...| elides arguments the clause carries through unchanged.  The agent is a volitional procedure whose question is the stream of the person's requests, of a union type, an alternative per request, served by an ordinary procedure:
\begin{verbatim}
Friend ::= Key.
Amount ::= Integer.
Request ::= mint(Amount) ; pay(Friend, Amount) ; ... .
procedure (Stream(Request)?)*agent(...).

%% Ask the person for their requests, once, and serve them.
(Reqs)*agent(Id, ..., Holdings, Serial) :-
    serve(Reqs?, Id?, ..., Holdings?, Serial?).

%% Pay_{p,q}(k), guarded by p: k q-coins to q.
serve([pay(Friend, Amount)|Reqs], Id, ..., Holdings,
      Serial) :-
    ground(Id?), ground(Friend?), integer(Amount?),
    Amount? >= 1 |
    select_coins(Friend?, Amount?, Holdings?, Status,
                 Give, Rest),
    do_pay(Status?, Id?, Friend?, Amount?, Give?, Rest?,
           Reqs?, ..., Serial?).
\end{verbatim}
\iffull\begin{verbatim}
%% A peer proposes a swap: spawn a responder, which asks
%% the person, and go on without waiting for a request.
on_act(swap(Want, Offered, Resp?), From, Id, Reqs,
       UserIn, ...) :-
    ground(Want?), ground(Offered?) |
    respond_swap(offer(From?, Want?, Offered?), Resp, Ans),
    merge(Ans?, UserIn?, UserIn1),
    serve(Reqs?, Id?, UserIn1?, ...).

%% yes accepts, no declines; the agent hears of either
%% on its own input.
decide_swap(yes, From, Want, Offered, Resp?, Ans?) :-
    accept_answer(From?, Want?, Offered?, Resp, Ans).
decide_swap(no, From, _, Offered, Resp?, Ans?) :-
    decline_answer(From?, Offered?, Resp, Ans).
\end{verbatim}
\fi
The agent asks its question once and passes the stream to \verb|serve|, ordinary GLP; the person answers one request at a time, and the clause of that request's alternative reads it at the head of the stream: the pay clause waits for a \verb|pay| request.  The clause of an arriving swap proposal\full{}{ (Appendix~\ref{app:vglp})} reads no request, so it reduces without waiting and passes the stream on.  The responder's question, \verb|respond_swap| of Section~\ref{sec:introduction}, carries output as well as input: its interactive type is in writer mode, so its clause does not wait but writes the card at once, with the proposal at the position \verb|Offer| for the person to see, and hands the person the writer at the position \verb|YesNo?|; the clauses of \verb|decide_swap| each wait for the person's answer, and a decline is an answer, \verb|no|, taken by its clause.  A clause whose interactive term is a variable reduces without waiting and passes the reader on, so the question stays open until the person answers it\full{:}{; the chat of Appendix~\ref{app:vglp} is one, its question a stream the person gives one line at a time, with a status the program writes into each line once sent.}
\iffull
\begin{verbatim}
Status ::= sent.
Line ::= line(String, Status?).
Msg ::= msg(Peer, String).
procedure (Stream(Line)?)*chat(Peer?, Stream(Msg)).
(Ls)*chat(Peer, Out?) :- ground(Peer?) |
    send_all(Peer?, Ls?, Out).

send_all(Peer, [line(Text, sent)|Ls], [msg(Peer?, Text?)|Out?]) :-
    ground(Peer?) | send_all(Peer?, Ls?, Out).
send_all(_, [], []).
\end{verbatim}
The chat input is one question whose answer is a stream, which the person gives one line at a time, and \verb|send_all| is ordinary GLP, calling itself on the tail of the stream after each line.  Each line has a position the program writes, its status: the person's answer hands the program its writer (Answer, Definition~\ref{def:vglp-ts}), and \verb|send_all| writes \verb|sent| into it, by head unification, as it puts the message on the network, so the line the person typed shows as pending until it is sent and as sent after.  A question whose type has no position the program writes is bare, as the agent's request stream is; one whose type holds a stream in each direction is continuous, as a game is, the program drawing the board on one stream while the person moves on the other.
\fi

\subsection{Operational Semantics}

Here we define vmaGLP, the operational semantics of vGLP, which extends maGLP (Definition~\ref{def:maglp-ts}) with the person of each agent: an agent's state adds the variables its person holds, and three transactions, each at one agent, add the person's part.  What vGLP adds is thus local; the semantics is multiagent only as maGLP is, by Communicate and Cold-call, and Communicate delivers the program's output to the person as it delivers an assignment to a goal.

\begin{definition}[vmaGLP Agent State]\label{def:vglp-agent-state}
Given a vGLP program $M$, an \temph{agent state} is a triple $(H_p, G_p, \sigma_p)$ of a finite set $H_p$ of typed variables \temph{held} by the agent's person and of assignments \temph{shown} to them and an asynchronous resolvent $(G_p, \sigma_p)$ over $M$ (Definition~\ref{def:maglp-ts}).  A shown assignment is a record of what the person was shown, and its variables are not occurrences.
\end{definition}

\begin{definition}[Reduce and Communicate in vmaGLP]\label{def:vglp-reduce}
Reduce applies to a unit goal that is not an $n$-ary goal of a volitional procedure, and a clause of a volitional procedure is taken as its guarded clause (Definition~\ref{def:volitional-clause}).  Communicate delivers an assignment $\{X? := T\}\in\sigma_p$ to a reader $X?\in H_q$ as to one occurring in $G_q$, then removing $X?$ from $H_q$ and adding the variables of $T$, each with the type of its position in $T$, and the assignment, shown; Reduce, Cold-call and every other Communicate leave the held variables unchanged.
\end{definition}

\begin{definition}[Ask]\label{def:vglp-rules}
The \temph{transactions on the person's variables} at an agent $p$ are \temph{Ask} of this definition and \temph{Answer} and \temph{Dismiss} of the next two.  Ask $p$: $A$ is an $n$-ary goal of a volitional procedure $q$ with interactive type $T$, $X$ a writer occurring nowhere in the configuration, the \temph{interactive variable} of the asked goal, of type $T'$ where $T$ is $T'$ or $T'?$; the goal takes its end of $X$ and the person the other: $(A', Y) = (q(S_1,\ldots,S_n,X?), X)$ where $T$ is in reader mode, and $(q(S_1,\ldots,S_n,X), X?)$ where it is in writer mode.  The asked goal keeps the identifier of $A$.
\[
\frac{A = q(S_1,\ldots,S_n) \in G_p \qquad X \text{ fresh} \qquad (A', Y) \text{ as above}}
     {(H_p,\, G_p,\, \sigma_p) \;\rightarrow\; (H_p \cup \{Y\},\; G_p \setminus \{A\} \cup \{A'\},\; \sigma_p)}
\]
\end{definition}

\begin{definition}[Answer]\label{def:vglp-answer}
Answer $p$: $X \in H_p$ is a writer whose reader $X?$ occurs in the configuration, and $t$ a term of the type of $X$, ground but for variables occurring nowhere in the configuration: a reader $Y_i?$ at each position the person writes of type \verb|Stream(|$T'$\verb|)| or \verb|List(|$T'$\verb|)|, unless the person ends the list there with \verb|[]|, and a writer $W_j$ at each outermost position the program writes; where $X$ itself is of such a type, $t$ is \verb|[|$e$\verb#|#$Y_1?$\verb|]| with $e$ such a term of type $T'$, or \verb|[]| for a list.  The person holds $Y_i$ and $W_j?$, each with the type of its position in $t$.
\[
\frac{X \in H_p \qquad X? \text{ occurs} \qquad t \text{ as above}}
     {(H_p,\, G_p,\, \sigma_p) \;\rightarrow\; (H_p \setminus \{X\} \cup \{Y_1,\ldots,Y_k,W_1?,\ldots,W_m?\},\; G_p,\; \sigma_p \cup \{X? := t\})}
\]
\end{definition}

\begin{definition}[Dismiss]\label{def:vglp-dismiss}
Dismiss $p$: $X \in H_p$ is a writer whose reader $X?$ occurs nowhere in the configuration.
\[
\frac{X \in H_p \qquad X? \text{ occurs nowhere}}
     {(H_p,\, G_p,\, \sigma_p) \;\rightarrow\; (H_p \setminus \{X\},\; G_p,\; \sigma_p)}
\]
\end{definition}

\begin{definition}[vmaGLP Transition System]\label{def:vglp-ts}
Given a vGLP program $M$ and agents $P\subset \Pi$, the \temph{vmaGLP transition system} over $P$ and $M$ is the transactions-based multiagent transition system (Definition~\ref{def:tbmts}) over $P$, local states being agent states, initial local state $(\emptyset, \{a_0\}, \emptyset)$ with $a_0$ as in Definition~\ref{def:maglp-ts}, the transactions being the Reduce, Communicate and Cold-call transactions of maGLP as Definition~\ref{def:vglp-reduce} adapts them and the transactions on the person's variables of Definition~\ref{def:vglp-rules}, read as those of Definition~\ref{def:maglp-rules}, and the equivalence ${\sim}$ below.
${\sim}$ is the partial equivalence whose domain is the Ask, Reduce, Communicate and Cold-call transactions, relating two Reduce, Communicate or Cold-call transactions as maGLP does, the person counting as a goal, and $t_1 \sim t_2$ iff both are Ask transactions at the same agent $p$ asking the same goal (by identity, Definition~\ref{def:goal-identity}).
\end{definition}

Distinct transactions have disjoint $P$-closures, since every participant of each changes state, so ${\sim}{\uparrow}P$ is an equivalence (Definition~\ref{def:closure}).  Ask is the pre-reduction of a goal of a volitional procedure.  Before the goal can reduce, it receives its end of its interactive variable as its added argument: the reader where the person writes the variable, and the writer where the program does.  The person holds the other end.  The semantics fixes no moment for Ask other than before the goal's reduction; the implementation asks at the goal's first reduction attempt (Section~\ref{sec:elicitation}).

Answer is the person's act: the program reads an answer only through its reader, and the person performs every answer.  Answer and Dismiss lie outside the domain of the equivalence and carry no liveness obligation, since the machine cannot compel the person.  An answer is an assignment and cannot be withdrawn; the person edits it in its construct before submitting it (Section~\ref{sec:elicitation}).

An interactive type carries output to the person as well as input from the person.  Its moded type~\cite{shapiro2026types} says which positions the program writes and which the person writes, and each side, writing a position, hands the other the variables of the positions below it that the other writes.  Communicate delivers the program's output to the person, with the writers inside it, from the agent whose program wrote it, and keeps it in $H_q$ as shown.  So the person holds the writer of an input only after the output containing it is delivered: every output comes before the inputs inside it, and a question can show its person what they are asked about (Section~\ref{sec:elicitation}).  An answer is ground at the positions the person writes, but for a stream or a list.  An answer of type \verb|Stream(|$T$\verb|)| is given one element at a time, each element leaving the person the writer of the stream's tail as a new question, and an answer of type \verb|List(|$T$\verb|)| likewise until the person ends it.  A stream or list inside an answer is left to the person as a new question, answered element by element.

\begin{proposition}[Conservativity]\label{thm:conservativity}
For a GLP program $M$, the runs of the vmaGLP transition system over $P$ and $M$, with $(\emptyset, G_p, \sigma_p)$ identified with $(G_p,\sigma_p)$, are the runs of the maGLP transition system over $P$ and $M$, and the same runs are correct in both.
\end{proposition}
\begin{proof}
$M$ has no volitional procedure, so in a run no goal is asked and the person holds no variable: Ask, Answer and Dismiss have no instances and every $H_p$ is $\emptyset$, so no reader is held.  Reduce, Communicate, and Cold-call are those of maGLP with $H_p$ unchanged, and ${\sim}$ is its equivalence, total on them.
\end{proof}

\begin{definition}[Proper Run, adapted from~\cite{shapiro2025glp}]
\label{def:proper-run}
A vmaGLP run is \temph{proper} if, at each agent $p$, a variable that occurs in $G_p$ after a transition but not before it also does not occur in $G_p$ or in $H_p$ at any earlier configuration, and the clause each Reduce renames is renamed apart from every agent's goals, held variables and readers substitution, in the configuration it reduces and in every earlier one.
\end{definition}

\begin{definition}[vcGLP]\label{def:vcglp}
Given a vGLP program $M$, \temph{vcGLP$(M)$} is the vmaGLP transition system over a singleton set of agents and $M$.
\end{definition}

Cold-call requires two agents, so vcGLP has Reduce, Communicate, Ask, Answer and Dismiss only; it is the single-agent semantics.

\begin{proposition}[Monotonicity in vcGLP]
\label{prop:vglp-monotonicity}
In any proper run of vcGLP$(M)$, if unit goal $A$ can reduce with clause $C$ at step $i$, then either an instance of $A$ has been reduced by step $j > i$, or an instance of $A$ can still reduce with $C$ at step $j$.
\end{proposition}

\iffull
\iffull\begin{proof}\else\begin{proof}[Proof of Proposition~\ref{prop:vglp-monotonicity}]\fi
In vcGLP a goal of a volitional procedure reduces only once asked, as a GLP goal with the guarded clauses of its procedure, and an answer reaches it as a communicated assignment, so Proposition~\ref{prop:glp-monotonicity} holds verbatim.
\end{proof}
\else
Monotonicity holds in vcGLP as in GLP (Appendix~\ref{app:vglp}).
\fi

\begin{definition}[Question, Open, Closed, Volition]\label{def:question}
\label{def:volition}
In a configuration of the vmaGLP transition system, a writer $X \in H_p$ is a \temph{question} of $p$'s person: an interactive variable the person writes, introduced by an Ask, a writer inside the program's output delivered to the person, or the writer of a stream or list that an Answer leaves to the person, the tail of the stream or list answered or one inside the answer; it is \temph{open} if its reader $X?$ occurs in the configuration, and \temph{closed} otherwise.  A \temph{volition} of $p$'s person is an Answer transition of $p$: the person's answer to a question.
\end{definition}

An open question is what the machine puts to its person: \emph{what is your answer?}  A question is answered by Answer, which takes it from the person's variables; it closes when its reader no longer occurs, the program having dropped the reader by the anonymous variable at an input position~\cite{shapiro2025glp}; in writer mode the program so closes a question inside its output and, where the type provides for it, writes more of the output --- a status that the widget reads, say.  A closed question takes no answer and asks nothing more of the person; it leaves their questions when they dismiss it (Dismiss), and not otherwise.  How questions reach the screen and answers return is the interface's task (Section~\ref{sec:elicitation}).

\subsection{Volitional Soundness}

The machine never answers for the person: a question is assigned by its person's answer and by no transition of the machine.

\begin{restatable}[Volitional Soundness]{theorem}{restateVolitionalSoundness}\label{thm:volitional-soundness}
Let $c^1 \rightarrow c^2 \rightarrow \cdots$ be a proper safe run of the vmaGLP transition system over $P$ and $M$, $p \in P$, and $X$ a question of $p$'s person in $c^i$.  Then the writer $X$ occurs in $c^i$ in $H_p$ and nowhere else, and a transition $c^i \rightarrow c^{i+1}$ that adds an assignment $X? := t$ to a readers substitution is the Answer of $p$ to $X$.
\end{restatable}

The proof is in Appendix~\ref{app:proofs}.  Every reduction that reads an answer therefore reads a term the person gave.  The theorem says who gives an answer.  A program has its person see a whole proposal before answering by writing the proposal before it hands over the answer's writer, as the responder's guards do.  Once answered, a goal is a GLP goal, and a correct run reduces it (Proposition~\ref{prop:answered}).

\begin{proposition}[Answered Goals Reduce]\label{prop:answered}
Let $r$ be a proper correct run of vcGLP$(M)$ and $A$ an asked goal in it whose question the person answers with a term on which some clause of $A$'s procedure reduces.  Then an instance of $A$ is reduced in $r$.
\end{proposition}
\iffull
\iffull\begin{proof}\else\begin{proof}[Proof of Proposition~\ref{prop:answered}]\fi
The answer is an assignment $X? := t$ in the readers substitution with $X?$ in $A$; the Communicate of it to $A$ is in the domain of the equivalence and enabled while the assignment stands, so $r$ takes it, and $A$ then holds $t$.  From there $A$ is a GLP goal with the guarded clauses of its procedure, and a clause reduces it on $t$; by Monotonicity (Proposition~\ref{prop:vglp-monotonicity}) that clause goes on succeeding until an instance of $A$ is reduced, so a Reduce class of $A$ is enabled in every configuration of a suffix of $r$, and $r$ takes it, with that clause or an earlier one.
\end{proof}

The person's requests are served in the same way: in the agent of Section~\ref{sec:vglp}, a request at the head of the stream is read by the clause of its alternative, which comes before the message clause, so the next reduction of \verb|serve| takes it; a clause placed before the request clauses that reduces without reading the stream could put a request off while its trigger keeps arriving, and the agent goals of Section~\ref{sec:cva-implementation} place none, reading every trigger from one fair merge.
\else
The proof is in Appendix~\ref{app:proofs}, with why the person's requests are served the same way: the clause of a request's alternative comes before the message clause, so the next reduction of \verb|serve| takes a request at the head of the stream.
\fi

\subsection{Compiling the Canonical CVA Platform into vGLP}\label{sec:cva-implementation}

A CVA platform specifies a set of app transactions (Definition~\ref{def:cva-platform}), possibly infinite and written in no syntax, and no program implements every such set.  Here we compile into vGLP the canonical platform of Section~\ref{sec:decomposition} of every effective local specification with simple guards, one goal per agent, and prove the compilation correct.

\begin{definition}[Effective Local Specification]\label{def:effective}
A local volition-guarded specification is \temph{effective} if its machine states are finite terms, its guarded classes, which enter the request type and the cards, are terms of a flat type (Definition~\ref{def:volitional-procedure}) in which no \verb|Stream| or \verb|List| occurs, and its unguarded classes are finite terms; $Q([t])$, $\calQ([t])$ and, for an unguarded class, $o([t])$ are computable from $[t]$; $f^{[t]}_p(s)$ is computable from $[t]$, $p$ and $s$, and whether $s\in\mathit{dom}(f^{[t]}_p)$ is decidable from them; and for every agent $p$ and machine state $s$, the unguarded classes $[t]$ with $o([t])=p$ and $s\in\mathit{dom}(f^{[t]}_p)$ are finitely many and computable from $p$ and $s$.
\end{definition}

The conditions and effects of the canonical platform of an effective specification are computable and terminate: they test and apply the specification's functions and the guard sets of a class.  A name of a transaction is of a flat type: a guarded class is a term of a flat type, and an attempt of an unguarded class, whose class may hold a list, as a delivery's holds its members, is named by its initiator and number.

\mypara{Triggers}  The compiled agent takes each app transaction of the canonical platform on a \emph{trigger}: \textbf{Propose} of a guarded class on its person's request \verb|propose(|$[t]$\verb|)|, which, with the agent's number, determines its attempt; \textbf{Start}, \textbf{Tally}, \textbf{Yield} and \textbf{Outcome} on the arriving message; \textbf{Answer} on its person's answer to the \emph{card} that \textbf{Start} raises, one card per attempt, a goal whose question shows the attempt and takes yes or no; and \textbf{Propose} of an unguarded class, by its opener, \textbf{Vote}, \textbf{Commit} and \textbf{Wound}, on a tick of the agent's clock.  So the person answers two kinds of question: the agent's \emph{request question}, a stream question whose elements are the person's requests, and the cards.  A request is taken or dropped when it is read, and is not a standing will, since an answer cannot be withdrawn (Section~\ref{sec:vglp}).  The else-branch of \textbf{Answer} is taken on the person's answer no, where CVA takes it on the absence of the will, so a card waits for its person.

\iffull
\mypara{The compilation}  Given an effective local volition-guarded specification with simple guards and its canonical platform $\Theta$, the \emph{compilation} $\lceil\Theta\rceil$ is the following vGLP program, standard stream techniques used without comment~\cite{shapiro2025glp}.  The \emph{agent} is a volitional procedure whose question is a pair: the stream of its person's requests, of a type with the alternative \verb|propose(|$[t]$\verb|)| for each guarded class $[t]$ and the alternative \verb|known(|$Q$\verb|)|, and the \emph{thread}, a stream the program writes, on which the agent goal writes the name of each app transaction it takes and \verb|discover(|$Q$\verb|)| for each peer $Q$ it comes to know.  The agent asks its question once, as \verb|agent| in Section~\ref{sec:vglp} asks its request stream, and calls the agent goal.

The \emph{agent goal} is an ordinary procedure holding the \emph{agent state} --- the app state, initially $(s_0,0,\bot,\emptyset)$; a count of ticks, initially $0$; the peers \verb|hello| has been sent to and received from; the starts it holds waiting; a writer for each open card, by which it closes it; and the tick since which each class its ticks have met has waited --- and reading one input stream, the fair merge (Example~\ref{ex:merge}) of the requests, the messages arriving on the network, the answers to its cards and the ticks of a \emph{clock goal}, which writes ticks without end; after reading any other trigger it takes a tick of its own.

The agent goal has one clause per kind of trigger, which reads the trigger at the head of the stream, computes in its body the condition of the transaction the trigger puts up, on the state it was called with, and calls a \emph{dispatch procedure} with its value, each of whose clauses calls the agent goal again.  A dispatch clause that \emph{takes} a transaction carries out its effect: the new app state; the messages it sends, appended to the network output stream; for a \textbf{Start} that records a question, a \emph{card}; and the closing of the open card of an attempt the transaction closes; and it writes the transaction's name on the thread.  A card is a volitional procedure whose question shows the attempt and takes yes or no; it also takes the reader of the writer by which the agent goal closes it, and its one clause hands that reader and the reader of the answer to an ordinary goal, which passes the answer on where it comes first and, where that reader is assigned first, drops the answer's reader, which closes the question (Section~\ref{sec:vglp}).

On the request \verb|propose(|$[t]$\verb|)| the dispatch takes \textbf{Propose} where its condition holds, and drops the request otherwise.  On a message $\mathsf{start}(\alpha)$ it takes \textbf{Start} where the start does not wait (Definition~\ref{def:waiting}) or $\alpha$ is closed at the agent, and otherwise holds it waiting; on $\mathsf{vote}(\alpha,y)$ it takes \textbf{Tally}, on $\mathsf{yield}(\alpha)$ \textbf{Yield}, and on $\mathsf{outcome}(\alpha,z)$ \textbf{Outcome}.  On the answer to the card of $\alpha$ it takes the then-branch of \textbf{Answer} on yes and its else-branch on no, where $\alpha$ is a question pending at the agent and unanswered, and drops the answer otherwise.

On a tick the agent goal counts it and puts up the \textbf{Propose} of each unguarded class the agent opens at its machine state, the \textbf{Vote} on each answered question pending at it, the \textbf{Commit} of each attempt it initiated with a yes of every other participant recorded, and the \textbf{Wound} of each start it holds waiting, or answered question, blocked by an attempt over which it has priority; the dispatch takes, among those whose condition holds, one of the class that has waited longest --- since it was last taken or, never taken, since a tick first found its condition holding --- ties broken in a fixed order, and drops the tick where none holds.  After every dispatch that changes the app state or consumes a start the agent goal puts up again the starts it holds waiting.

On the request \verb|known(|$Q$\verb|)| the agent goal sends $Q$ \verb|hello| and records $Q$ as sent to, unless already sent; on an arriving \verb|hello| from $Q$ it records $Q$ as received from; and on whichever of the two comes second it writes \verb|discover(|$Q$\verb|)| on the thread.  A condition's test that a peer is known is of the peers sent to and received from.
\fi

The reduction of a dispatch clause that takes a transaction, or the else-branch of \textbf{Answer}, is that transaction.  Discover is the sending of the second of the two hellos, and Communicate is the agent goal's reading of a message the network streams carried to it.  Every other reduction is a stutter.

\iffull
\begin{definition}[Platform Mapping]\label{def:platform-mapping}
Given an effective local volition-guarded specification with simple guards, its canonical platform $\Theta$, $P \subset \Pi$ and a configuration of the vmaGLP transition system over $P$ and $\lceil\Theta\rceil$, the \temph{platform mapping} $\sigma$ yields the configuration of $\calP$ with, at each agent $p$: $\mathit{known}_p$ the $q \in P$ such that $p$ and $q$ have each sent the other \verb|hello|; $o_p$ the app messages on $p$'s network output stream not yet read by their recipient's agent goal; $i_p$ the app messages $p$'s agent goal has read and not yet dispatched and the starts it holds waiting; $a_p$ the app state the agent goal holds or, between the reduction of a dispatch clause and the agent goal's next call, the one its effect yields.
\end{definition}

\begin{definition}[Platform Mapping: Volitional State]\label{def:platform-mapping-volition}
At each agent $p$, the platform mapping $\sigma$ yields the volitional state holding $p$'s approval class on a guarded class $[t]$ while the agent goal reads a request \verb|propose(|$[t]$\verb|)|, or the answer yes to the card of an attempt of $[t]$ --- from its reading until its dispatch --- and the condition of the transaction it puts up, but for its test of known peers, holds at $p$'s machine state, and the class of $\mathrm{Discover}(p,q)$, $q \in P$, while $p$ has requested \verb|known(|$q$\verb|)| and $q\notin\mathit{known}_p$.
\end{definition}
\fi

\begin{restatable}[Compilation of the Canonical Platform]{theorem}{restatePlatformCompilation}\label{thm:platform-compilation}
For every effective local volition-guarded specification with simple guards and $P \subset \Pi$, the platform mapping is a correct implementation of $\calP$, the transition system of its canonical platform over $P$ under its fairness condition (Definition~\ref{def:fairness-cva}), by the vmaGLP transition system over $P$ and $\lceil\Theta\rceil$, under the fairness condition that a run is proper, no card's question stays open in it, and its image under the platform mapping is connected and unobstructed.
\end{restatable}

The proof is in Appendix~\ref{app:proofs}.  The condition on the person concerns the cards: the request question need not be answered, since a request not made wills nothing, but a card left open keeps \textbf{Answer} from being taken by either branch; a card whose attempt closes first is closed by the program.  The condition on the run is that of Theorem~\ref{thm:decomposition} but for strong fairness, which the compilation secures (Section~\ref{sec:cascade}).  Completeness is not claimed: the compilation takes \textbf{Vote} and an opener's \textbf{Propose} only on a tick, which $\calP$ does not require.

\mypara{Act schemas}  A contract written as act schemas~\cite{shapiro2026formalising} is the layer above the transactions in Figure~\ref{fig:cascade}.  A compiler from an act-schema contract to the canonical platform of Section~\ref{sec:decomposition}, written in GLP, prints the platform's compilation in vGLP, in the form of this section.  It compiles the contracts of the social graph and of the currencies, with coins of one unit, and the test suite runs each compiled platform on GLP with scripted persons.\full{}{  The conditions under which such a contract is a local and effective specification are in Appendix~\ref{app:proofs}.}

\iffull
\mypara{Act schemas as local specifications}  A schema names its \emph{roles}, the parties to the act, marks those whose person must will the act, and lists at each role the \emph{atoms} --- facts in that party's state, as $\mathit{friend}(q)$ in the state of $p$ --- that the act requires, forbids, deletes and adds there; a \emph{schema-binding} is a schema with its roles bound to people.  The classes are those of the contract's equivalence, generated by relating two transactions when one schema-binding determines both~\cite{shapiro2026formalising}; $f^{[t]}_p$ is defined at $d_p$ iff some schema-binding of the class has the atoms it requires at $p$ in $d_p$ and those it forbids there not, and yields $d_p$ less the deleted atoms plus the added ones.  Such a contract is a local specification under three conditions.  First, every two schema-bindings of one class delete and add the same atoms at each participant, so that the class has one effect at each participant.  Second, they require and forbid the same atoms at every participant but one, so that the precondition of the class, which holds where the requirements of some schema-binding of it hold, is a conjunction of conditions, each on the state of one participant.  Third, the schema-bindings determining any two transactions of one class mark the same sets of people, so that equivalent transactions have the same minimal guards, as Definition~\ref{def:vg-spec} requires.  The guard sets of a class are the minimal sets among the people in the marked roles of its schema-bindings, so that unfriending, whose schema determines the same transaction at the binding of either party first, is one class guarded by either party; and the opener of an unguarded class is the least, in the order of identifiers, of the people in the first role of its unguarded schema-bindings.  The contract is effective (Definition~\ref{def:effective}) where, moreover, the schema-bindings of each unguarded class require and forbid the same atoms at every participant, and the atoms each unguarded schema requires at its first role name every other role and every variable of its deleted and added atoms, so that a state, a finite set of atoms, admits finitely many classes of it: the schema's other variables occur in no deleted or added atom, and its bindings that differ in them determine one class.

\fi

\section{Elicitation: Implementing vGLP by GLP}\label{sec:elicitation}

Here we implement vGLP by GLP (Section~\ref{sec:glp}).  A question stays open until its person answers it or the program drops its reader (Definition~\ref{def:question}), and only the person's answer reaches the goal that asked it, so the implementation requires elicitation: offering the person, for each open question, the means to answer it with a term of its type, and showing them what the program writes for them to see.  We derive the elicitation constructs from the interactive types, and compile vGLP onto GLP.

\subsection{The Elicitation Construct}

A question's type fixes what the person is shown and what they may answer: the program writes the positions the type gives it and the person the others, and the program's output reaches the person before the inputs inside it (Definition~\ref{def:vglp-ts}).  A \emph{widget} renders a type: it shows what the program writes, takes the person's input, and forms from it a term of the type.  Widgets are associated with types by declarations (Section~\ref{sec:display}), and the constructs they render are standard~\cite{plasmeijer2007itasks,steenvoorden2019tophat}.

\begin{definition}[Construct, Submission, Complete Widget]\label{def:manifest}
The \temph{construct} of an open question $X$ at agent $p$ is the widget associated with the interactive type of the interactive variable that is $X$ or whose term holds $X$, presented at $p$ while $X$ is open.  It shows the assignments shown at $p$ (Definition~\ref{def:vglp-ts}) to the readers of the interactive variable's term, as they are shown, and forms from the person's input a term of the type of $X$; a \temph{submission} of it is such a term, which the person gives as their Answer to $X$.  Where $X$ is of type \verb|Stream(|$T'$\verb|)| or \verb|List(|$T'$\verb|)|, a submission gives one element and leaves the tail an open question with the same construct, or, for a list, ends it.  A widget of $T$ is \temph{complete} if, for every question whose construct it is, every answer the Answer transition allows for the question (Definition~\ref{def:vglp-ts}) is formed by some input.
\end{definition}

The person edits the input of a construct until they submit it; they give a submission as their Answer, and an answer cannot be withdrawn (Definition~\ref{def:vglp-ts}).

\begin{definition}[Offered]\label{def:offered}
A term $t$ is \temph{offered} for a question $X$ at $p$ if the construct of $X$ presented at $p$ forms $t$, as a submission, from some input.
\end{definition}

Two widgets recur, and we name questions after them, informally.  A question of a tuple type is a \emph{compose form}, whose fields are the tuple's arguments, and a question of a union of tuple types is a menu of such forms: the volition originates with the person, who supplies the transaction's parameters --- mint, pay, post, offer friendship.  A question whose type shows content the program writes and takes one of a few constants is an \emph{inbox card}, whose content is the offer and whose buttons are the constants: the responder's \verb|Card| of Section~\ref{sec:introduction}, with the proposal shown and the buttons yes and no.

\begin{restatable}[Elicitation Completeness]{theorem}{restateCompleteness}\label{thm:completeness}
In every configuration of a safe run of the vmaGLP transition system over $P$ and $M$ and every $p \in P$, a term offered for a question $X$ at $p$ is an answer to $X$ that the Answer transition allows, $X$ being open; and if the widget of the construct of $X$ is complete, every answer the Answer transition allows for an open question $X$ is offered.
\end{restatable}

The proof is in Appendix~\ref{app:proofs}.  The first half follows from typing: a widget forms only terms of its type, so the validation of an input is its typing.  The second is a requirement on widgets, which the default widgets below meet, the picker excepted (Proposition~\ref{prop:default-complete}).  A construct family draws the widgets (Remark~\ref{rem:families}); nothing below the widget is proved.

\subsection{Widgets}\label{sec:display}

User-interface design and specification languages are a large field~\cite{calvary2003cameleon,paterno2000model,abrams1999uiml,limbourg2004usixml,akiki2015survey}, and any of them could be brought to bear here.  For the purpose of this paper we present the basic types and their widgets, and the declaration that associates a widget with a type.

\begin{remark}[Construct Families]\label{rem:families}
A question's type fixes a construct's semantic content --- what it shows, what it takes, and the answer each submission forms --- and nothing of its layout.  A \emph{construct family} is a rendering of the widgets that preserves this content; the theorems of this section, Theorem~\ref{thm:completeness} and Theorem~\ref{thm:compilation} below, hold for every construct family whose widgets are complete, so families are interchangeable and each person may choose theirs --- the same program rendered plainly for one person and child-friendly for another.
\end{remark}

\begin{definition}[Widget Declaration]\label{def:display}
A \temph{widget declaration} has the form $T$ \verb|=::=| $W$, where $T$ is a moded type and $W$ a widget of the construct family that renders the program (Remark~\ref{rem:families}).  Widget declarations are scoped as type declarations are~\cite{shapiro2026types}: a declaration at the root holds for every program, one in a module holds in that module, and a local declaration overrides a global one.
\end{definition}

\begin{definition}[Default Widget]\label{def:default-widget}
The widget of a question's type is the one its declaration in scope names or, where none does, its \temph{default widget}, built from the type's constructors:
\begin{itemize}
\item a constant: a button; a union of constants: a row of buttons; a position of type \verb|Constant| or \verb|Module|: a text field;
\item \verb|String|, a number (\verb|Integer| or \verb|Real|) and \verb|Key|, the agent identifier: a text field, a number field and a peer field;
\item a tuple: a form headed by its functor, with one field per argument, each labelled by the name of its declared type; a stream or list inside a form the person writes: submitted as a separate question and drawn under the form; any other union: a menu, one item per alternative, each drawn by its widget;
\item a value the program writes: shown and not edited; a tuple of two arguments, a list the program writes and a position of its element type that the person writes: a picker in place of a form, offering the elements of the list;
\item a stream the person writes: the widget of its element type, kept open, one element per submission; a list the person writes: the same, with a button that ends the list; a stream the program writes: a thread.
\end{itemize}
A value the program writes to its person outside a question is drawn by the default widget of its declared type.
\end{definition}

The mode matters: the same \verb|Stream(String)| is an input box when the person writes it and a thread when the program writes it.  The widget holds the other end of the question, so its type is the dual of $T$: it writes where the program reads and reads where the program writes.  The values of a primitive type are those the runtime represents, and a field takes each of them.  Every interactive type has a default widget, so a program with no widget declarations still renders; the default widgets are complete for their types, except the picker, which offers the elements of its list:

\begin{proposition}[Default Widgets]\label{prop:default-complete}
The default widget of every interactive type is complete, but for a picker.
\end{proposition}
\begin{proof}
By induction on the type $T'$ of a question, which is flat (Definition~\ref{def:volitional-procedure}).  The answers the Answer transition allows for the question (Definition~\ref{def:vglp-ts}) are the terms of $T'$ ground at the positions the person writes, with a fresh reader at each stream and list position and a fresh writer at each outermost position the program writes, and, where $T'$ is \verb|Stream(|$T''$\verb|)| or \verb|List(|$T''$\verb|)|, the terms \verb|[|$e$\verb#|#$Y?$\verb|]| with $e$ such a term of $T''$, and \verb|[]| for a list.  A constant, or a union of constants: each answer is a constant, formed by its button.  \verb|String|, \verb|Integer|, \verb|Real|, \verb|Key|, \verb|Constant| and \verb|Module|: each answer is a value of the type, formed by its field.  A tuple: each answer is a tuple whose arguments are answers for their types at the positions the person writes, formed by the form's fields by induction, a fresh reader at a stream or list argument, which the form submits as a separate question, and a fresh writer at a position the program writes, which the form shows and does not edit; the form forms the tuple.  Any other union: the menu selects the alternative, and its widget forms the rest, by induction.  A stream the person writes: each answer \verb|[|$e$\verb#|#$Y?$\verb|]| has $e$ an answer for the element type, which is flat and holds no stream or list, formed by the element's widget, and the widget forms the fresh tail; a list likewise, and \verb|[]| by the button that ends it, whose grant is distinct from every element.  A named type is its definition.  A question inside the program's output is answered by the field of its position, the default widget of its type, by the same induction.  A picker forms only the elements of its list.
\end{proof}
  A declared widget chooses within the construct family, and Theorem~\ref{thm:completeness} holds for it where it is complete.  In the grassroots app (Section~\ref{sec:ui-primitives}), $W$ names a Dart widget: the construct family draws it, and its type fixes the terms it forms.

\subsection{The Implementation: Compiling vGLP onto GLP}
\label{sec:compilation}

vGLP needs no new machine: it compiles onto GLP, clause by clause, into a GLP program over an uninterpreted person channel, and Theorem~\ref{thm:compilation} relates the runs of the two.  GLP already connects the program to the person: each agent's first channel is to the person operating it; messages flowing toward the person are rendered by the runtime, and the person's responses flow back as GLP messages~\cite{shapiro2025glp,shapiro2026implementing}.  In the compiled program each goal of a volitional procedure, at its first reduction attempt, has the dispatcher spawn a \emph{construct process} for its question and receives its end of its interactive variable.  The \emph{bridge}, the Dart side of the runtime, draws each construct's widget and returns the person's acts on it, each act a \emph{grant}.  We first make the person's side formal.

\begin{definition}[Person Channel, Person Writer, GLP with Persons, Grant]\label{def:person-ts}
Let $N$ be a GLP program and $P\subset\Pi$.  A \temph{person channel} of a goal $G_0$ satisfying SO is a variable pair $(U, U?)$ such that $U?$ occurs in $G_0$ and $U$ does not; $U$ is its \temph{person writer}.  \temph{GLP with persons} over $P$ and $N$, with a designated person channel of each agent's initial resolvent, is the transition system obtained from the maGLP transition system over $P$ and $N$ (Definition~\ref{def:maglp-ts}) by adding, at each agent $p$, \temph{grant} transitions
\[
(G_p, \sigma_p) \;\rightarrow\; (G_p,\; \sigma_p \circ \{U? := [c\,|\,U_1?]\})
\]
where $c$ is a ground term \verb|input(|$\mathit{Id}$\verb|, |$\mathit{Pos}$\verb|, |$R$\verb|)|, $U$ the person writer of $p$, and $U_1$ a fresh variable, the person writer after the transition, every other agent stationary; the partial equivalence is the maGLP transaction equivalence, with grant transitions outside its domain.
\end{definition}

A grant is the person's act: a tap, a keystroke or a submission on a construct grants the term it carries, the construct $\mathit{Id}$, the position $\mathit{Pos}$ and the person's raw input $R$.  The person channel is typed by these declarations, at the root of every program:
\begin{verbatim}
PersonChannel(T) ::= Channel(PersonIn(T), PersonOut(T)).
PersonIn(T) ::= [Input(T) | PersonIn(T)].
PersonOut(T) ::= [Draw(T) | PersonOut(T)].
Input(T) ::= input(Integer, Path, T).
Path ::= [] ; [Integer | Path].
Draw(T) ::= draw(Integer, _, T) ; withdraw(Integer).
\end{verbatim}
Neither stream ends: the person and the program stay connected, and a goal standing for the person, as in a test, waits on the draws when it has nothing to grant.  The parameter \verb|T| is the type of the values drawn and of the person's raw inputs, which the compilation forms (Appendix~\ref{app:proofs}, \textbf{The construct process}).  Definition~\ref{def:proper-run} (proper run) applies to GLP with persons verbatim, with no held variables and grants leaving the goals unchanged, and properness is the fairness condition of GLP with persons: nothing is required of the person, since the machine owes nothing to a question its person leaves open.

The compilation is defined for every vGLP program $M$.  A goal of a volitional procedure $q$ with interactive type $T$, called with its $n$ arguments, first reduces with an \emph{asking clause}, which sends an \emph{ask} to the \emph{dispatcher} and calls the guarded clauses of $q$, under a fresh name, with its end of a new interactive variable of type $T$ added; this reduction is the Ask.  Asks travel on \emph{ask streams}: every procedure that may reach a question carries one, and a clause merges the ask streams of its body goals into its own.  On an ask the dispatcher spawns the \emph{construct process} of the question, which holds the person's end of the variable: the dispatcher routes the person's grants to it by its identifier; it has the bridge draw the widget of $T$ and the output the program writes into the variable; and it binds each question it holds with the term the person's input forms, one element per submission for a stream type.  The construct process withdraws when every question it holds is answered.  The dispatcher takes asks, grants and draws in turn, so a grant is routed and a draw is passed on however many asks arrive.  A reader the program drops leaves its construct on screen, showing what the program writes into it.  The person's input on it dismisses the question (Dismiss).  The program cannot tell that the reader is gone, so the construct process binds the question all the same, and the construct of a stream stays on screen, its further input reaching no reader.

\iffull
\mypara{The canonical compilation}  Let $M$ be a vGLP program.  A procedure of $M$ \emph{reaches a question} if it is volitional or a clause of it calls a procedure that does; each such procedure has one more argument, last, its \emph{ask stream}, which its goal writes and the dispatcher reads.  In a clause of such a procedure, each call in the body of a procedure that reaches a question is given a fresh writer as its ask stream, and the head carries $D?$, where $D$ is the clause's ask stream, with the body's ask streams merged into $D$: the call is given $D$ itself where there is one such call, \verb|merge| goals merge them where there are more, and the head carries \verb|[]| in place of $D?$ where there is none.

For a volitional procedure $q$ of $M$ with interactive type $T$, the \emph{asking clause} of $q$, under the name $q$ so that every call of $q$ reaches it unchanged, is
\begin{verbatim}
q(S1, ..., Sn, [ask(T, t(X))|D?]) :- q1(S1', ..., Sn', X?, D).
\end{verbatim}
where $S'_l$ passes the argument on: at an input position the head carries $S_l$ and $S'_l$ is its reader $S_l?$, and at an output position the head carries $S_l?$ and $S'_l$ is the writer $S_l$.  \verb|X| and \verb|X?| are exchanged where $T$ is in writer mode.  \verb|t| is the name of $T$ in lowercase, its type arguments appended, with \verb|_w| for writer mode and \verb|_r| for reader mode, as in \verb|card_w| and \verb|stream_request_r|, so that the asks of $M$ are of one type, the union over those functors.

The \emph{canonical compilation} $\lceil M\rceil$ of a vGLP program $M$ is the GLP program consisting of:
\begin{itemize}
\item every ordinary clause of $M$, extended with its ask streams (above);
\item for each volitional procedure $q$ of $M$, its clauses as guarded clauses of arity $n{+}1$ under a fresh name, \verb|q1| above, extended, and its asking clause;
\item for each interactive type $T$ of $M$, the \emph{construct process} of $T$, which the dispatcher spawns on an ask with its variable and its channels to the dispatcher (Appendix~\ref{app:proofs});
\item the \emph{dispatcher}, which reads the ask stream of the initial goal and the person channel, spawns a construct process with a fresh identifier on each ask, and routes the person's grants to the construct processes by identifier, by the standard stream techniques (Appendix~\ref{app:proofs} prints it).
\end{itemize}
The initial goal of $\lceil M\rceil$ at each agent is that of $M$, extended, and the dispatcher on the ask stream of that goal and on the person channel, which is its person side.

No vGLP module imports from another module a procedure that reaches a question, so the questions of a program are those of one module; the GLP module that forms the program's initial goal imports that module's compiled procedures, ask streams included, and the dispatcher.  A \texttt{.vglp} source in the directory tree of a program is compiled into a module of that program, as a \texttt{.glp} source is.  The compilation is typed: $\lceil M\rceil$ carries the type definitions and procedure declarations of $M$, each procedure that reaches a question declared with its ask stream added, a volitional procedure's declaration being its asking clause's and its guarded clauses declared under their fresh name with the argument of its interactive type added in the mode of that type, the type of the asks, and the declarations of the construct processes and the dispatcher.

A vGLP program declares the type of its output to its person by defining the type \verb|Output|, a reserved name; the compiled program draws that output as a thread, construct \verb|0|, whose elements are drawn by the widget of \verb|Output|, and the dispatcher merges these draws with the constructs'.  A program that writes nothing to its person defines no \verb|Output| and draws no thread.  A program that writes to its person and defines no \verb|Output| is refused, and so is one whose \verb|Output| has a position the person writes.  The thread is drawn when the program starts, and the dispatcher numbers the questions from \verb|1|.

The asking clause passes each argument of the goal on through a pair of its own, so until the pair's assignment is communicated a compiled clause that reads the argument suspends; a later clause that does not read it could then be taken where the source goal takes the earlier one, and the restriction on the clauses of a volitional procedure (Definition~\ref{def:volitional-clause}) excludes it.  At an output position the asking clause's head carries a reader, which a writer of the goal matches and a term of the goal that is not a writer does not, so the calls of a volitional procedure hold writers there (Definition~\ref{def:volitional-clause}), and a goal keeps them until it is asked, holding their one occurrence.

\fi

The construct process is the program's end of the widget: it forms from the person's input a term of its question's type, so the type checker~\cite{shapiro2026types} guarantees that a question is bound only with a well-typed term.  The bridge draws the widget and reports raw input; everything crossing to Dart is a ground term or a construct identifier, and answers and channels stay on the GLP side.

\begin{definition}[Implementation of vGLP by GLP, Correct, Complete]\label{def:vglp-by-glp}
The \temph{implementation of vGLP by GLP} is the function taking a vGLP program $M$ to its canonical compilation $\lceil M\rceil$ (Section~\ref{sec:compilation}).  It is \temph{correct} if for every $M$ and $P\subset\Pi$ some $\sigma$ is a correct implementation of the vmaGLP transition system over $P$ and $M$, under the fairness condition that its runs are proper, by GLP with persons over $P$ and $\lceil M\rceil$, and \temph{complete} for $M$ if for every $P$ some such $\sigma$ is also complete.
\end{definition}

\begin{restatable}[Compilation]{theorem}{restateCompilation}\label{thm:compilation}
The implementation of vGLP by GLP is correct, and complete for every program whose interactive types have complete widgets.
\end{restatable}

The proof, with the mapping $\sigma$, is in Appendix~\ref{app:proofs}.  Correctness is per transition and then per class: the asking clause's reduction is the Ask, a construct process binding its open question's answer is the Answer, and every other step of a construct process, of the dispatcher, of a \verb|merge| goal the compilation adds and of a grant is a stutter; a class enabled forever is taken, the asking clauses and the compiled clauses being GLP clauses.  Completeness is by construction: for each Answer of the given run, the person grants the input that forms its term, which a complete widget allows.  The compilation is per agent, the person channel being local to its agent, and Communicate and Cold-call are carried unchanged by the multiagent GLP runtime~\cite{shapiro2025glp,shapiro2026implementing}.

\mypara{The implementation}  The standard GLP implementation, extended with the compiler of the canonical compilation and the Dart bridge rendering the constructs, implements vGLP: a vGLP program is compiled into its canonical compilation, the engine runs it unchanged, and the bridge draws each construct's widget and returns the person's raw input as grants.  The composite --- engine, compiler, construct processes, dispatcher and bridge --- is the implementation of vGLP; Section~\ref{sec:platforms} shows each platform's requests and cards.

We claim, informally, that the composite is GLP with persons over the compiled program: a tap, a keystroke or a submission on a construct is a grant, and the engine is the GLP runtime of~\cite{shapiro2026implementing}.  One step is serial: the person's raw input reaches the GLP side on one stream, and the dispatcher routes it to the construct processes by identifier.

\subsection{The Cascade}\label{sec:cascade}

The implementations of Sections~\ref{sec:decomposition}, \ref{sec:cva-implementation} and~\ref{sec:compilation} compose.  For an effective local specification with simple guards (Definitions~\ref{def:vg-spec} and~\ref{def:effective}), let $\Theta$ be its canonical platform, $\lceil\Theta\rceil$ the compilation of $\Theta$ into vGLP, and $\lceil\lceil\Theta\rceil\rceil$ the canonical compilation of that program onto GLP; let $\sigma_3$ be the compilation mapping (Definition~\ref{def:compilation-sigma}), $\sigma_2$ the platform mapping (Definition~\ref{def:platform-mapping}), and $\sigma_1$ the CVA implementation (Definition~\ref{def:decomposition-sigma}).

\begin{restatable}[Cascade]{theorem}{restateCascade}\label{thm:cascade}
For every effective local volition-guarded specification with simple guards and $P\subset\Pi$, $\sigma_1\circ\sigma_2\circ\sigma_3$ is a correct implementation of $\calF(P)$ by GLP with persons over $P$ and $\lceil\lceil\Theta\rceil\rceil$, under the fairness condition that a run $\rho$ is proper, that no card's question stays open in $\sigma_3(\rho)$, and that $\sigma_2(\sigma_3(\rho))$ is connected and unobstructed.
\end{restatable}

The proof is in Appendix~\ref{app:proofs}.  The fairness condition asks of the persons that they answer the cards put to them, and of the run that the parties to a transaction have discovered one another and that no agent is kept busy without end.  It does not ask for the strong fairness of Definition~\ref{def:fairness-cva}: the compiled agent takes a person's proposal at the dispatch that reads it, and an opener's at the tick it takes after every trigger, so a \textbf{Propose} enabled again and again is taken again and again.

The results of the paper, and the conditions each places on the runs of the implementing system, are these:
\begin{center}\small
\begin{tabular}{@{}l l p{0.5\textwidth}@{}}
Result & Arrow & Condition on the runs \\
\hline
Proposition~\ref{prop:commitment} & & safe; for termination, also live \\
Theorem~\ref{thm:decomposition} & $(3)\rightarrow(4)$ & connected, unobstructed and strongly fair \\
Theorem~\ref{thm:platform-compilation} & $(4)\rightarrow(5)$ & proper, no card's question staying open, and its image connected and unobstructed \\
Theorem~\ref{thm:compilation} & $(5)\rightarrow(6)$ & proper \\
Theorem~\ref{thm:cascade} & $(3)\rightarrow(6)$ & proper, no card's question staying open, and its image in $\calP$ connected and unobstructed \\
\end{tabular}
\end{center}

\section{Grassroots Platforms}\label{sec:platforms}

This section specifies three grassroots platforms --- the grassroots social graph, grassroots social network, and grassroots currencies --- as volition-guarded transactions, each a local specification (Definition~\ref{def:local}).  The CVA implementation of each is its canonical CVA platform (Definition~\ref{def:canonical-cva}), correct by Theorem~\ref{thm:decomposition}, and, the specification being effective with simple guards (Definition~\ref{def:effective}), its vGLP program is the compilation of that platform (Section~\ref{sec:cva-implementation}), correct by Theorem~\ref{thm:platform-compilation}.  A binary transaction guarded by both parties is then the two-phase commit over two participants (Section~\ref{sec:decomposition}): an offer is the initiator's \textbf{Propose}, an acceptance is the other party's \textbf{Answer} to the card that \textbf{Start} raises, and the initiator commits on receipt of the acceptance's vote and tells the other party; the person's no to the card takes \textbf{Answer}'s else-branch.

The programs are in the public \GLP repository\full{, \repo,}{} under \texttt{programs/}.  The compilations of the canonical platforms of the social graph and of the currencies, with coins of one unit, are in \texttt{vglp/cva}, printed by the compiler from act schemas of Section~\ref{sec:cva-implementation}.  Programs written earlier, by hand, are in \texttt{social/graph} (the social graph, one-to-one messaging and the super-app~\cite{shapiro2026gsg}), in \texttt{currencies} (the coins agent, in the syntax of Section~\ref{sec:vglp}) and in \texttt{cssn} (the child-safe network, in an earlier syntax of vGLP).  A \texttt{.vglp} source in the syntax of Section~\ref{sec:vglp} is compiled by the canonical compilation of Section~\ref{sec:compilation} into the \texttt{.glp} file beside it.  The requests below are given in the syntax of Section~\ref{sec:vglp}.

\subsection{Grassroots Social Graph}
\label{sec:grassroots-sn}

In a grassroots social network~\cite{shapiro2023gsn}, the social graph is stored distributively under the control of the people themselves, with each person storing the local neighbourhood pertaining to them, and no third-party having access unless explicitly granted.  For agents $P \subset \Pi$, the machine state of each agent $p$ is a finite set $c_p\subseteq P$ of its friends, initially empty, so that $S(P) := 2^P$.

\begin{definition}[Grassroots Social Graph Volition-Guarded Transactions]\label{def:sg-transactions}
The \temph{social graph volition-guarded transactions} are the following.
\begin{enumerate}
    \item \textbf{Befriend$(p,q)$}: $c'_p := c_p \cup \{q\}$, $c'_q := c_q \cup \{p\}$, provided $q \notin c_p$ and $p \notin c_q$.  Guarded by $\{p,q\}$.
    \item \textbf{Unfriend$(p,q)$}: $c'_p := c_p \setminus \{q\}$, $c'_q := c_q \setminus \{p\}$, provided $q \in c_p$ and $p \in c_q$.  Guarded by either $p$ or $q$.
\end{enumerate}
Two Befriend transactions of one pair $\{p,q\}$ are equivalent, and so are two Unfriend transactions of one pair.
\end{definition}

The grassroots social graph and social network are introduced and analysed in~\cite{shapiro2023gsn,lewis2026volitional}.

\mypara{vGLP implementation}  The compilation of the canonical platform (Section~\ref{sec:cva-implementation}).  The agent's request type has the alternative \verb|propose(|$[t]$\verb|)| for each class $[t]$ its person may propose --- Befriend$(p,q)$ and Unfriend$(p,q)$ for a peer $q$ --- beside \verb|known(Peer)|, by which the person meets a peer whose identifier they have.  The card \textbf{Start} raises on an arriving offer of friendship shows the offering peer and takes yes or no, no being the refusal.  The request form is a menu of the requests, and each card an inbox card.

\iffull
\subsection{Grassroots Social Network}
\label{sec:sn}

Groups extend the social network beyond bilateral communication.  A group is a set of agents who can all communicate with each other; any agent may create a group.  Adding a member requires both the creator and the member to be willing; removing a member or leaving a group can be initiated by either.

Group communication proceeds in two steps: a member posts a message to the group, placing it in the queue the group's creator keeps (a transaction of the member and the creator, guarded by the member), and the message is then delivered to all members (an unguarded transaction).

We extend the grassroots social graph (Definition~\ref{def:sg-transactions}) with group memberships and group communication.  Let $N$ be a set of group names and $X$ a set of message contents.  A \emph{group identifier} is a pair $g = (a, n)$ where $a \in P$ is the \emph{creator} and $n \in N$ is the name; we write $\mathit{cr}(g) = a$ and let $\calG$ denote the set of all group identifiers.

The local state of each agent $p$ is extended with a finite set $\mu_p \subseteq \calG$ of group memberships, with $p$ being a \emph{member} of $g$ when $g \in \mu_p$.  For each $g \in \calG$, $p$ also maintains a \emph{chat stream} $s_p^g \in (P \times X)^*$, and, for each group $g$ it created, the set $\nu_p^g\subseteq P$ of its \emph{members}, a \emph{post queue} $o_p^g \in (P \times X)^*$ and a \emph{delivery index} $\delta_p^g \in \mathbb{N}$.  Initially $\mu_p = \nu_p^g = \emptyset$, $\delta_p^g = 0$, and $o_p^g = s_p^g = \varepsilon$.  The friend sets of Definition~\ref{def:sg-transactions} are unchanged by the transactions below.

\begin{definition}[Grassroots Social Network Volition-Guarded Transactions]\label{def:sn-transactions}
The \temph{grassroots social network transactions} are those of Definition~\ref{def:sg-transactions} and the following.
\begin{enumerate}
    \item \textbf{Create group$(g)$}: A unary transaction by agent $a := \mathit{cr}(g)$, guarded by $a$, with participants $\{a\}$, provided $g \notin \mu_a$: $\mu'_a := \mu_a \cup \{g\}$ and ${\nu'}^g_a := \{a\}$.
    \item \textbf{Join group$(g, q)$}: A binary transaction with participants $\{a, q\} \subseteq P$, $a := \mathit{cr}(g) \ne q$, guarded by $\{a, q\}$, provided $g \in \mu_a$ and $g \notin \mu_q$: $\mu'_q := \mu_q \cup \{g\}$ and ${\nu'}^g_a := \nu^g_a \cup \{q\}$.
    \item \textbf{Leave group$(g, q)$}: A binary transaction with participants $\{a, q\} \subseteq P$, $a := \mathit{cr}(g) \ne q$, guarded by either $a$ or $q$, provided $g \in \mu_q$: $\mu'_q := \mu_q \setminus \{g\}$ and ${\nu'}^g_a := \nu^g_a \setminus \{q\}$.
    \item \textbf{Post to group$(g, p, x)$}: A transaction with participants $\{p, a\}$, $a := \mathit{cr}(g)$, guarded by $p$, provided $g \in \mu_p$: $o'^{g}_a := o_a^g \cdot (p, x)$.
    \item \textbf{Deliver to group$(g, Q, p, x)$}: An unguarded transaction with participants $Q$, $a := \mathit{cr}(g)\in Q$, provided $\nu^g_a = Q$, $\delta_a^g < |o_a^g|$ and $o_a^g[\delta_a^g] = (p, x)$: $\delta'^{g}_a := \delta_a^g + 1$ and, for every $q \in Q$, $s'^{g}_q := s_q^g \cdot (p, x)$.
\end{enumerate}
In each transaction, all components not explicitly updated are preserved.  Two transactions of one name and arguments are equivalent.
\end{definition}

Deliver-to-group is an abstraction of an atomic-broadcast primitive in the sense of~\cite{cristian1995atomic}: a posted message is appended, by a single unguarded transaction, to the chat stream of every current member of the group.  Its class names the members it reaches and the post it delivers, the transactions of a class having the same participants (Definition~\ref{def:equivalence}).

A \emph{public feed} created by agent $a$ is defined to be a group $g$ with $\mathit{cr}(g) = a$ in which the post-to-group transaction is restricted to $p = a$; all other members are \emph{followers}.

A \emph{direct message channel} between friends $p$ and $q$ is implemented by $p$ creating a group $g$ with $\mathit{cr}(g) = p$ and $q$ joining it; once those transactions have been carried out, the members of $g$ are $\{p, q\}$.  The guards of the join-group transaction ($\{p, q\}$) coincide with the befriend guards.

\mypara{vGLP implementation}  The compilation of its canonical platform (Section~\ref{sec:cva-implementation}): Create group, Join group, Leave group and Post to group are classes the person proposes, alternatives \verb|propose(|$[t]$\verb|)| of the agent's request type; a Join group proposed by the creator raises a card at the member, which shows the group and its creator and takes yes or no; and Deliver to group, unguarded, is proposed by the creator's machine, as its opener, and voted on by the members' machines.  Posting is a request, one message per submission, as the chat of Section~\ref{sec:vglp} is.  Groups are specified, not yet deployed (Section~\ref{sec:ui-primitives}).

\subsection{Grassroots Currencies}
\label{sec:coins}

Grassroots coins~\cite{shapiro2024gc,lewis2023grassroots} are units of debt that can be issued and traded digitally by any person.  Each person's coins are backed by the goods and services they offer, priced in their own currency, with liquidity arising from mutual credit via coin exchange among persons that know and trust each other.  We recall the transactions of grassroots currencies of~\cite{shapiro2024gc}.  A grassroots $p$-coin, denoted \textcent$_{p}$, is a unit of debt issued by $p \in \Pi$, and \textcent$^k_{p}$ is a multiset of $k$ coins \textcent$_{p}$.  The machine state of a party $p$ is their holdings $c_p$, a multiset of coins, initially empty.  We use $\cup$ and $\setminus$ for multiset sum and difference.

\begin{definition}[The Transactions of Grassroots Currencies~\cite{shapiro2024gc}]\label{def:gc}
For persons $p \ne q \in \Pi$, $u, v, r \in \Pi$ and amounts $k, k' \ge 1$, the \temph{transactions of grassroots currencies} are the volition-guarded transactions $(c \rightarrow c', Q')$ below, $Q'$ their guards~\cite{lewis2026volitional}, over the participants named, in which $c$ satisfies the precondition stated and $c'_x = c_x$ for every person $x$ not named:
\begin{enumerate}
    \item \temph{Mint$_p(k)$}, over $\{p\}$, guarded by $\{p\}$: $c'_p = c_p \cup \text{\textcent}^k_p$.
    \item \temph{Swap$_{p,q}(\text{\textcent}^k_u, \text{\textcent}^{k'}_v)$}, over $\{p,q\}$, guarded by $\{p,q\}$: precondition $\text{\textcent}^k_u \subseteq c_p$ and $\text{\textcent}^{k'}_v \subseteq c_q$; $c'_p = c_p \setminus \text{\textcent}^k_u \cup \text{\textcent}^{k'}_v$ and $c'_q = c_q \setminus \text{\textcent}^{k'}_v \cup \text{\textcent}^k_u$.
    \item \temph{Pay$_{p,q}(k)$}, over $\{p,q\}$, guarded by $\{p\}$: precondition $\text{\textcent}^k_q \subseteq c_p$; $c'_p = c_p \setminus \text{\textcent}^k_q$ and $c'_q = c_q \cup \text{\textcent}^k_q$.
    \item \temph{Redeem$_{p,q}(r)$}, over $\{p,q\}$, guarded by $\{p\}$: precondition $\text{\textcent}_q \in c_p$ and $\text{\textcent}_r \in c_q$; $c'_p = c_p \setminus \text{\textcent}_q \cup \text{\textcent}_r$ and $c'_q = c_q \setminus \text{\textcent}_r \cup \text{\textcent}_q$.
\end{enumerate}
\temph{Equivalence} is the least equivalence on the underlying machine transactions under which two transactions with the same name and arguments, differing only in the configurations $c$ and $c'$, are equivalent.
\end{definition}

Mint, pay, and redeem are guarded by the initiator; a swap is guarded by both parties.  A redemption is also the swap of the same two coins, so the two are equivalent, and their class is guarded by the bearer; $p$'s $q$-coin taken for $q$'s $p$-coin is a redemption by either of them, so its class is guarded by either.  These transactions define a grassroots currencies protocol $\mathit{GC}$~\cite{lewis2026volitional}.

\mypara{vGLP implementation}  The compilation of the canonical platform (Section~\ref{sec:cva-implementation}): every class is one the person proposes, an alternative \verb|propose(|$[t]$\verb|)| of the agent's request type whose arguments the request form takes, coins being named by their issuer and their amount.  In the coins agent of the repository, in \texttt{currencies/coins}, the requests are \verb|mint(Amount)|, \verb|swap(Friend,| \verb|GiveCoin,| \verb|GiveAmount,| \verb|WantCoin,| \verb|WantAmount)|, \verb|pay(Friend, Amount)| and \verb|redeem(Friend, WantCoin)|; the holdings are a list the program writes to the person.  A swap raises a card at the other party, which shows the coins given and the coins taken and takes yes or no.  The other party's machine votes no on a swap whose wanted coins it does not hold, without asking its person.  It votes on a payment and on a redemption by its own precondition, and its person is asked by neither.

The agent's request form is a menu of its requests, each a form whose fields are the request's arguments, so the person's input is the term the agent consumes.

\mypara{Effective specifications}  The three specifications are effective local specifications with simple guards (Definitions~\ref{def:vg-spec} and~\ref{def:effective}), each class guarded by one guard set or by either party: their machine states are finite sets, multisets and sequences and their effects operations on them; a class of the currencies names its coins by the flat term of their issuer and their amount; the social graph and the currencies have no unguarded class; and the unguarded classes of the social network, the Deliver to group of each group with a post to deliver, are opened by the group's creator, whose state records the group's members, its posts and how many it has delivered.  So their canonical platforms compile into vGLP (Section~\ref{sec:cva-implementation}).

\fi

\section{The Grassroots App}\label{sec:ui-primitives}

The grassroots app is the grassroots super-app~\cite{shapiro2026gsg} as deployed on a smartphone: its mini-apps are the grassroots social graph, social network and coins, three programs in an earlier syntax of vGLP, rendered as the three panels of one app --- \emph{Friends}, \emph{Chats} and \emph{Coins} --- each reached by a bottom-bar icon badged with the panel's open questions (Figure~\ref{fig:grassapp}).  Every element of a panel is a construct (Definition~\ref{def:manifest}) or output the program writes to its person: the form of the agent's request stream, the cards of its responders, and the lists the program writes --- the friends, the chat threads, the balances --- each drawn by the default widget of its declared type (Definition~\ref{def:default-widget}).  The deployed build implements the social graph, one-to-one conversations and grassroots coins; groups, specified in Section~\ref{sec:platforms}, add no new kind of construct and are not yet deployed.

\mypara{Friends}  The social graph.  Its form is the agent's request stream, a menu of meet, offer friendship, unfriend and message a friend, meet being the request \verb|known(Peer)| by which the person wills a Discover of a peer whose identifier they have (Section~\ref{sec:cva-implementation}).  Its card is the responder to a friend offer, with yes and no, and its list is the friends, which the program writes on \texttt{connected}, \texttt{unfriended} and \texttt{rejected}.  The deployed build has no meet: its offer of friendship reaches the peer by a cold call.

\mypara{Chats}  The social network --- one-to-one messaging, deployed, and group messaging, specified (Section~\ref{sec:platforms}).  The open conversation's input is the chat of Section~\ref{sec:vglp}: a stream question that stays open, one line per submission, each line carrying the status the program writes when the message is sent; the chat list and the threads are the program's output; the panel's card is a group invitation, accepted before the group joins the chat list, which the deployed build has not.  In the deployed build a machine carries messages only between friends, and a stranger's first contact is a friend offer, held as a card in \emph{Friends} until the person accepts it.

\mypara{Coins}  Grassroots coins among friends.  Its form is the request stream, a menu of mint, propose a swap, pay and redeem.  Its card is the proposed swap, answered yes or no.  Its list is the balances, organised by friend.  Mutual credit is confined to the social graph, so the wallet lists the person's friends, each drilling down to the coins they hold.

\mypara{The implementation as it stands}  The compilation of Section~\ref{sec:elicitation} is implemented: the compiler emits each program's asking clauses, construct processes and dispatcher, the type checker checks the result, and the test suite runs the programs of Sections~\ref{sec:introduction} and~\ref{sec:vglp} with scripted persons that read the draws and grant on the constructs, each run checked against the committed sources.  A compiled module \verb|M| is run by the goals
\begin{verbatim}
M # agent(MCh?, NetCh?, Asks), M # dispatch(Asks?, PersonCh?, MCh)
\end{verbatim}
with the person channel \verb|PersonCh = ch(Grants?, Draws)|, its grants \verb|input(Id, P, R)| --- the construct, the position and the person's raw input --- and its draws \verb|draw(Id, W, V)| and \verb|withdraw(Id)| (Section~\ref{sec:compilation}).  The Dart bridge that draws these constructs and returns the person's input is written, and the platforms' ports to it are in progress; the app as deployed, the super-app of~\cite{shapiro2026gsg}, runs the three platforms as its mini-apps, in an earlier syntax of vGLP, compiled by an earlier compilation into an agent and a mediator, the super-app's, and rendered by the bridge of that design, and Figure~\ref{fig:grassapp} shows it.  One bridge, carrying only the platforms' compiled sources, builds and runs unchanged on the desktop, the iPhone simulator and the phone, where the grassroots app is deployed as a native app.

\begin{figure}[t]
\centering
\iffull
\begin{minipage}[t]{0.1575\textwidth}\centering\includegraphics[width=\linewidth]{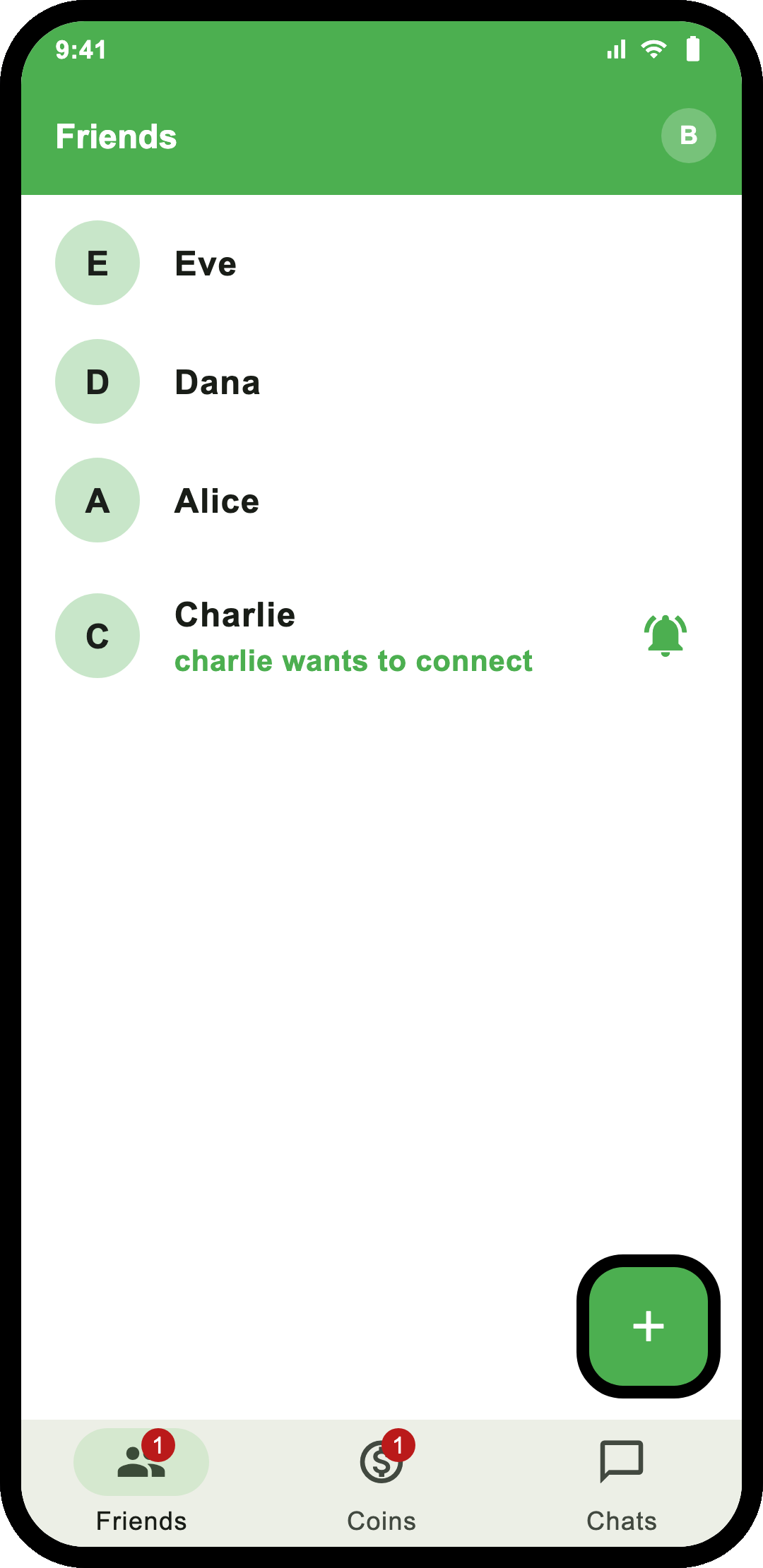}\\[2pt]{\footnotesize\textbf{(1)}}\end{minipage}\hfill
\begin{minipage}[t]{0.1575\textwidth}\centering\includegraphics[width=\linewidth]{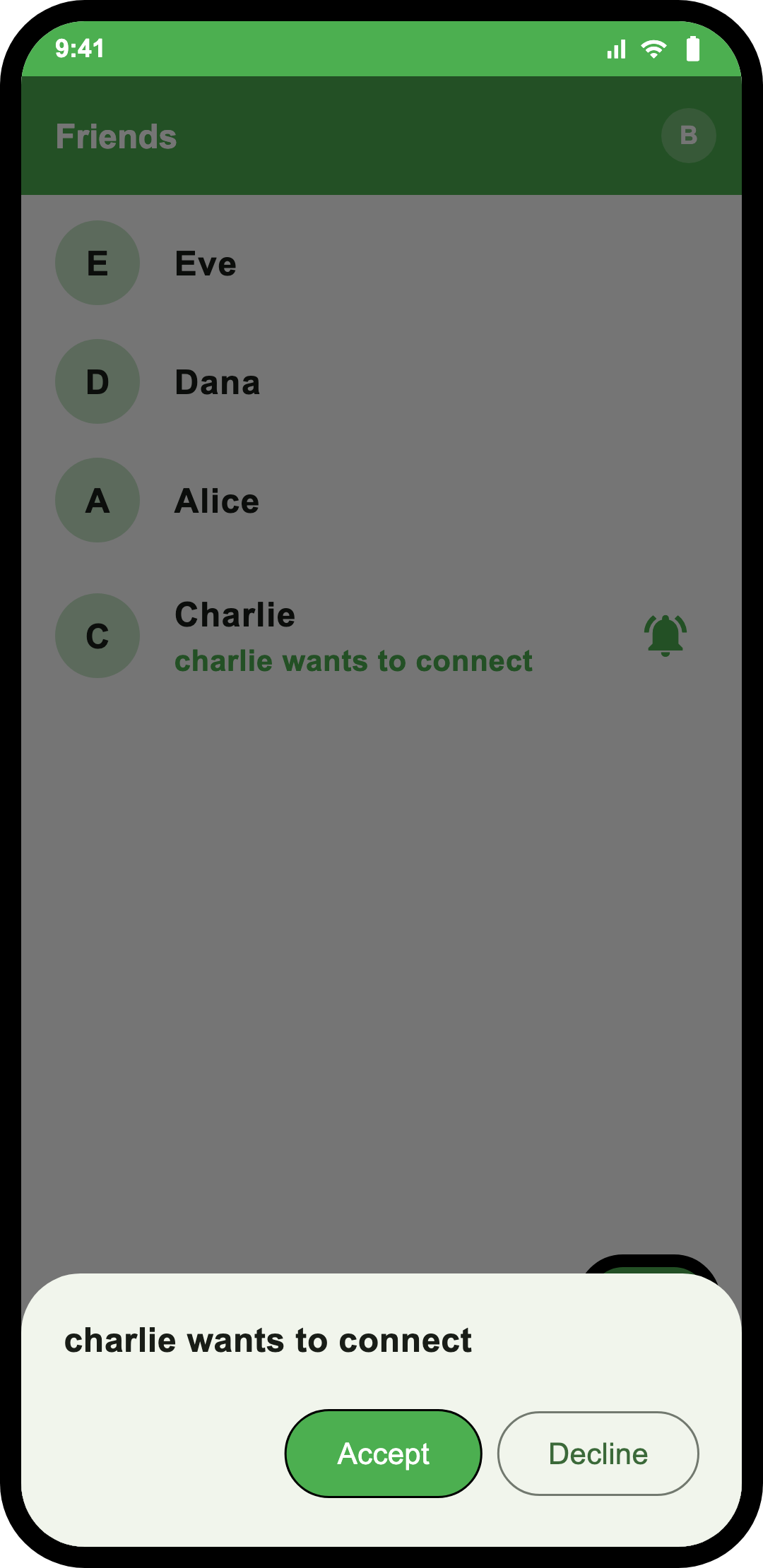}\\[2pt]{\footnotesize\textbf{(2)}}\end{minipage}\hfill
\begin{minipage}[t]{0.1575\textwidth}\centering\includegraphics[width=\linewidth]{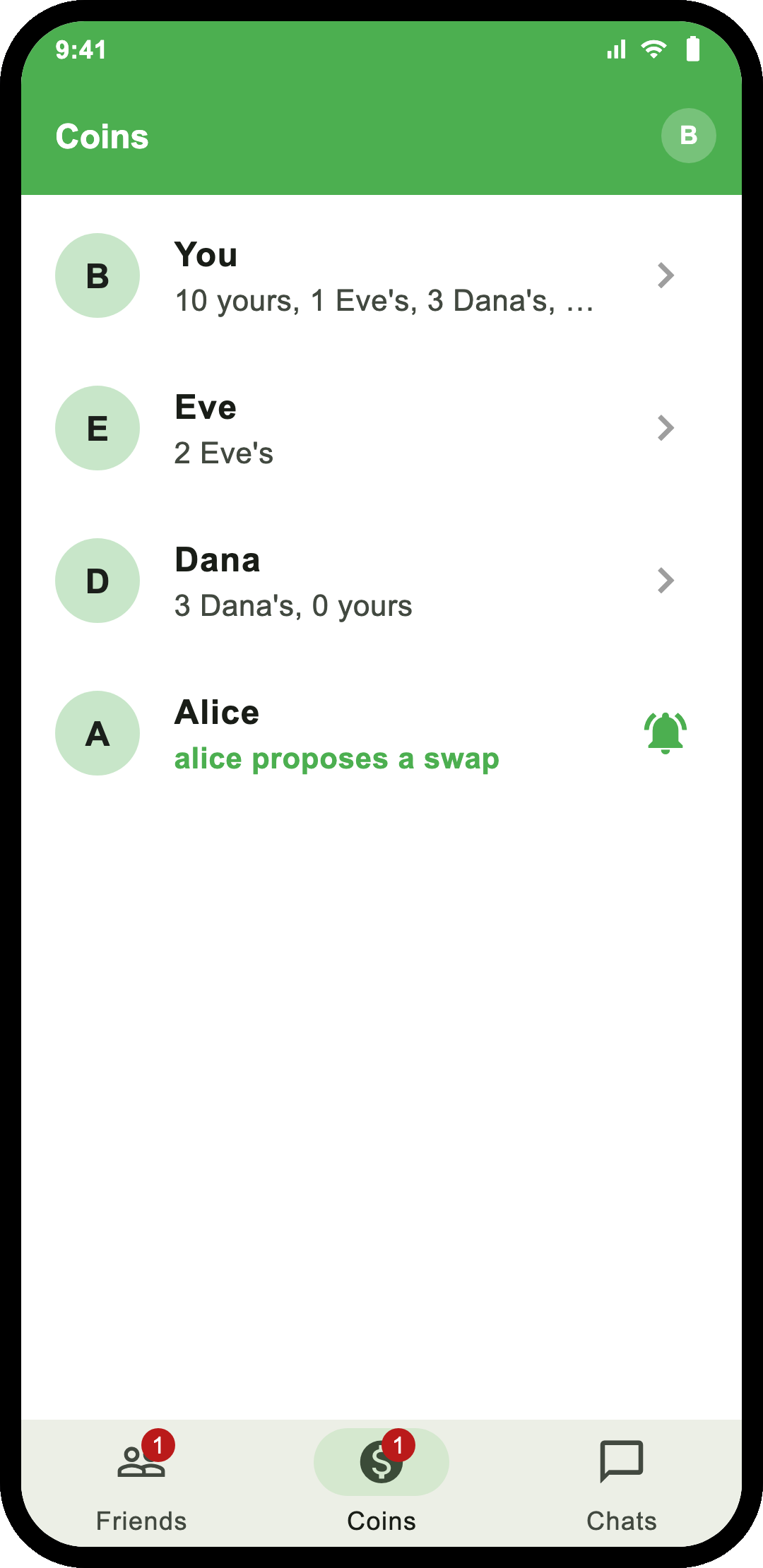}\\[2pt]{\footnotesize\textbf{(3)}}\end{minipage}\hfill
\begin{minipage}[t]{0.1575\textwidth}\centering\includegraphics[width=\linewidth]{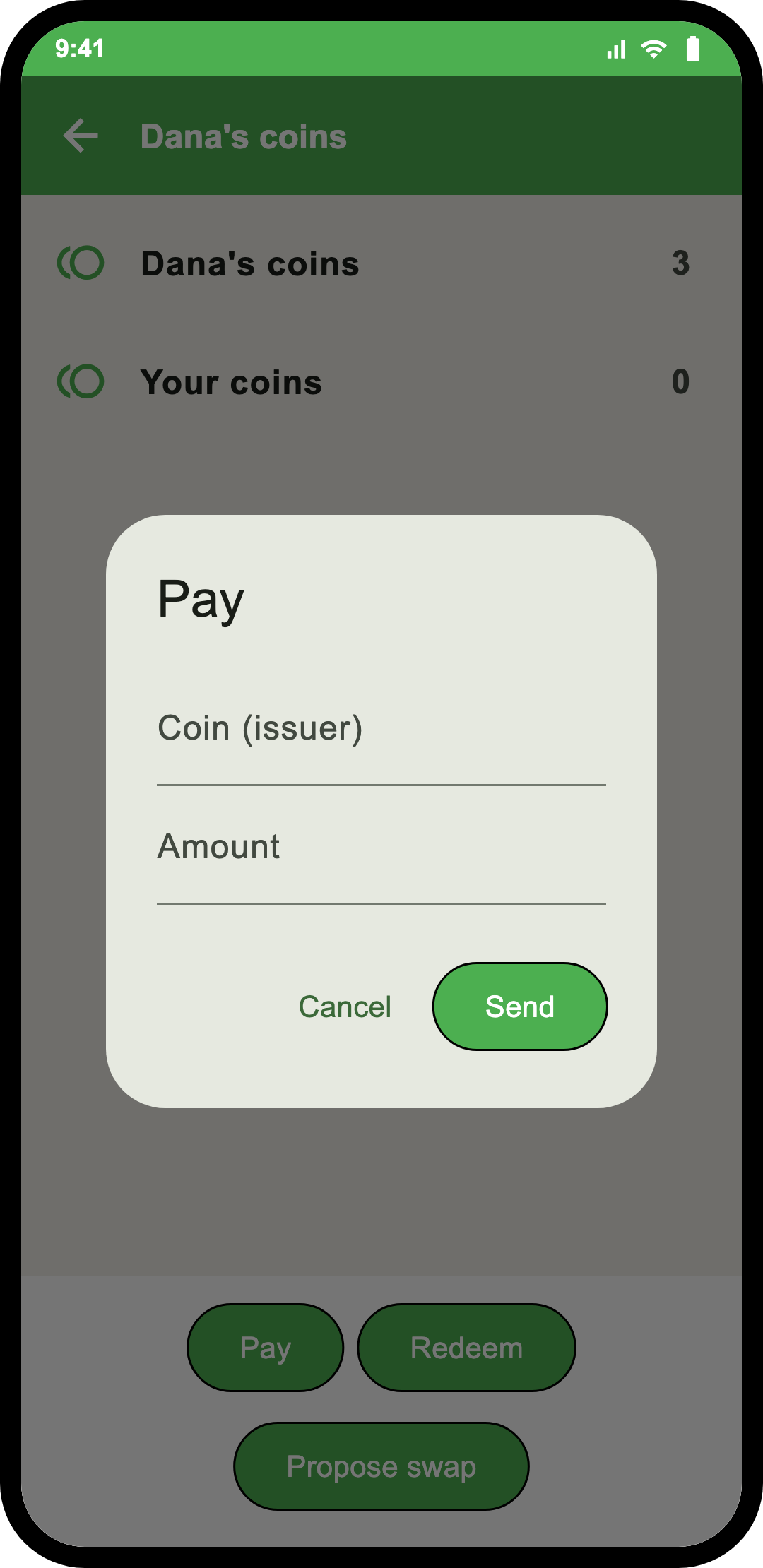}\\[2pt]{\footnotesize\textbf{(4)}}\end{minipage}\hfill
\begin{minipage}[t]{0.1575\textwidth}\centering\includegraphics[width=\linewidth]{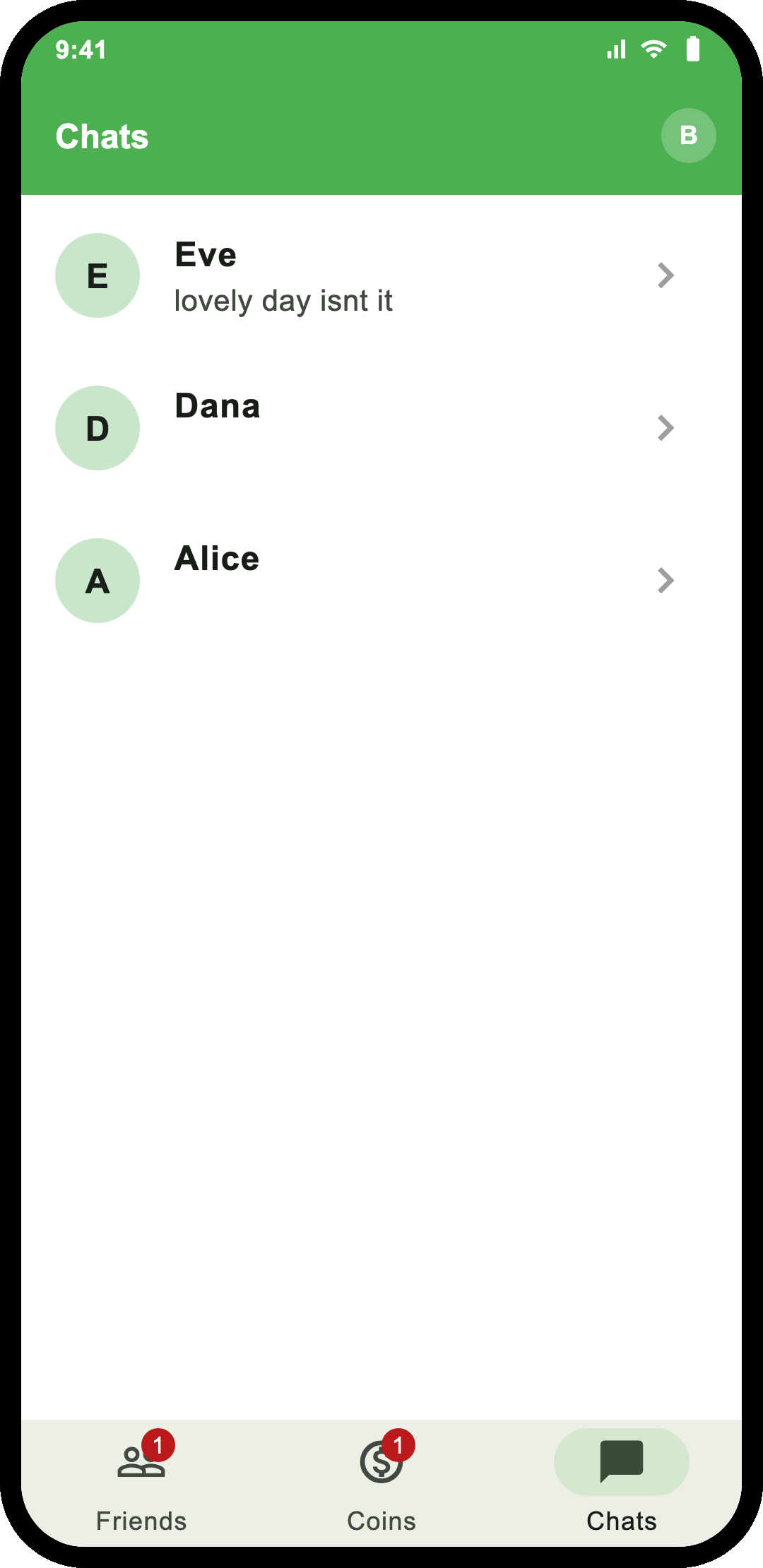}\\[2pt]{\footnotesize\textbf{(5)}}\end{minipage}\hfill
\begin{minipage}[t]{0.1575\textwidth}\centering\includegraphics[width=\linewidth]{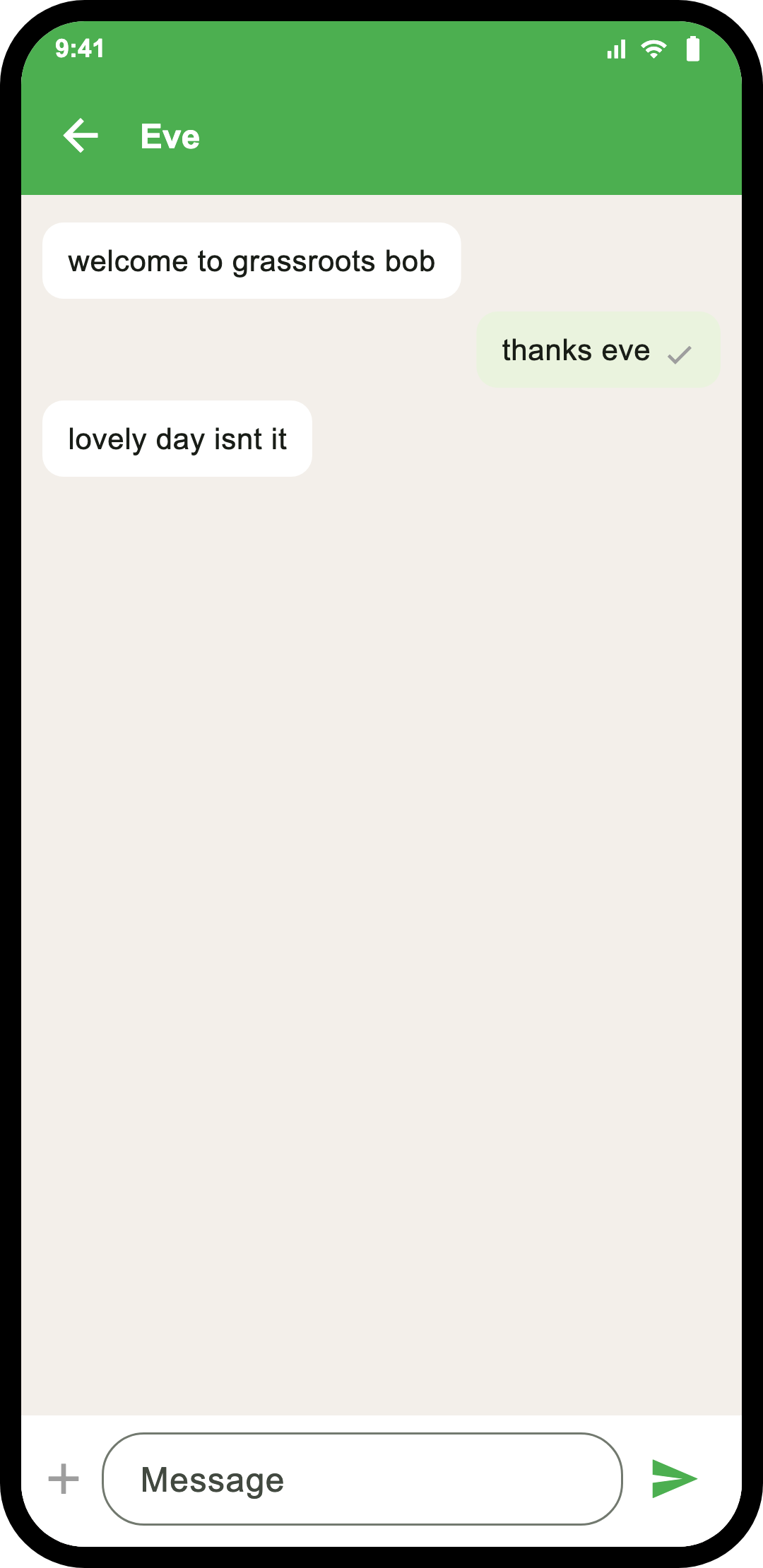}\\[2pt]{\footnotesize\textbf{(6)}}\end{minipage}%
\caption{\textbf{The grassroots app}, as deployed: the grassroots super-app~\cite{shapiro2026gsg} over three grassroots platforms, its mini-apps, each a vGLP program rendered by a single Dart bridge, and each shown here by its panel and one construct.  \emph{Social graph}: (1)~the \emph{Friends} panel --- the friends list, with Charlie's request unanswered; (2)~the friend-offer card --- the responder's Accept and Decline, its yes and no.  \emph{Coins}: (3)~the \emph{Coins} panel --- balances organised by friend, with Alice's swap unanswered; (4)~the pay form over the open friend's wallet --- an alternative of the request type, whose fields are the person's inputs, coin and amount.  \emph{Social network}: (5)~the \emph{Chats} panel --- the conversation list; (6)~an open conversation --- the chat's stream question.  Each panel is reached by a bottom-bar icon badged with its open questions, and every element is a construct (Definition~\ref{def:manifest}) or the program's output.}
\label{fig:grassapp}
\fi
\end{figure}

\section{Related Work}\label{sec:related-work}

\iffull
vGLP makes two independent contributions: it conditions a machine step on a person's volition, and it derives the interface from the program's semantics.  The first relates to formal approaches to user interface semantics, in which the person is a source of input, an environment, a chooser, an editor, or a signer; we review these first, by the person's role.  The second relates to work that derives the interface automatically, but from an application's data, object model, or tasks rather than from a semantics whose steps a person guards; and a further tradition predicts or verifies a given interface.  In none does a transaction wait on the standing wills of its guards, with an else-branch taken when a guard does not will it, and in none is a running interface compiled from such transactions with a proof that it implements them.  Underlying these differences is what a human input \emph{is}: in the languages that present a program to its user it is a \emph{command} driving a tool; in those that generate an interface from data or tasks it is a \emph{value} completing a task; in vGLP it is a \emph{volition} authorising an otherwise-ready machine transaction --- a guard.

\mypara{Formal approaches to user interface semantics}  Functional reactive programming gave interaction a denotational semantics with Fran~\cite{elliott1997fran}; Fudgets composed graphical interfaces from stream processors in a lazy functional language~\cite{carlsson1993fudgets}; Elm gave a GUI language an operational semantics~\cite{czaplicki2013elm}; and Krishnaswami and Benton gave GUIs a denotational model with user input as nondeterminism, by a monad~\cite{krishnaswami2011gui}.  Task-oriented programming executes typed editors the person fills: iTasks~\cite{plasmeijer2007itasks,plasmeijer2012top}, and its formal core TopHat, whose labelled transitions consume user inputs~\cite{steenvoorden2019tophat}.  iTasks derives those editors generically from the task's data types, and TopHat from the task's syntax too, the user continuing a step explicitly, so an interface \emph{is} derived, from the data and the task, and the inputs drive a workflow.  TopHat proves that the inputs it computes for a task are exactly those its semantics accepts~\cite{steenvoorden2019tophat}, as Theorem~\ref{thm:completeness} proves of a construct's answers.  In vGLP only the person gives an answer (Theorem~\ref{thm:volitional-soundness}), the answers guard multiagent transactions, and the cascade is proved correct. Game semantics models a program against an Opponent whose strategy is the environment's~\cite{hyland2000pcf}; interaction trees represent impure programs whose events the environment answers~\cite{xia2020itrees}; interactive computation makes the environment part of the computational model~\cite{wegner1997interaction}. 
Formal human-computer interaction modelled interactive systems directly: the PIE model~\cite{dix1991formal}, interactors~\cite{duke1993abstract}, Petri-net interactive cooperative objects~\cite{palanque1994ico}, statecharts~\cite{harel1987statecharts}, and the synchronous languages~\cite{berry1992esterel}.  These works define the semantics of an interface already designed.  The concurrent logic languages connected keyboard and screen as stream-processing agents~\cite{shapiro1989family}; a message to the person could carry an unbound variable that the person's response assigns.  Ceptre's interactive stages let the player select which enabled linear-logic rule fires~\cite{martens2015ceptre}: the person resolves nondeterminism. Hazelnut gives an operational semantics to the person's edit actions on programs~\cite{omar2017hazel}. Smart-contract semantics fire a transition only on a transaction signed by a key~\cite{sergey2019scilla}.  
Beliefs, desires, and intentions in agent programming~\cite{rao1995bdi} are the software agent's own, not a person's.

\mypara{The person in concurrency theory}  In CSP~\cite{hoare1985communicating}, CCS~\cite{milner1980calculus} and the $\pi$-calculus~\cite{milner1992calculus} the environment is a process, and a person is one more process or an external choice the environment resolves; session types~\cite{honda2016multiparty} discipline the conversation between processes and say nothing of who wills it.  The person's role across the formal traditions --- choice machines, I/O automata and reactive systems, angelic and demonic nondeterminism, alternating-time logics, ceremony analysis --- is reviewed in~\cite{lewis2026volitional}, and GLP's moded types are set against session types in~\cite{shapiro2026types}.  Here the agent is a person and a machine, and a machine transaction is subject to the person's volition.  A process calculus would host the transition systems of Section~\ref{sec:foundations} as well as GLP does; it does not by itself derive an interface from a procedure's interactive type.

\mypara{The interface derived from the program}  Deriving the interface from the program, rather than authoring it separately, is an established idea, on bases distinct from ours.  Naked objects presents a system's domain objects directly and creates the interface automatically from their definitions by reflection~\cite{pawson2004naked,pawson2002naked}, in a noun-verb style: the person views objects and invokes their methods, with no hand-written presentation layer.  Model-based user-interface development derives interfaces from separate task and domain models --- the CAMELEON reference framework~\cite{calvary2003cameleon}, ConcurTaskTrees and model-based design~\cite{paterno2000model}, and description languages such as UIML~\cite{abrams1999uiml} and UsiXML~\cite{limbourg2004usixml}, surveyed in~\cite{akiki2015survey}; task-oriented programming (above) derives its editors generically from data types and tasks; and structure editors derive an editing interface from a grammar~\cite{reps1984synthesizer}.  In each the interface follows from a description of the application's data, objects, or tasks.

\mypara{Typed widgets}  Three works associate an interface element with a type, as vGLP associates a widget with an interactive type (Section~\ref{sec:display}).  Formlets~\cite{cooper2008formlets} abstract a part of a web form that collects a value of a given type: a formlet renders its HTML and parses the submission into a typed value, formlets compose as idioms, and the handler attached to a form consumes the value.  Tangible values~\cite{elliott2007tangible} pair a value with a visualizer whose structure follows the structure of the value's type, with a default input and output per type and alternatives chosen explicitly, and end users compose them gesturally into new functions.  Livelits~\cite{omar2021livelits} fill a typed hole in a program with a user-defined GUI whose expansion is an expression of the hole's type, kept live by evaluating the expressions spliced into it in the closures collected around the hole.  In each, the interface produces a value of a type for a program or a programmer; in vGLP a widget forms a person's answer to a question the running program asks, and the reduction that reads the answer is the transaction the person wills.

\mypara{Atomic commitment}  Two-phase commit~\cite{gray1978notes,lampson1981chapter} commits a transaction across sites on the votes of all, a coordinator collecting the votes and announcing the outcome; Skeen's decentralised form has every site send its vote to every site and decide by itself, and his nonblocking protocol adds a phase so that a site may decide under failures~\cite{skeen1981nonblocking}; the textbooks~\cite{bernstein1987concurrency,gray1993transaction} treat the family.  The two-phase commit of Section~\ref{sec:decomposition} is the coordinated form, the agent that initiates the transaction coordinating it, with a guard set in place of unanimous consent: the persons of the guard set the initiator proposes with are asked, every other participant votes on its own precondition, and the else-branch is taken when a person does not will the then-branch.  Concurrent attempts are ordered by wound-wait~\cite{rosenkrantz1978system}, a deadlock-prevention scheme for distributed databases: a busy participant lets a start wait rather than refuse it, an attempt of higher priority makes the one of lower priority that blocks it yield, and an attempt tried again after yielding keeps its priority, so that contention neither deadlocks nor livelocks.  Fault tolerance is left to work in preparation.

\mypara{Predicting and verifying a given interface}  Interface design is often treated as a craft outside formal computer science, yet two lines formalise it.  Predictive engineering models estimate a design's performance from the person's elementary actions --- the keystroke-level model and the GOMS family~\cite{card1980keystroke,card1983psychology,johnKieras1996goms}.  Formal methods for interactive systems specify and verify interfaces --- from grammars for user interfaces~\cite{reisner1981formal} through the collected work of the field~\cite{palanque1998formal,weyers2017handbook}, of which the dialogue models above are part.  

To the best of our knowledge, vGLP is the first programming language whose operational semantics conditions transactions on a person's volition --- elicited, awaited until given, and fulfilled --- with the interface derived from the questions that semantics asks, each rendered by the widget of its interactive type, rather than from the application's data or tasks.
\fi

\section{Conclusion}\label{sec:conclusion}

We incorporated people in the operational semantics of a programming language: a goal of a vGLP volitional procedure puts a typed question to its person, a clause that reads the answer reduces only after the person has answered, and the person's answers are their own.  The structure of the interface is derived from the semantics: a program's questions are the interactive variables of its goals, each rendered by the widget of its interactive type --- a card, a form, the command-line box of a stream question, an input box for a stream.  The abstraction cascade proceeds from volition-guarded multiagent atomic transactions, by two-phase commit coordinated by the initiating agent, to communicating volitional agents, then by compilation to vGLP, and by compilation onto GLP; it is proved correct down to GLP.  The social graph and the currencies are compiled down the cascade and run on GLP with scripted persons, and the grassroots app runs earlier programs of the three platforms on a physical smartphone.  The two-phase commit orders contending attempts by wound-wait: every attempt ends, and contention keeps no transaction that stays enabled from being carried out.  An interactive type is the contract a user-interface developer programs against: any construct family whose widgets preserve it inherits the theorems, so the rendering of a platform is the person's choice.

Future work: a widget language beyond the basic types, drawing on user-interface specification languages; a GLP channel per construct in the bridge, removing the dispatcher's serial step; and standing volitions --- a person willing a class of goals rather than answering one question.

\bibliography{bib}

\newpage
\appendix

\section{\full{Grassroots Protocols}{Agents, People, Machines, Volition-Guarded Transactions and Volitional Transactions}}\label{app:vmat}

\iffull
\label{app:grassroots}
\fi

Section~\ref{sec:cva} states that the protocol of every CVA platform is volitionally grassroots; the definitions are these.  We are particularly interested in \emph{grassroots protocols}~\cite{shapiro2023grassrootsBA,shapiro2025atomic,lewis2026volitional}, which formalise the idea that grassroots platforms can have multiple instances that form and operate independently, yet may coalesce into ever-larger instances, possibly (but not necessarily) into a single global platform.  We recall the definitions of~\cite{lewis2026volitional}.  A \emph{prefix} of a computation is a finite initial segment of it.

\begin{definition}[Interleaving~\cite{lewis2026volitional}]\label{def:interleaving}
Let $P, P' \subset \Pi$ be disjoint and nonempty, $r = c_0, c_1, \ldots$ a run of $\calF(P)$ and $r' = d_0, d_1, \ldots$ a run of $\calF(P')$.  An \temph{interleaving} of $r$ and $r'$ is a sequence $e_0, e_1, \ldots$ of configurations of $\calF(P \cup P')$ for which there are non-decreasing index sequences $(i_k)$ and $(j_k)$ with $i_0 = j_0 = 0$ such that for every $k$: $(e_k)_p = (c_{i_k})_p$ for every $p \in P$; $(e_k)_q = (d_{j_k})_q$ for every $q \in P'$; and if $e_{k+1}$ exists then exactly one of $i_{k+1} = i_k + 1$ with $j_{k+1} = j_k$, a \temph{$P$-step}, or $i_{k+1} = i_k$ with $j_{k+1} = j_k + 1$, a \temph{$P'$-step}.  If $r$ is finite of length $n$ then $i_k = n$ for some $k$, and if $r$ is infinite then every index of $r$ is reached by some $i_k$; likewise for $r'$ and $(j_k)$.
\end{definition}

\begin{definition}[Interaction~\cite{lewis2026volitional}]\label{def:interaction}
Let $Q \subset \Pi$ and let $P, P' \subseteq Q$ be disjoint and nonempty; write $c|_R$ for the restriction of a configuration $c$ over $Q$ to $R \subseteq Q$.  A transition $e \rightarrow e'$ of $\calF(Q)$ is an \temph{interaction between $P$ and $P'$} if $e|_P \ne e'|_P$ and $e|_{P'} \ne e'|_{P'}$, and $e|_P \rightarrow e'|_P$ is not a transition of $\calF(P)$ or $e|_{P'} \rightarrow e'|_{P'}$ is not a transition of $\calF(P')$.
\end{definition}

\begin{definition}[Oblivious, Interactive, Grassroots~\cite{lewis2026volitional}]\label{def:grassroots}
A protocol $\calF$ is \temph{oblivious} if for every disjoint nonempty $P, P' \subset \Pi$ every interleaving of a correct run of $\calF(P)$ and a correct run of $\calF(P')$ is a correct run of $\calF(P \cup P')$; \temph{interactive} if for every $Q \subset \Pi$ and disjoint nonempty $P, P' \subseteq Q$, every prefix of a correct run of $\calF(Q)$ containing no interaction between $P$ and $P'$ is a prefix of a correct run of $\calF(Q)$ that contains one; and \temph{grassroots} if it is oblivious and interactive.
\end{definition}

\begin{definition}[Volitionally Grassroots~\cite{lewis2026volitional}]\label{def:volitionally-grassroots}
A volitional transactions-based protocol $\calF$ is \temph{volitionally grassroots} if it is grassroots and, for every disjoint nonempty $P, P' \subset \Pi$ and every safe run of $\calF(P \cup P')$ containing an interaction between $P$ and $P'$, the first interaction of $P$ and $P'$ is induced by a volition-guarded transaction $(t,Q')$ with $Q' \cap P \ne \emptyset$ and $Q' \cap P' \ne \emptyset$.
\end{definition}

Two groups thus become connected only by an act a person of each wills.  An else transaction (Definition~\ref{def:else}) is over a single participant and so is never an interaction; a binary transaction with an else-branch would let an agent of another group, who never wills it, trigger it, and could be a first interaction willed by neither, which is why Definition~\ref{def:refusal} allows only unary ones.

Two results of~\cite{lewis2026volitional} are used below and stated without proof.  The first is that an agent's volitional state, in any configuration of $\calF(P)$, holds only classes of transactions whose participants lie in $P$.  The second reduces obliviousness to a condition on those classes.

\begin{proposition}[\cite{lewis2026volitional}]\label{prop:oblivious-criterion}
A volitional transactions-based protocol is oblivious provided that for every disjoint nonempty $P, P' \subset \Pi$, no equivalence class whose transactions have participants in both $P$ and $P'$ is ever enabled in any interleaving of correct runs of $\calF(P)$ and $\calF(P')$.
\end{proposition}

\subsection{An Example: Child-Safe Befriending}\label{app:childsafe}

The example the Introduction summarises, as volition-guarded transactions.

\mypara{Example: Volition-guarded atomic child-safe befriending}  Child-safe befriending, from work in preparation, extends the grassroots social graph~\cite{shapiro2023gsn}, the foundational grassroots platform on which all other grassroots platforms are built.  Each agent $p$ holds: a set $c_p \subseteq P$ of its friends; sets $\rho_p, \kappa_p \subseteq P$ of the parties that declared themselves its parents and of those it declared itself a parent of; and natural numbers $\alpha_p$ and $\omega_p$ counting the acts by which a child of its gained, or lost, a friend or a group.  All are initially empty or zero.  A party is an \emph{adult} when $\rho_p = \emptyset$ and a \emph{child} otherwise.  In each transaction below, all components not explicitly updated are preserved.
\begin{tcolorbox}[colback=gray!5!white,colframe=black!75!black,top=2pt,bottom=2pt]
\begin{enumerate}
    \item \textbf{Child befriend$(r,s,p,q)$}, where $p \in \rho_r$, $q \in \rho_s$, $q \in c_p$ or $q = p$, $s \notin c_r$ and $r \ne s$:\\
$c'_r := c_r \cup \{s\}$, $c'_s := c_s \cup \{r\}$, $\alpha'_p := \alpha_p + 1$ and $\alpha'_q := \alpha_q + 1$, guarded by $\{r,s,p,q\}$.
    \item \textbf{Child unfriend$(r,s,p,q)$}, where $p \in \rho_r$, $q \in \rho_s$, $s \in c_r$ and $r \ne s$:\\
$c'_r := c_r \setminus \{s\}$, $c'_s := c_s \setminus \{r\}$, $\omega'_p := \omega_p + 1$ and $\omega'_q := \omega_q + 1$, guarded by any one of $\{r,s,p,q\}$.
\end{enumerate}
\end{tcolorbox}
\noindent Child befriending requires the two children and one parent of each to be willing, and requires the two named parents to be friends already.  Child unfriending is licensed by any one of the four, all four being its participants: each parent's count grows by one and names nobody, a parent holds no record of whom its child befriended, and any parent of either child may end any friendship of that child.  When $p = q$ --- two children of one parent --- the participants are $\{r,s,p\}$ and $\alpha_p$ grows by one.  The transactions by which parties declare parenthood are omitted: the example assumes the parents declared.

\section{Proofs}\label{app:proofs}

This appendix proves the results of Sections~\ref{sec:cva}, \ref{sec:decomposition}, \ref{sec:vglp}, \ref{sec:cva-implementation} and~\ref{sec:elicitation}.

\begin{proof}[Proof of Proposition~\ref{prop:cva-grassroots}]
Let $\calF$ be the protocol of a CVA platform (Definition~\ref{def:cva-platform}).  We show first that in every configuration of a safe run of $\calF(Q)$ and for every $p\in Q$, every message in $o_p$ has sender $p$ and a recipient in $\mathit{known}_p$.  It holds initially, $o_p$ being empty.  An app transaction adds to $o_p$ only messages $\mathsf{message}(p,r,x)$ with $r\in\mathit{known}_p$ and leaves $\mathit{known}_p$ as it is; Communicate removes a message from $o_p$; Discover and Introduce add to $\mathit{known}_p$; a volition change leaves both as they are.  And an introduction $\mathsf{intro}(q)$ that an agent $r$ sends has $q\in\mathit{known}_r$ (Definition~\ref{def:cva-platform}), and, knowing never being undone, $r$ knows both $q$ and the recipient from then on.

Second, which transitions are interactions.  Let $P, P'\subseteq Q$ be disjoint and nonempty.  An interaction between $P$ and $P'$ changes the state of an agent of each (Definition~\ref{def:interaction}); a unary transaction changes the machine state of its one participant, and a transaction taken removes its class from volitional states, which hold only classes of transactions their agent participates in (Definition~\ref{def:agent-state}); so an interaction is a transition of a binary transaction with a participant in each group, a Discover, an Introduce or a Communicate between some $p\in P$ and $q\in P'$.  A Discover$(p,q)$ or an Introduce$(p,q)$ is an interaction: it changes the states of both, and it takes $\mathit{known}_p$ to a set containing $q\notin P$, which no machine state in $S(P)$ holds, so $e|_P\rightarrow e'|_P$ is not a transition of $\calF(P)$.  Hence in a prefix of a safe run of $\calF(Q)$ with no interaction between $P$ and $P'$, $\mathit{known}_p\cap P'=\emptyset$ for every $p\in P$, since $\mathit{known}_p$ grows by Discovers and Introduces of $p$ only and one of $p$ with an agent of $P'$ is an interaction, and likewise $\mathit{known}_q\cap P=\emptyset$ for every $q\in P'$.

\emph{Volitionally grassroots.}  Let $P, P'\subset\Pi$ be disjoint and nonempty, $r$ a safe run of $\calF(P\cup P')$ interactive between $P$ and $P'$, and its first interaction between $p\in P$ and $q\in P'$.  A Communicate$(p,q,m)$ requires $m=\mathsf{message}(p,q,x)\in o_p$, so $q\in\mathit{known}_p$ by the first part, where $\mathit{known}_p\cap P'=\emptyset$ by the second: it is not the first interaction, and nor is a Communicate$(q,p,m)$.  An Introduce$(p,q)$ requires an introduction $\mathsf{intro}(q)$ in $i_p$, sent by an agent $r\in P\cup P'$ that by the first part knows $q$ and $p$: if $r\in P$, $q\in\mathit{known}_r\cap P'$, and if $r\in P'$, $p\in\mathit{known}_r\cap P$, against the second part; so it is not the first interaction either.  The first interaction is therefore a Discover$(p,q)$, induced by the volition-guarded transaction guarded by $\{p,q\}$, which meets both $P$ and $P'$.

\emph{Oblivious.}  By Proposition~\ref{prop:oblivious-criterion} it suffices that no class of transactions with participants in both $P$ and $P'$ is enabled in an interleaving of correct runs of $\calF(P)$ and $\calF(P')$.  Such a class is of a Discover, an Introduce or a Communicate between some $p\in P$ and $q\in P'$, and in an interleaving the state of $p$ is a state of $\calF(P)$.  Discover$(p,q)$ is enabled only if $p$ wills it, and the volitional state of $p$ in $\calF(P)$ holds only classes of transactions whose participants lie in $P$ (Definition~\ref{def:agent-state}), which $q$ does not.  Communicate$(p,q,m)$ is enabled only if $\mathsf{message}(p,q,x)\in o_p$, and $o_p\subseteq M(P)$ in $\calF(P)$, whose messages have their recipients in $P$.  Introduce$(p,q)$ is enabled only if an introduction $\mathsf{intro}(q)$ is in $i_p$, sent in $\calF(P)$ by an agent that knew $q$, and the known peers of a state of $\calF(P)$ lie in $P$, which $q$ does not.  None is enabled, and the same holds with $P$ and $P'$ exchanged.

\emph{Interactive.}  Let $Q\subset\Pi$, $P, P'\subseteq Q$ disjoint and nonempty, and $\hat r$ a prefix of a correct run of $\calF(Q)$ with no interaction between $P$ and $P'$, ending in $e$; take $p\in P$ and $q\in P'$.  By the second part $q\notin\mathit{known}_p$ and $p\notin\mathit{known}_q$ in $e$, so Discover$(p,q)$ changes both their states, and it is enabled once both will it.  Extend $\hat r$ by a volition change of $p$ adding $[\mathrm{Discover}(p,q)]$, one of $q$ adding it, and the Discover, an interaction between $P$ and $P'$.  The result is a finite safe run of $\calF(Q)$ with $\hat r$ as a prefix, and by the Extension Lemma of~\cite{lewis2026volitional} a prefix of a correct run of $\calF(Q)$.

$\calF$ is thus oblivious and interactive, hence grassroots, and volitionally grassroots.
\end{proof}

Five facts about the runs of the canonical CVA platform (Definition~\ref{def:canonical-cva}) are used in the proofs of Proposition~\ref{prop:commitment} and Theorem~\ref{thm:decomposition}.

(\ia) \emph{Votes and outcomes.}  No transaction removes a term from a record, and a participant votes at most once on an attempt: an initiator proposes an attempt once, its number growing, and tells $\mathsf{start}(\alpha)$ once to each other participant, which reads it once, by \textbf{Start}; \textbf{Start} votes, records a question, or, the attempt closed, does nothing; \textbf{Answer} records one answer to a pending question; and \textbf{Vote} votes on it once it is answered, which ends it.  So an attempt once decided stays decided, and one on which a participant votes $\mathsf{no}$ is never decided.  An attempt is closed at its initiator only by its \textbf{Propose}, if it has one participant, and otherwise by the first \textbf{Tally}, \textbf{Yield} or \textbf{Wound} that closes it, each requiring it not closed and telling each other participant one outcome; at any other participant it is closed only by that outcome's \textbf{Outcome}.

(\ib) \emph{Locks.}  For an attempt $\alpha$ of more than one participant, $L_p=\alpha$ exactly while $p$ has recorded $\mathsf{voted}(\alpha,\mathsf{yes})$ and not $\mathsf{closed}(\alpha)$: a vote yes, and \textbf{Propose}, set $L:=\alpha$ and require $L=\bot$ --- \textbf{Propose} by readiness, \textbf{Start} by its start not waiting, \textbf{Vote} explicitly --- and $L$ returns to $\bot$ only as $p$ closes $\alpha$, by \textbf{Tally}, \textbf{Yield} or \textbf{Wound} at its initiator and by \textbf{Outcome} elsewhere.  While $L_p=\alpha$, $s_p$ does not change: $s$ changes only by the \textbf{Propose} of an attempt of one participant, which requires $L=\bot$, by the \textbf{Commit} of $\alpha$ at its initiator, and by the \textbf{Outcome} that commits an attempt $\beta$, which, being decided, $p$ voted yes on, so that $L_p=\beta$ from that vote to that \textbf{Outcome}, and $\beta=\alpha$.  A vote yes being cast at a state in the domain of its effect, $s_p\in\mathit{dom}(f^{[t]}_p)$ while $L_p$ is an attempt of $[t]$, so the settled state applies the effect of every decided lock.

(\ic) \emph{Questions.}  At most one question of a class is pending at an agent, and while a question of $[t]$ is pending at $p$, $L_p$ is no attempt of $[t]$: \textbf{Start} records a question of $[t]$ only where its start does not wait, so with $p$ free for $[t]$: $L_p=\bot$ and none pending; and while it is pending, \textbf{Propose} and \textbf{Start} lock $p$ on no attempt of $[t]$ and record no second question of $[t]$, requiring $p$ free for $[t]$, \textbf{Vote} on another question locks $p$ on the attempt of that question, of another class, and \textbf{Vote} on it ends it.

(\id) \emph{Peers.}  Discover and Introduce are the only transactions that change $\mathit{known}$, and each adds each of two agents to the other's, so knowing is symmetric and never undone.  \textbf{Propose} requires the initiator to know every participant, so every participant an attempt reaches knows its initiator, and the proviso that $p$ knows the agents it sends to holds of every \textbf{Start}, \textbf{Vote}, \textbf{Tally}, \textbf{Yield} and \textbf{Wound} whose other conditions hold in a configuration of a run, \textbf{Wound} sending a yield to the initiator of an attempt whose start its agent has read.

(\iv) \emph{Priorities.}  For every attempt $\alpha$, finitely many attempts have priority over $\alpha$, and each attempt is the $\alpha$ of finitely many \textbf{Wound} transactions.  The second first: an attempt $\gamma$ is the $\alpha$ of a \textbf{Wound} at $q$ only while its start waits at $q$, or its question, answered yes, is pending there, once each; while it does, $q$ records $\mathsf{yielded}(\beta)$ or closes $\beta$ at the first \textbf{Wound} for each $\beta$ blocking it, and $q$ comes to be blocked by an attempt $\gamma$ has priority over only by \textbf{Vote} on a question pending before, finitely many: \textbf{Start} of such an attempt waits behind $\gamma$, so that no new question of it is raised, a question of it answered yes waits behind $\gamma$ where $\gamma$ is a question, and $q$, not clear, proposes nothing.  Now the first.  A rank is either the number of the \textbf{Propose} that sets it, \emph{fresh}, or that of the initiator's last attempt of the same class, recorded wounded; an agent's numbers grow, so finitely many fresh ranks are at most the rank of $\alpha$, and with $P$ finite finitely many pairs of an initiator and a rank give priority over $\alpha$.  The attempts of such a pair are of one class, and each but the last is wounded, by a \textbf{Yield} or a \textbf{Wound} for an attempt with priority over it, of a pair with priority over it.  Were some pair to have infinitely many attempts, the one of highest priority among them would be wounded infinitely often by the finitely many attempts of the pairs with priority over it, each wounding finitely often by the second part: so none has.

\restateCommitment*
\begin{proof}
Fix a safe run of $\calP$ and an attempt $\alpha=(i,k,\pi,[t])$ proposed in it.

\emph{Agreement.}  By (\ia), $\alpha$ is closed at its initiator once, and at every other participant by the outcome the initiator tells it on closing $\alpha$: committed everywhere if the initiator commits it, and aborted everywhere if it aborts it.

\emph{Validity.}  $\alpha$ is committed at its initiator by its \textbf{Propose}, if it has one participant, or by its \textbf{Commit}, which requires a yes of every other participant recorded, each recorded by its voter as it cast it; the initiator's own was recorded by its \textbf{Propose}.  So $\alpha$ is decided, and by agreement it is committed elsewhere only if committed at its initiator.  Every vote yes is cast with $L=\bot$ at a state in the domain of its voter's effect: \textbf{Propose} requires readiness, \textbf{Start} its start not waiting and checks the state, and \textbf{Vote} checks both.  The initiator of a guarded attempt took its \textbf{Propose}, guarded by it and a member of its approval class, willing that class; an asked participant votes yes only by \textbf{Vote} on its answer yes, its \textbf{Start} casting only $\mathsf{no}$, and its answer yes is recorded by the then-branch of \textbf{Answer}, guarded by it and a member of its approval class.

\emph{Termination}, the run being safe and live.  If $\alpha$ has one participant, its \textbf{Propose} closes it.  Otherwise, once its initiator closes it, it tells every other participant the outcome, which reaches it by liveness and is read by \textbf{Outcome}, closing $\alpha$ there; so it suffices that every attempt of more than one participant is closed at its initiator.  Call an attempt \emph{stuck} if it is never closed at its initiator; suppose some is, and let $\alpha=(i,k,\pi,[t])$ be a stuck attempt over which no stuck attempt has priority, one existing as finitely many attempts have priority over any (\iv).  Its initiator tells $\mathsf{start}(\alpha)$ to every other participant $q$, and the Communicate of that message is enabled until taken, so by liveness it reaches $q$, which knows $i$ (\id).  There the start waits only while $q$'s inbox holds the start of an attempt with priority over $\alpha$, or a question with priority over $\alpha$ and answered yes is pending at $q$ --- finitely many such attempts, none stuck --- or while $q$ is blocked by an attempt $\beta$.  If $\alpha$ has priority over $\beta$, \textbf{Wound} is enabled until taken or $\beta$ is closed at $q$, and it closes $\beta$ or sends a yield to its initiator, whose \textbf{Yield} closes $\beta$ unless it is closed already; if not, $\beta$ is not stuck; either way $\beta$ is closed at its initiator and, its outcome reaching $q$, at $q$.  While the start waits, $q$ is blocked anew only by \textbf{Vote} on a question pending before, finitely often, as in (\iv); so from some configuration on the start does not wait, and \textbf{Start} on it is enabled until taken.  So liveness takes it: $q$ votes on $\alpha$ or records the question $\alpha$.  In the last case $\alpha$ is the one question of $[t]$ pending at $q$ (\ic), and while it is pending and unanswered, $q$'s approval class on $[t]$ is enabled --- by the then-branch of \textbf{Answer} where $q$ wills that class and by its else-branch where not --- with no other member enabled, \textbf{Propose} requiring $q$ free for $[t]$; so liveness takes \textbf{Answer} unless $\alpha$ is closed at $q$ first, and then \textbf{Vote}, the answer yes waiting, as the start did, only finitely long.  Each vote reaches $i$ likewise and is read by \textbf{Tally}.  No participant votes $\mathsf{no}$, and no \textbf{Yield} or \textbf{Wound} closes $\alpha$, since each would close it; and no participant closes $\alpha$, $i$ telling no outcome; so every participant votes $\mathsf{yes}$, the \textbf{Tally} of the last of these votes records a yes of every other participant, and \textbf{Commit}, enabled from then until taken, closes $\alpha$, against the supposition.  So no attempt is stuck.
\end{proof}

The mapping for Theorem~\ref{thm:decomposition} reads the specification's machine states off the settled states and its volitions off the volitional states, the locks and the answers.  An agent $p$ has \temph{approved} an attempt $\alpha$ if $p\in G_\alpha$ and $p$ has recorded $\mathsf{voted}(\alpha,\mathsf{yes})$ or $\mathsf{answered}(\alpha,\mathsf{yes})$.

\begin{definition}[CVA Implementation]\label{def:decomposition-sigma}
The \temph{CVA implementation} of a local volition-guarded specification with simple guards over $P\subset\Pi$ is the function $\sigma$ taking a configuration $c$ of its canonical CVA platform over $P$ to the agent configuration over $P$ with, at each $p\in P$, machine state the settled state of $p$ in $c$, and volitional state the set of guarded classes $[t]$ with $p$ in some guard set of $[t]$ such that either (\emph{A}) the volitional state of $p$ holds its approval class on $[t]$ and $L_p$ is not an attempt of $[t]$ decided in $c$, or (\emph{B}) $p$ has approved an attempt of $[t]$ not decided in $c$ that is $L_p$ or a question pending at $p$.
\end{definition}

\restateCVAImplementation*
\begin{proof}
Fix a local volition-guarded specification with simple guards and $P\subset\Pi$, with $\calP$ and $\calF(P)$ as in Section~\ref{sec:decomposition}, and let $\sigma$ be the mapping of Definition~\ref{def:decomposition-sigma}.  The initial configuration has no locks and no wills, so $\sigma$ maps it to the initial agent configuration: every settled state is $s_0$ and every volitional state is empty; $(\calP,\sigma)$ is an implementation of $\calF(P)$ (Definition~\ref{def:implementation}).

\emph{Safety}, per transition $c\rightarrow c'$.

\emph{The deciding transition.}  Let the transition decide an attempt $\alpha=(i,k,\pi,[t])$: the \textbf{Propose} of $\alpha$, if $Q([t])=\{i\}$, and otherwise the \textbf{Tally} at $i$ that records the last of the other participants' yes votes.  In the second case, before it, every participant $q$ of $\alpha$ has voted yes on $\alpha$ and has $L_q=\alpha$ by (\ib), $\alpha$ not yet being closed at $q$ --- its initiator closes it, telling the outcome, only by a \textbf{Tally} of a $\mathsf{no}$, which no participant of $\alpha$ casts (\ia), by a \textbf{Yield} or \textbf{Wound} before it records every yes, or by \textbf{Commit} after --- so its settled state is $s_q$, in $\mathit{dom}(f^{[t]}_q)$ by (\ib).  After it the settled state of every participant $q$ is $f^{[t]}_q(s_q)$, and no other settled state changes; by locality $d\rightarrow d'$, with $d_q=s_q$ and $d'_q=f^{[t]}_q(s_q)$, is a member of $[t]$ over $Q([t])$.  Before the transition every member of $G_\alpha$ holds $[t]$ under $\sigma$, having approved $L_q=\alpha$, undecided, clause (\emph{B}).  After it no agent holds $[t]$: every participant has $L_q=\alpha$, decided, and so no question of $[t]$ pending (\ic), which fails both clauses; and only participants of $[t]$ are in its guard sets.  No agent's holding of another class changes, the transition changing only the record of $i$ and the decidedness of $\alpha$.  In the first case $s_i$, in $\mathit{dom}(f^{[t]}_i)$, becomes $f^{[t]}_i(s_i)$; before it $i$ holds $[t]$ by clause (\emph{A}), if $[t]$ is guarded, willing its approval class, its \textbf{Propose} being guarded by it, with $L_i=\bot$; after it $i$ has fulfilled that class, has $L_i=\bot$ and, having been free for $[t]$, no question of $[t]$ pending, so that no agent holds $[t]$; and nothing else changes.  Either way $\sigma(c)\rightarrow\sigma(c')$ is the volitional machine transaction induced by that member and $G_\alpha$ (Definition~\ref{def:vmat}), $G_\alpha$ being a guard set of every member of $[t]$ (Definition~\ref{def:vg-spec}), and it is enabled in $\sigma(c)$ (Definition~\ref{def:enabled}).

\emph{Other votes yes.}  A \textbf{Propose} of an attempt of more than one participant, or a \textbf{Vote} or \textbf{Start} casting a yes on an attempt $\alpha$ of $[t]$, sets $L_p:=\alpha$, undecided, and decides no attempt, so no settled state changes; an undecided lock leaves clause (\emph{A}) as it was.  A \textbf{Propose} of a guarded class fulfils $p$'s approval class on $[t]$, so $p$, which held $[t]$ by (\emph{A}) before, willing that class with $L_p=\bot$, holds it by (\emph{B}) after, having approved $\alpha$; a \textbf{Vote} yes ends the question $\alpha$, by which $p$ held $[t]$ by (\emph{B}) before, and locks $p$ on it, by which it holds $[t]$ by (\emph{B}) after; a \textbf{Start}, or a \textbf{Propose} of an unguarded class, fulfils no will, and $p$ is not in the guard set of $\alpha$.  No other class is affected: a stutter.

\emph{Answers.}  The then-branch of \textbf{Answer} records the answer yes to a question $\alpha$ of $[t]$ pending at $p$, undecided, $p$ not having voted on it, and fulfils $p$'s approval class on $[t]$; so $p$, which held $[t]$ by (\emph{A}) before, willing that class with $L_p$ no attempt of $[t]$ (\ic), holds it by (\emph{B}) after, having approved $\alpha$.  The else-branch fulfils no will and records the answer no, which $\sigma$ does not read.  Both are stutters.

\emph{Votes no.}  A vote $\mathsf{no}$, by \textbf{Start} or by \textbf{Vote}, locks nothing, changes no settled state, its attempt never being decided (\ia), and fulfils no will.  One by \textbf{Start}, or by \textbf{Vote} on the answer no, is a stutter.  One by \textbf{Vote} on the answer yes ends the question $\alpha$, of $[t]$, by which $p$ held $[t]$ by (\emph{B}), so $p$ may cease to hold $[t]$: a volition change of $p$, or a stutter.

\emph{Closing.}  A \textbf{Tally}, a \textbf{Commit}, a \textbf{Yield}, a \textbf{Wound} of its agent's own attempt, or an \textbf{Outcome}, that closes $\alpha$ at $p$ sets $L_p:=\bot$ where $L_p=\alpha$, and applies $p$'s effect where $\alpha$ is committed: if $\alpha$ is decided the effect moves from the settled state into $s_p$, and if not, $\alpha$ is aborted (Proposition~\ref{prop:commitment}) and the lock it releases carried no effect.  No settled state changes, and only $p$'s holding of $[t]$ can, as (\emph{B}) ends or (\emph{A}) resumes: a volition change of $p$, or a stutter.

\emph{The rest.}  A \textbf{Start} that records a question or reads the start of an attempt closed at its agent, a \textbf{Tally} that records a yes without deciding its attempt or reads the vote of an attempt closed at its agent, a \textbf{Wound} that sends a yield, and a \textbf{Yield} that reads the yield of an attempt closed or decided at its agent, change nothing $\sigma$ reads, numbers and ranks included: stutters.  Discover, Introduce and Communicate change $\mathit{known}$ and the messages, which $\sigma$ does not read: stutters.  A volition change of $p$ changes which approval classes $p$ holds, hence at most which classes $\sigma$ gives $p$, by (\emph{A}): a volition change of $p$, or a stutter.

Every safe run of $\calP$ therefore maps to a safe run of $\calF(P)$, and a member of $[t]$ occurs in $\sigma(\rho)$ only at a transition deciding an attempt of $[t]$.

\emph{Liveness}.  Let $\rho$ be a correct run of $\calP$ --- connected, unobstructed and strongly fair, by its fairness condition --- and suppose a class $[t]$ is enabled in every configuration of a suffix of $\sigma(\rho)$ in which no member of $[t]$ occurs.  By the safety cases no attempt of $[t]$ is decided in the corresponding suffix of $\rho$, and enabledness gives, at each of its configurations, the settled state of every $q\in Q([t])$ in $\mathit{dom}(f^{[t]}_q)$ and, the guard of an enabled member containing a guard set of $[t]$, a guard set of $[t]$ all of whose members hold $[t]$ under $\sigma$, the empty one if $[t]$ is unguarded.  A vote $\mathsf{no}$ on an attempt of $[t]$ cast in the suffix by \textbf{Start}, or by \textbf{Vote} on the answer yes, is cast with $L=\bot$ at a machine state outside the domain of its voter's effect, its settled state, which enabledness excludes; so from some configuration $c_1$ of the suffix on, after the votes on the answers no recorded before the suffix, at most one at each participant (\ic, \ia), every vote $\mathsf{no}$ on an attempt of $[t]$ is cast by \textbf{Vote} on an answer no recorded in the suffix.

None is.  If $[t]$ is unguarded, or its guard sets are singletons, an attempt of $[t]$ asks nobody, and no question of $[t]$ is raised.  Otherwise, the guards being simple, $[t]$ has one guard set, and enabledness gives that every member of it holds $[t]$ under $\sigma$ throughout the suffix.  A member $g$ answering a question $\alpha$ of $[t]$ in the suffix has $L_g$ on no attempt of $[t]$ and no other question of $[t]$ pending (\ic), and has not approved $\alpha$, so it holds $[t]$ by (\emph{A}): it wills its approval class on $[t]$, and the else-branch is not enabled.  So no answer no to a question of $[t]$ is recorded in the suffix.

Infinitely many attempts of $[t]$ are proposed.  Suppose not.  By Proposition~\ref{prop:commitment} each is eventually closed at each of its participants, so from some configuration $c_2$ on no lock is on an attempt of $[t]$ and no question of $[t]$ is pending, and an agent holds $[t]$ under $\sigma$ only by (\emph{A}), willing its approval class.  Let $g$ be the opener of $[t]$ if $[t]$ is unguarded, and otherwise a guard that holds $[t]$ under $\sigma$ in infinitely many configurations after $c_2$, one of the finitely many guards doing so.  In infinitely many configurations after $c_2$, then, the settled state of $g$ is in $\mathit{dom}(f^{[t]}_g)$ and $g$ holds its approval class on $[t]$ or is the opener of $[t]$; $\rho$ being unobstructed, $g$ is free for $[t]$ and clear in infinitely many of them, where its settled state is $s_g$ and it is ready for $[t]$; and $\rho$ being connected, $g$ knows every other participant of $[t]$ from some configuration on.  So a \textbf{Propose} of $g$ for an attempt of $[t]$ is enabled in infinitely many configurations, and $\rho$ being strongly fair, a \textbf{Propose} of $g$ for an attempt of $[t]$ is taken, or an attempt of $[t]$ is decided, in infinitely many, against the supposition either way.

Every attempt of $[t]$ proposed after $c_1$ is closed at its initiator (Proposition~\ref{prop:commitment}), not by a \textbf{Tally} of a vote $\mathsf{no}$, of which none is cast after $c_1$, nor by \textbf{Commit}, which follows its decision, against the hypothesis; so it is wounded, by a \textbf{Yield} or a \textbf{Wound}.  Infinitely many being proposed, by the finitely many guards of $[t]$ or its opener, some agent $g$ initiates infinitely many of them; each from its second on retries a wounded one, so all have one rank, and each is wounded in favour of an attempt with priority over it, of which there are finitely many, each the $\alpha$ of finitely many \textbf{Wound} transactions (\iv): so finitely many are wounded, a contradiction.  Hence $\sigma(\rho)$ is live, and with safety, correct.

\emph{Completeness}.  Let $r = c^0, c^1, \ldots$ be a correct run of $\calF(P)$ that is strongly fair.  We build a correct run $r'$ of $\calP$ with $\sigma(r') = r$, in stages, keeping between stages the invariant that the current configuration $e$ of $\calP$ is \emph{quiescent for} the current configuration $c$ of $r$: every agent knows every other; every outbox and every inbox is empty; every lock is $\bot$ and every attempt proposed is closed at each of its participants; $s_p = c^m_p$ for every $p$; and the volitional state of $p$ in $e$ holds exactly $p$'s approval classes on the members of $c^v_p$.  A quiescent $e$ has $\sigma(e)=c$: with no locks, the settled state of $p$ is $s_p$, and clause (\emph{A}) yields $c^v_p$; and every agent is clear in it, its inbox being empty.

\emph{Start.}  From the initial configuration, for each pair $p\ne q$ of $P$ in turn, a volition change of $p$ and one of $q$ add the class of Discover$(p,q)$, and the Discover is taken; each step is a stutter under $\sigma$, which reads neither $\mathit{known}$ nor the Discover classes, and the result is quiescent for $c^0$.

\emph{Step.}  Suppose $e$ is quiescent for $c^i$ and $c^i\rightarrow c^{i+1}$ is a transition of $r$.  If it is a volition change of $p$, take the volition change of $p$ in $\calP$ that replaces $p$'s approval classes by those on the members of $c^{i+1,v}_p$: it maps to $c^i\rightarrow c^{i+1}$, and the result is quiescent for $c^{i+1}$.  If it is the volitional machine transaction induced by $(t,Q')$, $t=d\rightarrow d'$ over $Q$, enabled in $c^i$: every $q\in Q$ has $s_q = d_q\in\mathit{dom}(f^{[t]}_q)$ and, being free, is ready for $[t]$, and clear; and every $g\in Q'$ holds $[t]$ in $c^i$, so it wills its approval class on $[t]$.  Let $G$ be a guard set of $[t]$ contained in $Q'$, the empty one if $[t]$ is unguarded; let $i$ be the opener of $[t]$ if $G=\emptyset$ and otherwise a member of $G$; and let $\alpha=(i,k,k,[t])$, $k$ one more than $i$'s number and its rank, no attempt being wounded in $r'$, whose guard set is $G$, the guard set of $[t]$ containing $i$.  Take, in this order: $i$'s \textbf{Propose} of $\alpha$; if $Q\ne\{i\}$, for each other participant $q$ in turn, the Communicate of $\mathsf{start}(\alpha)$ to $q$, $q$'s \textbf{Start}, which records the question where $q\in G$ and votes $\mathsf{yes}$ where not, and, where $q\in G$, the then-branch of its \textbf{Answer} and its \textbf{Vote}, voting $\mathsf{yes}$, $q$ being ready, then the Communicate of the vote and $i$'s \textbf{Tally} of it, the last of which decides $\alpha$; $i$'s \textbf{Commit}; a volition change of each participant $q\notin G$ that wills its approval class on $[t]$, removing it; and the Communicate of the outcome to each $q\ne i$, and $q$'s \textbf{Outcome}.  Each step but one is a stutter, by the safety cases: \textbf{Propose}, the answers and the votes before the last move their agents from clause (\emph{A}) to (\emph{B}), or within (\emph{B}), or leave them; the Communicates and the earlier \textbf{Tally} transactions change nothing $\sigma$ reads; the volition changes remove an approval class from an agent whose lock is on $\alpha$, decided, so that it holds $[t]$ by neither clause before or after; and the \textbf{Commit} and the \textbf{Outcome}s move a decided effect into $s$, the agents they release willing no approval class on $[t]$.  The one that is not is the \textbf{Tally} that decides $\alpha$, of the last participant's yes, or \textbf{Propose} if $Q=\{i\}$, which by the safety case maps to the volitional machine transaction induced by $d\rightarrow d'$ and $G$ from $\sigma(e)=c^i$, which is $c^i\rightarrow c^{i+1}$, the one induced by $(t,Q')$, $G$ being contained in $Q'$.  After the stage $s_q = c^{i+1,m}_q$ for every $q$, every lock is $\bot$, $\alpha$ is closed everywhere, and no agent holds an approval class on $[t]$, those of $G$ having been fulfilled and the others removed, as $c^{i+1}$ removes $[t]$ from every volitional state: the result is quiescent for $c^{i+1}$.

\emph{Idle.}  A guard $g$ that holds its approval class on a class $[u]$ not enabled in $c^{i+1}$ and is ready for $[u]$ has a \textbf{Propose} enabled in a quiescent configuration, and so has the opener $g$ of an unguarded class $[u]$ not enabled in $c^{i+1}$ that is ready for $[u]$; neither may stay enabled forever.  For each such pair $(g,[u])$, taken in a fair order over the countably many of them --- a pair enabled at infinitely many quiescent configurations being taken infinitely often --- and interleaved with the stages above, run an \emph{idle attempt}.  Take $g$'s \textbf{Propose} of an attempt $\alpha$ of $[u]$, and let $G=G_\alpha$.  As $[u]$ is not enabled in $c^{i+1}$, some participant $q\ne g$ votes $\mathsf{no}$ on $\alpha$ with no will fulfilled: a participant not ready for $[u]$, its machine state outside the domain of its effect, votes $\mathsf{no}$ by \textbf{Start}; and if every participant is ready, a member of $[u]$ is machine-enabled, by locality, and some member $q$ of $G$ does not will its approval class on $[u]$, $g$ willing its own, so $q$ is asked, and it records the question by \textbf{Start}, the answer no by the else-branch of \textbf{Answer}, and votes $\mathsf{no}$ by \textbf{Vote}.  Take the Communicate of $\mathsf{start}(\alpha)$ to $q$, its transactions up to its vote, and the Communicate of the vote; then, if $[u]$ is guarded, a volition change of $g$ adding its approval class on $[u]$ back; then $g$'s \textbf{Tally} of the vote, which aborts $\alpha$; and, for every participant other than $g$, the Communicate of the outcome and its \textbf{Outcome}, and, for every one other than $g$ and $q$, whose start is still undelivered, then the Communicate of $\mathsf{start}(\alpha)$ and its \textbf{Start}, which does nothing.  Every step is a stutter: $g$ holds $[u]$ under $\sigma$ throughout, by (\emph{B}) from its \textbf{Propose} and by (\emph{A}) as well once its approval class is back, which is before its \textbf{Tally} releases its lock; no vote $\mathsf{yes}$ is cast but $g$'s; and no other agent's holding changes.  The result is again quiescent for $c^{i+1}$.  A pair whose class is enabled in $c^{i+1}$ is not idled there.

\emph{Finite runs.}  If $r$ is finite, the stages after its last configuration consist of idle attempts.

The run $r'$ so built is safe, and $\sigma(r')$, with stutters removed, is $r$.  Each stage is finite and begins and ends at a quiescent configuration, where every agent is free for every class and clear; within a stage an agent's settled state changes only at the deciding vote, and its approval classes only by fulfilment, by the removals that follow that vote and by an idle attempt's restoring one; so an agent whose settled state is in the domain of its effect on a class and which holds its approval class on it, or opens it, at a configuration of a stage does so at the stage's first or last configuration.  A \textbf{Propose} enabled at a configuration of a stage is likewise enabled at its first or last configuration.  Every attempt is closed within its stage, every message sent is read within it, and every question raised is answered within it.  So $r'$ is connected, every agent knowing every other; unobstructed, its configurations of an agent at a state in the domain of its effect on a class and willing or opening it lying in infinitely many stages if they are infinitely many; and strongly fair and live: a Discover class is enabled only at the start; an Introduce never is, no introduction being sent; and a \textbf{Propose} enabled at infinitely many configurations is enabled at infinitely many quiescent ones, where it is taken by an idle attempt in its turn, at those where its class is not enabled in $r$, or its class is enabled in infinitely many configurations of $r$, which, being strongly fair, takes it in infinitely many, each a stage that decides an attempt of the class with the proposer as a participant, whose lock keeps its \textbf{Propose} disabled there.  Hence $r'$ is a correct run of $\calP$ under its fairness condition with $\sigma(r')=r$, and $\sigma$ is complete for strongly fair runs.
\end{proof}

\restateVolitionalSoundness*
\begin{proof}
We show first that in every configuration $c^k$ of the run (\ia) every writer a person holds occurs in that person's $H$ and nowhere else, and (\ib) every writer in the term of an assignment to a reader a person holds occurs in that term once and nowhere else.  In $c^1$ no person holds a variable.  Suppose both hold in $c^k$.  An Ask introduces a writer occurring nowhere in $c^k$ and gives its person one end of it, the other end going to the asked goal.  An Answer to a question removes it from the variables its person holds, adds an assignment to its reader, which the person does not hold, and introduces variables occurring nowhere in $c^k$: writers the person holds, their readers in the answer, and writers in the answer for the program to write, their readers held by the person.  A Communicate to a reader a person holds removes the assignment to it and gives that person the variables of its term and the assignment, shown, whose writers occurred there and nowhere else by (\ib).  A Dismiss removes a writer from the variables its person holds and adds nothing.  A Reduce at an agent $q$ renames its clause apart, and the run being proper, the variables of the renamed clause occur nowhere in $c^k$ or earlier, at any agent; its writer mgu $\hat\sigma$ binds writers of the goal and of the renamed head, none of them held by a person, and a writer in a term it assigns is a variable of the renamed clause or of the goal, occurring in that term once and nowhere else, by SRSW and by the SO of the goal, which the reduction removes.  Any other Communicate, and a Cold-call, move between goals a term already in $c^k$, which by (\ia) holds no writer a person holds.  Hence both hold in $c^{k+1}$, and (\ia) holds in $c^i$.

Now let $c^i \rightarrow c^{i+1}$ add $X? := t$ to a readers substitution.  An Ask, a Dismiss, a Communicate and a Cold-call add no assignment.  A Reduce adds $\hat\sigma?$, whose domain is the readers of the writers $\hat\sigma$ binds, writers of the goal and of the renamed head, and by (\ia) $X$ is neither.  So the transition is an Answer to a question $Z$, adding $Z? := t$; hence $Z = X$, and $X$ being among the questions of $p$'s person and of no other, it is the Answer of $p$ to $X$.
\end{proof}

\iffull\else
\fi

\restatePlatformCompilation*
\begin{proof}
Fix an effective local volition-guarded specification with simple guards, its canonical platform $\Theta$ and $P\subset\Pi$; write $TS'$ for the vmaGLP transition system over $P$ and $\lceil\Theta\rceil$, and $\sigma$ for the platform mapping.  $\sigma$ maps the initial configuration of $TS'$ to that of $\calP$: each agent's resolvent is the initial goal, so no \verb|hello| is sent, no message is on a stream, nothing is read and nothing is willed, and the agent goal, once called, holds $(s_0,0,\bot,\emptyset)$.

Four facts about the configurations of a safe run of $TS'$.  First, a request, a message, an answer and a tick each pass from where they are written --- the person's answer, the network input stream, the card's answer, the clock goal --- through the merge to the head of the agent goal's stream, where the clause of their kind reads them and computes in its body the condition of the transaction they put up, on the agent state the clause was called with; the agent goal is one goal, so no dispatch intervenes before the dispatch procedure's reduction, and the app state and the starts held waiting do not change between the reading of a trigger and its dispatch.  Second, a peer an agent has sent a hello to and received one from is one with which it has exchanged hellos in $\sigma$, a hello received having been sent, and not conversely until the peer's hello arrives: a condition's test that a peer is known, true at the agent, is true in $\sigma$.  Third, an agent reads the start of an attempt only after it has sent its sender a hello and read the sender's: the sender sends only to a peer it has received a hello from and sent its own to, which precedes the start on its output stream, and the messages of one stream are delivered and read in order; so an agent knows the sender of every start it reads, and, knowing being never undone, every agent \textbf{Start}, \textbf{Vote} and \textbf{Tally} send to.  Fourth, an agent sends a cargo to a recipient at most once, by fact (\ia) of the proofs for Section~\ref{sec:decomposition} and its record of yields, so no two messages on the streams are one message of $\calP$.

\emph{Safety}, per transition of $TS'$.  An Ask, a Dismiss, a Communicate to a reader the person holds or of an assignment along a stream, a Reduce of the clock goal, a merge goal, the agent's guarded clause, a card's clause, the goal a card hands its answer to, or a goal computing a condition or an effect, and a Cold-call, an app message it delivers staying in $o_p$ until its recipient's agent goal reads it, change nothing $\sigma$ reads: stutters.  An Answer of $p$ to its request question is a stutter, the request being unread, except with \verb|known(|$q$\verb|)|, $q\in P$: then it adds $[\mathrm{Discover}(p,q)]$, a volition change, unless $q\in\mathit{known}_p$ or the class is already there.  An Answer to a card is a stutter, the answer being unread.

A Reduce of the agent goal's clause reading a request \verb|propose(|$[t]$\verb|)|, or the answer yes to the card of an attempt of $[t]$, adds $p$'s approval class on $[t]$ to its volitional state where the condition of the transaction it puts up, but for its test of known peers, holds at $p$'s machine state, a volition change of $p$, the agent goal reading one trigger at a time; otherwise it is a stutter.  Reading an app message $m$ that $q$ sent, it is Communicate$(q,p,m)$, $m$ leaving $o_q$ and entering $i_p$; reading a hello, an answer no or a tick, it is a stutter.

A Reduce of a dispatch clause taking a transaction is that transaction of $\calP$ at the configuration $\sigma$ reads: its condition holds there, by the first fact on the app state and the inbox and by the second and the third on the peers it requires known.  \textbf{Propose} of a guarded class and the then-branch of \textbf{Answer} are guarded by $p$, and $p$'s approval class on $[t]$ is in its volitional state, read with its condition holding, so the transaction is enabled, and the reduction removes the class, fulfilled (Definition~\ref{def:vmat}).  The else-branch of \textbf{Answer}, taken on the answer no, is enabled, $p$ holding no approval class while the agent goal reads an answer no.  \textbf{Start}, \textbf{Tally}, \textbf{Yield} and \textbf{Outcome} are unguarded and take a message in $i_p$, which they remove; an opener's \textbf{Propose}, \textbf{Vote}, \textbf{Commit} and \textbf{Wound}, taken on a tick, are unguarded, \textbf{Wound} leaving the start it reads in $i_p$.  The new app state and the messages the effect yields, addressed to peers known at the agent and so in $\sigma$, are the transaction's effect.  A Reduce of a dispatch clause holding a start waiting is a stutter, the start staying in $i_p$.  A Reduce of a dropping clause on a request or an answer yes is a volition change of $p$ removing its approval class on $[t]$ where the condition held but for its test of known peers, and a stutter otherwise; on an answer no or a tick it is a stutter.  A Reduce of the request clause on \verb|known(|$q$\verb|)| sending \verb|hello| is a stutter where $q$ has not sent \verb|hello| to $p$ or $q\in\mathit{known}_p$ already, and otherwise Discover$(p,q)$: both will it, $p$ having requested \verb|known(|$q$\verb|)| and $q$ \verb|known(|$p$\verb|)| with $q\notin\mathit{known}_p$, and after it $\mathit{known}_p$ and $\mathit{known}_q$ each gain the other and both volitional states lose the class.  Every safe run of $TS'$ therefore maps to a safe run of $\calP$.

\emph{Liveness}.  Let $r'$ be a proper correct run of $TS'$ in which no card's question stays open and whose image $r=\sigma(r')$ is connected and unobstructed.  In $r'$ every element written to a stream is read by the agent goal's clause of its kind: the Communicate of its assignment and the reductions of the merge goals are enabled until taken and the merge is fair, and once the element heads the agent goal's stream the clause of its kind reduces, by Monotonicity and the clause order as for the requests of Section~\ref{sec:vglp}; its condition is then bound, being computable and terminating, and the dispatch clause reduces.  The clock goal's reduction is enabled until taken and spawns the next, so ticks are written without end, and each is read.  Suppose a class of $\calP$ is enabled in every configuration of a suffix of $r$ and no member of it occurs there.

Communicate$(p,q,m)$: $m$ is on $p$'s output stream, the Cold-calls of the messages before it and then of $m$ are each enabled until taken, so $m$ is delivered and $q$'s agent goal reads it.  The class of Discover$(p,q)$: no introduction being sent, it is enabled only through Discover, and then both will it, so $p$ has requested \verb|known(|$q$\verb|)| and $q$ \verb|known(|$p$\verb|)|, and each request is read and its \verb|hello| sent, if not before, $q$ not being known at $p$ nor $p$ at $q$; the second sending is a Discover$(p,q)$.  \textbf{Tally}, \textbf{Yield} or \textbf{Outcome} on a message $m\in i_p$: $m$ is read and not yet dispatched, and its dispatch takes it.  \textbf{Start} on $m$: read and not yet dispatched, the dispatch takes it, $m$ not waiting; held waiting, its waiting held at the state at which it was last put up and does not hold now, so a dispatch has since changed the app state or consumed a start, after which it was put up again and taken.

An opener's \textbf{Propose}, a \textbf{Vote} on a question pending at $p$, a \textbf{Commit}, or a \textbf{Wound}: at each tick, a configuration of the suffix, an instance of the class has its condition holding, where it tests a known peer once the peer's \verb|hello| has arrived, which it does, the peer being known in $r$.  Each tick takes an instance of the class that has waited longest; a class taken waits from then on less long than this one, and the classes that have waited longer than it are finitely many --- those taken, or found holding by a tick, before it began to wait, finitely many at each tick, the unguarded classes an agent opens at a state, the questions pending at it, the attempts it initiated and the starts it holds being finitely many --- so within finitely many ticks the dispatch takes an instance of it.

$p$'s approval class on $[t]$, whose members are $p$'s \textbf{Propose} of $[t]$ and its \textbf{Answer} on the questions of $[t]$.  If \textbf{Propose} is enabled in infinitely many configurations of the suffix, $p$ holds the class in them, so the agent goal reads, in infinitely many, a request \verb|propose(|$[t]$\verb|)| whose condition holds but for its test of known peers, each dispatched within finitely many steps; each dispatch takes \textbf{Propose} unless a peer its condition tests is known in $r$ and not yet at the agent, which, each hello arriving, happens finitely often.  Otherwise, from some configuration on, \textbf{Answer} is enabled throughout, on a question $\alpha$ of $[t]$ pending at $p$ and unanswered, at most one being pending, by fact (\ic) of the proofs for Section~\ref{sec:decomposition}; $\alpha$ is never closed at $p$, the configuration after its closing enabling neither \textbf{Propose} of $[t]$, by the supposition, nor \textbf{Answer}, no \textbf{Start} having yet raised another question.  So the card \textbf{Start} raised on $\alpha$ is neither answered and read nor closed: it asks its question, which is open until answered; it is answered, no card's question staying open; the answer is read; and the dispatch takes the then-branch on yes and the else-branch on no.

In every case a member of the class occurs in the suffix: $r$ is live.

\emph{Fairness.}  $r$ is connected and unobstructed by hypothesis, and strongly fair.  A \textbf{Propose} of $p$ of a guarded class $[t]$ is enabled in $r$ only while $p$ holds its approval class on $[t]$, which by the safety case is while the agent goal reads a request \verb|propose(|$[t]$\verb|)| whose condition, but for its test of known peers, holds --- an answer yes read with its condition holding is to a card of $[t]$, whose pending question disables \textbf{Propose} of $[t]$ --- and the dispatch of the request takes \textbf{Propose}, unless a peer its condition tests is known in $r$ and not yet at the agent, which happens finitely often.  An opener's \textbf{Propose} has a condition whose value at the agent changes only as the agent goal reads a trigger, after which it takes a tick of its own; so if it is enabled in infinitely many configurations of $r$, its condition holds at infinitely many ticks, and its class is taken in infinitely many, each tick at which it holds taking it or one of the finitely many classes that have waited longer.  So $r$ is a correct run of $\calP$ under its fairness condition.
\end{proof}

\restateCompleteness*
\begin{proof}
A construct is presented at $p$ only for an open question $X$ (Definition~\ref{def:manifest}).  A submission $t$ it forms is a term of the type of $X$: a widget forms only terms of its type, the forming procedure of each construct process being typed (Section~\ref{sec:compilation}).  The person's input fills the positions the person writes, but for the stream and list positions, at which the widget places new readers or, for a list the person ends, \verb|[]|; and the positions the program writes hold new writers, whose readers the construct holds.  So $t$ is ground at the positions the person writes but for the stream and list positions, which is an answer the Answer transition allows for $X$ (Definition~\ref{def:vglp-answer}), and $X$ is open.  Conversely, if the widget of the construct of $X$ is complete, every answer the Answer transition allows for $X$ is formed by some input: it is offered.
\end{proof}

The mapping for Theorem~\ref{thm:compilation} rests on what the dispatcher and the construct processes do.  Here is the dispatcher, as the repository holds it in \texttt{programs/vglp/dispatcher.glp}, which the compilation emits into every compiled program under fresh names.  \verb|dispatch| reads the ask stream and the person channel; \verb|fair| merges the asks, the person's grants and the constructs' draws into one stream, taking them in turn, and \verb|serve| reads that stream and numbers the questions from \verb|1|.  A program that declares \verb|Output| is dispatched by the second \verb|dispatch|, with the draws of construct \verb|0|, the thread on which the compiled program draws its output to its person; \verb|merge_draws| merges them with the constructs' draws, and the merged stream never ends.  On an ask, \verb|serve| writes on its \verb|Spawn| stream the question with a fresh identifier \verb|N|, the reader of a new grant stream and the writer of a new draw stream, and the compiled program's \verb|constructs/1| reads that stream and calls the construct process of the question's type; it routes each grant \verb|input(Id, P, R)| to the construct with identifier \verb|Id|; it passes the constructs' draws on; and it moves a withdrawn construct's grant writer to \verb|Deads|, which it never reads, since the person channel never closes.
{\footnotesize
\begin{verbatim}
Inputs(T) ::= [Input(T) | Inputs(T)].
Route(T) ::= route(Integer, Inputs(T)?).
Routes(T) ::= [] ; [Route(T) | Routes(T)].
Deads(T) ::= [] ; [Inputs(T)? | Deads(T)].
Ask(Q) ::= ask(Constant, Q).
Spawn(Q, T) ::= spawn(Integer, Q, Inputs(T), Stream(Draw(T))?).

exported procedure(Q, T) dispatch(Stream(Ask(Q))?, PersonChannel(T)?,
    Channel(Stream(_), Stream(_)), Stream(Spawn(Q, T))).
dispatch(Asks, ch(PIn, POut?), ch(MIn?, _), Ss?) :-
    split_grants(PIn?, Inputs, MIn),
    fair(ask, Asks?, Inputs?, COut?, Es),
    serve(Es?, COut, POut, [], [], 1, Ss).

exported procedure(Q, T) dispatch(Stream(Ask(Q))?, PersonChannel(T)?,
    Stream(_), Stream(Draw(T))?, Stream(Spawn(Q, T))).
dispatch(Asks, ch(PIn, POut?), MIn?, ODs, Ss?) :-
    split_grants(PIn?, Inputs, MIn),
    merge_draws(ODs?, DOut?, POut),
    fair(ask, Asks?, Inputs?, COut?, Es),
    serve(Es?, COut, DOut, [], [], 1, Ss).

procedure(T) merge_draws(Stream(Draw(T))?, PersonOut(T)?, PersonOut(T)).
merge_draws([X | Xs], Ys, [X? | Zs?]) :- merge_draws_(Ys?, Xs?, Zs).
merge_draws(Xs, [Y | Ys], [Y? | Zs?]) :- merge_draws(Xs?, Ys?, Zs).
merge_draws([], Ys, Ys?).

procedure(T) merge_draws_(PersonOut(T)?, Stream(Draw(T))?, PersonOut(T)).
merge_draws_([X | Xs], Ys, [X? | Zs?]) :- merge_draws(Ys?, Xs?, Zs).
merge_draws_(Xs, [Y | Ys], [Y? | Zs?]) :- merge_draws_(Xs?, Ys?, Zs).
merge_draws_(Xs, [], Xs?).

procedure(T) split_grants(PersonIn(T)?, Inputs(T), Stream(_)).
split_grants([input(Id, P, R) | Gs], [input(Id?, P?, R?) | Is?], Ms?) :-
    split_grants(Gs?, Is, Ms).

Event(Q, T) ::= Ask(Q) ; Input(T) ; Draw(T).
Events(Q, T) ::= [Event(Q, T) | Events(Q, T)].
Turn ::= ask ; grant ; draw.

procedure(Q, T) fair(Turn?, Stream(Ask(Q))?, Inputs(T)?, Stream(Draw(T))?,
    Events(Q, T)).
fair(ask, [A | As], Is, Ds, Es?) :- Es = [A? | Es1?], fair(grant, As?, Is?, Ds?, Es1).
fair(grant, As, [I | Is], Ds, Es?) :- Es = [I? | Es1?], fair(draw, As?, Is?, Ds?, Es1).
fair(draw, As, Is, [D | Ds], Es?) :- Es = [D? | Es1?], fair(ask, As?, Is?, Ds?, Es1).
fair(ask, As, [I | Is], Ds, Es?) :- Es = [I? | Es1?], fair(draw, As?, Is?, Ds?, Es1).
fair(grant, As, Is, [D | Ds], Es?) :- Es = [D? | Es1?], fair(ask, As?, Is?, Ds?, Es1).
fair(draw, [A | As], Is, Ds, Es?) :- Es = [A? | Es1?], fair(grant, As?, Is?, Ds?, Es1).
fair(ask, As, Is, [D | Ds], Es?) :- Es = [D? | Es1?], fair(ask, As?, Is?, Ds?, Es1).
fair(grant, [A | As], Is, Ds, Es?) :- Es = [A? | Es1?], fair(grant, As?, Is?, Ds?, Es1).
fair(draw, As, [I | Is], Ds, Es?) :- Es = [I? | Es1?], fair(draw, As?, Is?, Ds?, Es1).

procedure(Q, T) serve(Events(Q, T)?, Stream(Draw(T)), PersonOut(T),
    Routes(T)?, Deads(T)?, Integer?, Stream(Spawn(Q, T))).
serve([ask(_, Q) | Es], CInto?, DOut?, Rs, Xs, N,
      [spawn(N?, Q?, Gs?, Ds) | Ss?]) :-
    ground(N?) |
    merge(Ds?, CInto1?, CInto),
    N1 := N? + 1,
    serve(Es?, CInto1, DOut, [route(N?, Gs) | Rs?], Xs?, N1?, Ss).
serve([input(Id, P, R) | Es], CInto?, DOut?, Rs, Xs, N, Ss?) :-
    route_grant(input(Id?, P?, R?), Rs?, Rs1),
    serve(Es?, CInto, DOut, Rs1?, Xs?, N?, Ss).
serve([draw(Id, W, V) | Es], CInto?, [draw(Id?, W?, V?) | DOut?], Rs, Xs, N,
      Ss?) :-
    serve(Es?, CInto, DOut, Rs?, Xs?, N?, Ss).
serve([withdraw(Id) | Es], CInto?, [withdraw(Id?) | DOut?], Rs, Xs, N, Ss?) :-
    ground(Id?) |
    remove_route(Id?, Rs?, Rs1, Xs?, Xs1),
    serve(Es?, CInto, DOut, Rs1?, Xs1?, N?, Ss).

procedure(T) route_grant(Input(T)?, Routes(T)?, Routes(T)).
route_grant(input(Id, P, R), [route(Id1, [input(Id?, P?, R?) | Gs1?]) | Rs],
            [route(Id1?, Gs1) | Rs?]) :-
    Id? =?= Id1? | true.
route_grant(G, [route(Id1, Gs?) | Rs], [route(Id1?, Gs) | Rs1?]) :-
    otherwise |
    route_grant(G?, Rs?, Rs1).
route_grant(_, [], []).

procedure(T) remove_route(Integer?, Routes(T)?, Routes(T), Deads(T)?,
    Deads(T)).
remove_route(Id, [route(Id1, Gs?) | Rs], Rs?, Xs, [Gs | Xs?]) :-
    Id? =?= Id1? | true.
remove_route(Id, [route(Id1, Gs?) | Rs], [route(Id1?, Gs) | Rs1?], Xs,
             Xs1?) :-
    otherwise |
    remove_route(Id?, Rs?, Rs1, Xs?, Xs1).
remove_route(_, [], [], Xs, Xs?).
\end{verbatim}
}
The dispatcher binds only variables of its own streams.

\mypara{The construct process}  The \emph{construct process} of an ask \verb|ask(|$T$\verb|, |$t$\verb|(|$Y$\verb|))| with identifier $\mathit{Id}$ is a goal holding $Y$, the person's end of the interactive variable --- its writer where $T$ is in reader mode, and its reader where $T$ is in writer mode --- a stream of grants and a stream of draws.  It holds a set of \emph{questions}, each at its path from the root of the interactive variable's term: initially $Y$ at the root where $T$ is in reader mode, and none otherwise; as the program's output arrives on a reader it holds, each writer inside it at a position the person writes is added to its questions.

On a grant \verb|input(|$\mathit{Id}$\verb|, |$\mathit{Pos}$\verb|, |$R$\verb|)| for a question $W$ it holds at path $\mathit{Pos}$, of type $T'$, the construct process forms $R$ by the forming procedure of $T'$, which yields \verb|formed(|$u$\verb|)| or \verb|refused|; $u$ is a term of type $T'$ with $R$ at the positions the person writes and with fresh variables where Answer has them (Definition~\ref{def:vglp-ts}).  On \verb|formed(|$u$\verb|)| it binds $W$ to $u$, removes $W$ from its questions and adds to them the writers of the new variables at the stream and list positions of $u$, and it holds the readers of the new variables at the positions the program writes.  On \verb|refused|, and on a grant for a path at which it holds no question, it binds nothing.  It binds no variable but its questions, and those only on a grant.

A construct process has the bridge draw the widget of $T$ with the output so far, again as more arrives, each widget position the program has not yet written whole --- the root, a field, an element of any stream, the thread's included --- drawn as pending, and a stream's unwritten tail not drawn, so that the questions beside it are offered at once, and withdraws once it holds no question and no reader at a position whose output may hold one.  It draws \verb|draw(|$\mathit{Id}$\verb|, |$W$\verb|, |$V$\verb|)|, $W$ the widget of $T$ and $V$ the output so far, with \verb|input(|$\mathit{Pos}$\verb|)| at each question and \verb|pending(|$\mathit{Pos}$\verb|)| at each widget position not yet written whole, $\mathit{Pos}$ the position's path, and with $V$ under the functor $t$ of its ask; the output thread's value is under the functor \verb|output|.

So the values a program draws are of one type, the union of these, and the grants' raw inputs are of the types of its questions; the compilation forms the union of both, the parameter \verb|T| of the program's person channel (Definition~\ref{def:person-ts}).

The mapping for Theorem~\ref{thm:compilation} reads the source configuration off the compiled goals and the construct processes.  An asking clause passes each argument of the goal it reduces to the goal it calls through a variable pair of its own, $S_l$ and $S_l?$, and the mapping reads the pair as the argument: at an input position the writer mgu binds $S_l$ to the goal's term $u$ and the call receives $S_l?$, with $S_l? := u$ pending, and at an output position it binds the goal's writer $W$ to $S_l?$ and the call receives $S_l$, with $W? := S_l?$ pending.

\begin{definition}[Compilation Mapping: Goals]\label{def:compilation-sigma}
Given a configuration of GLP with persons over $P$ and $\lceil M\rceil$, the mapping $\sigma$ yields the vmaGLP configuration over $P$ and $M$ with, at each agent, the goals $G$ of this definition and the held variables $H$ and the readers substitution $\sigma_r$ of the next two: $G$ is the multiset of the agent's goals, each read back to its source with its ask stream dropped: $q(S_1,\ldots,S_n,D)$ as the $n$-ary goal $q(S_1,\ldots,S_n)$, $q_1(S_1,\ldots,S_n,X?,D)$ and $q_1(S_1,\ldots,S_n,X,D)$ as the asked goal, with the identifier of the goal its asking clause reduced, and any other compiled goal as itself; the construct processes, the dispatcher and the \verb|merge| goals the compilation adds contributing nothing; and in every goal an asking clause's pair for an argument read as that argument: at an input position $S_l?$ as the term $u$, until the assignment is communicated, and at an output position $S_l$ as $W$ and $S_l?$ as $W?$.
\end{definition}

\begin{definition}[Compilation Mapping: Held Variables]\label{def:compilation-sigma-held}
$H$ is the set of variables held for the person, by the construct processes and by the asks the dispatcher has not yet taken, each with its type: the writers of the positions the person writes and the readers of the positions the program writes; and the assignments communicated to the construct processes, shown.  A variable a construct process holds from binding a closed question, or from binding such a variable in turn, is not in $H$.
\end{definition}

\begin{definition}[Compilation Mapping: Substitution]\label{def:compilation-sigma-subst}
$\sigma_r$ is the restriction of the accumulated readers substitution to the readers of \temph{source variables} --- the variables of the initial goal, of the compiled clauses of $M$ as renamed, but for their ask streams, the interactive variables, and the variables of the terms construct processes bind questions with --- less the assignments the asking clauses' reductions make, but for the assignment $W? := S_l?$ of an output position, which it reads through --- absent while $S_l?$ is unassigned, and $W? := u$ once $S_l? := u$ is --- and less the assignments by which construct processes bind closed questions or variables that Definition~\ref{def:compilation-sigma-held} leaves out of $H$.
\end{definition}

\restateCompilation*
\begin{proof}
Fix a vGLP program $M$ and $P\subset\Pi$, and let $\sigma$ be the mapping of Definition~\ref{def:compilation-sigma}.  $\sigma$ maps the initial configuration to the initial vmaGLP configuration: the goals read back to the initial goal of $M$, and there is no construct process.

\emph{Safety}, per transition of GLP with persons over $\lceil M\rceil$.  A grant, a reduction of the dispatcher or of a \verb|merge| goal the compilation adds, a reduction of a construct process that binds nothing, and a Communicate on the person channel, on an ask stream, or between the dispatcher and the construct processes change nothing $\sigma$ reads: stutters.  A construct process binding a question whose reader occurs nowhere, on the person's input to a construct left on screen, removes it from $H$ and adds nothing $\sigma$ reads, the variables of its term being left out of $H$ and its assignment out of $\sigma_r$: the Dismiss of it by $p$.  A construct process binding a variable left out of $H$ changes nothing $\sigma$ reads: a stutter.  A reduction of the asking clause of $q$ replaces $q(S_1,\ldots,S_n,D)$ by the goal $q_1(S'_1,\ldots,S'_n,X?,D')$ or $q_1(S'_1,\ldots,S'_n,X,D')$, holding one end of a fresh interactive variable of the interactive type of $q$, which $\sigma$ reads, through the pairs of its arguments and with the identifier of the goal reduced, as $q(S_1,\ldots,S_n,X?)$ or $q(S_1,\ldots,S_n,X)$, the assignments the reduction makes being outside $\sigma_r$, and writes on its ask stream an ask holding the other end: the Ask of the source goal.  A Communicate of an assignment an asking clause's reduction made replaces a pair by what $\sigma$ already read it as: a stutter.  A construct process binding an open question $X$ with the term $t$ its person's input forms, a term of the type of $X$ ground at the positions the person writes but for the stream and list positions, which hold readers of new writers it holds or, for a list the person ends, \verb|[]|, with new writers at the positions the program writes, whose readers it holds: the Answer of $p$ to $X$ with $t$.  A reduction of a compiled clause of $M$: its guard and body are those of the source clause, with calls of volitional procedures read back, ask streams dropped and the \verb|merge| goals the compilation adds contributing nothing; its added head argument, the ask stream, matches the goal's writer whatever it is; and the earlier compiled clauses of the procedure, the source's in the same order, fail or suspend on the goal: a clause that reads an argument whose pair's assignment, made by the asking clause, is not yet communicated suspends on it and is followed by no clause that does not read it (Definition~\ref{def:volitional-clause}), so neither the reducing clause nor an earlier one reads such an argument, and each earlier clause fails or suspends on the source goal too; so it maps to the Reduce of the source goal with the source clause.  A Communicate, within the agent or from another, maps to the source Communicate of the same assignment, to the person where its reader is held for the person, and a Cold-call is carried unchanged.

Every safe run therefore maps to a safe run.  \emph{Correctness}: let $\rho$ be a proper correct run of GLP with persons over $\lceil M\rceil$ and suppose a class in the domain of the equivalence is enabled in every configuration of a suffix of $\sigma(\rho)$ with no member taken.  An Ask class of an $n$-ary goal corresponds to its goal $q(S_1,\ldots,S_n,D)$, whose asking clause succeeds throughout, the goal holding a writer at each output position (Definition~\ref{def:volitional-clause}); a Reduce class of an asked goal or of an ordinary goal corresponds to the compiled clause, which succeeds throughout by the correspondence of the safety case once the Communicates of the assignments its asking clause made, each enabled until taken, are taken; a Communicate class of an assignment to a reader the person holds corresponds to the Communicate to the construct process, and any other Communicate or Cold-call class to itself.  Each is enabled in the corresponding suffix of $\rho$, whose liveness takes it, and its image is a member.  Hence $\sigma(\rho)$ is correct.

\emph{Completeness}, for a program whose interactive types have complete widgets.  Let $r$ be a correct run of the vmaGLP transition system over $P$ and $M$.  We build a proper correct run $r'$ of GLP with persons over $\lceil M\rceil$ with $\sigma(r') = r$, taking the transitions of $r$ in turn.  For an Ask, reduce the goal's asking clause, let the \verb|merge| goals carry its ask to the dispatcher, and let the dispatcher spawn its construct process, stutters but the first.  For a Dismiss of $X$, grant an input on its construct, which is on screen, and let the construct process bind $X$, stutters but the last.  For an Answer to $X$ with $t$, the widget being complete, grant the input that forms $t$, let the dispatcher route it, and let the construct process bind $X$ with $t$, stutters but the last.  For a Reduce, take the compiled clause.  For a Communicate to a reader the person holds, take the Communicate to the construct process and let it have the bridge draw the output, stutters but the first.  Any other Communicate, and a Cold-call, of $r$ is taken as it is, followed by the Communicates on the channels of the dispatcher that it enables.  The run $r'$ is proper, its variables fresh, and live: a construct process awaiting input is suspended, the dispatcher and the \verb|merge| goals are reduced in the stage that enables them, and every other class enabled from some configuration on is of a source class that $r$ takes, $r$ being correct.  Hence $\sigma$ is complete.
\end{proof}

\restateCascade*
\begin{proof}
Fix an effective local volition-guarded specification with simple guards and $P\subset\Pi$, and let $\rho$ be a correct run of GLP with persons over $P$ and $\lceil\lceil\Theta\rceil\rceil$ under the fairness condition.  By Theorem~\ref{thm:compilation}, $\sigma_3(\rho)$ is a correct, proper run of the vmaGLP transition system over $P$ and $\lceil\Theta\rceil$, in which no card's question stays open; by Theorem~\ref{thm:platform-compilation}, $\sigma_2(\sigma_3(\rho))$, being connected and unobstructed, is a correct run of $\calP$ under its fairness condition; and by Theorem~\ref{thm:decomposition}, $\sigma_1(\sigma_2(\sigma_3(\rho)))$ is a correct run of $\calF(P)$.  $\sigma_1\circ\sigma_2\circ\sigma_3$ takes the initial configuration to the initial configuration, as each of its parts does, so it is a correct implementation.
\end{proof}

\end{document}